\newif\iflncs
\documentclass[smallextended]{svjour3}
\usepackage{fullpage}

\usepackage{appendix}
\usepackage{amsmath}
\usepackage{amssymb}
\usepackage{ifpdf}
\usepackage{algorithmic}
\usepackage{graphicx}
\usepackage{subfig}
\usepackage{wrapfig}
\usepackage{cite}
\usepackage{colortbl}
\usepackage{tikz}
\usepackage{morefloats}
\usepackage{enumitem}
\usepackage{xcolor}

\usepackage{commands-tam}

\newif\iftodo

\newif\ifhighlight
\newcommand{\update}[1]{\ifhighlight\colorlet{colorsaved}{.}\color{blue} #1 \color{colorsaved} \else #1 \fi}

\ifpdf

  \usepackage[pdftex]{epsfig}
  \usepackage[pdftex]{hyperref}

\else

    \usepackage[dvips]{epsfig}
    \newcommand{\href}[2]{#2}

\fi
\newif\ifinproceedings

\newif\ifabstract
\newif\iffull

\newtheorem{observation}[theorem]{Observation}

\newcommand{\cgSE}{CG_{SE}}

\newcommand{\cgNE}{CG_{NE}}
\newcommand{\cgNW}{CG_{NW}}
\newcommand{\cgSW}{CG_{SW}}

\newcommand{\necav}{conc_{NE}}
\newcommand{\secav}{conc_{SE}}
\newcommand{\swcav}{conc_{SW}}
\newcommand{\nwcav}{conc_{NW}}

\newcommand{\necon}{conv_{NE}}
\newcommand{\secon}{conv_{SE}}
\newcommand{\swcon}{conv_{SW}}
\newcommand{\nwcon}{conv_{NW}}

\newcommand{\on}{\mathtt{on}}
\newcommand{\off}{\mathtt{off}}
\newcommand{\latent}{\mathtt{latent}}

\newcommand{\frametiles}{T_{frame}}

\newcommand{\cdtiles}{T_{CD}}
\newcommand{\alphatiles}{T_{\alpha}}

\title{Replication of arbitrary hole-free shapes via self-assembly with signal-passing tiles}

\author{
    Andrew Alseth
        \thanks{Department of Computer Science and Computer Engineering, University of Arkansas, Fayetteville, AR, USA.
        \protect\url{awalseth@uark.edu}.
        This author's work was supported in part by National Science Foundation grant CAREER-1553166}
\and
    Jacob Hendricks
        \thanks{Department of Computer, Information, and Data Sciences, University of Wisconsin -- River Falls, River Falls, WI, USA.
        \protect\url{jacob.hendricks@uwrf.edu}.
        }
\and
	Matthew J. Patitz
        \thanks{Department of Computer Science and Computer Engineering, University of Arkansas, Fayetteville, AR, USA.
        \protect\url{mpatitz@self-assembly.net}.
        This author's research was supported in part by National Science Foundation grants CCF-1422152 and CAREER-1553166}
\and
    Trent A. Rogers
        \thanks{Hamilton Institute and Department of Computer Science, Maynooth University, Ireland.  
        \protect\url{trent.rogers@mu.ie}.
        Research supported by European Research Council (ERC) under the European Union’s Horizon 2020 research and innovation programme (grant agreement No 772766, Active-DNA project), and Science Foundation Ireland (SFI) under Grant number 15/ERCS/5746.}
}
\institute{}
\date{}

\begin{document}

\maketitle

\begin{abstract}
In this paper, we investigate the abilities of systems of self-assembling tiles which can each pass a constant number of signals to their immediate neighbors to create replicas of input shapes.  Namely, we work within the Signal-passing Tile Assembly Model (STAM), and we provide a universal STAM tile set which is capable of creating unbounded numbers of assemblies of shapes identical to those of input assemblies.  The shapes of the input assemblies can be arbitrary 2-dimensional hole-free shapes.  This improves previous shape replication results in self-assembly that required models in which multiple assembly stages and/or bins were required, and the shapes which could be replicated were more constrained, as well as a previous version of this result that required input shapes to be represented at scale factor 2.



\vspace{-12pt}
\end{abstract}


\section{Introduction}

    
    
    
    
    
            
        
    
    
    
    

The process of self-assembly, in which a disorganized collection of relatively simple components autonomously combine to form more complex structures, occurs in many natural systems (e.g. the formation of crystals such as snowflakes, or a variety of cellular components). With a desire to harness the power of self-assembly to build systems with molecular precision while creating complex structures via rational design, research has often focused on the subfield of \emph{algorithmic self-assembly}, which occurs when systems and components are designed so that the combination of the components is forced to follow the steps of a designated algorithm. The algorithms can be designed to control the shapes and dimensions of assemblies, potentially with high complexity and precision. Algorithmic self-assembly was initially shown to have great theoretical power in the abstract Tile Assembly Model (aTAM) \cite{Winf98} with systems being capable of universal computation, and since then many theoretical (e.g. \cite{SolWin07,IUSA,RotWin00,jCCSA}) and even experimental results (e.g. \cite{evans2014crystals,drmaurdsa,BarSchRotWin09}) have continued to explore the possibilities.

Most work in algorithmic self-assembly focuses on systems with input seed assemblies (often only a single tile) which then grow into a target structure. The goal is often to do so while using the smallest number of unique types of tiles, or optimizing some other system parameter. However, a variety of other types of behaviors have also been studied, including performing series of complex computations \cite{jCCSA,jSADS}, the identification of target shapes from a set input assemblies of multiple shapes \cite{ShapeIdentAlgo}, and the replication of input patterns \cite{STAMPatternRep} (theoretical) and \cite{SchulYurWinfEvolution} (experimental) or shapes \cite{RNaseSODA2010,SelfReplicationDNA}. In this paper, we focus on the latter goal. Specifically, we present a single tile set that is capable of replicating the shapes of any two-dimensional hole-free shapes which are given as input assemblies. Our construction works within the Signal-passing Tile Assembly Model (STAM) \cite{jSignals}, in which the binding of a tile's glue is able to initiate a ``signal'' that causes one or more other glues on the tile to either turn ``on'' or ``off''. Each signal can be activated only once, and is completely asynchronous, meaning that the time taken for the activation or deactivation of a glue may be arbitrarily long (although not infinite). The tile set that we present is robust to both the asynchronous nature of the signals, as well as to an input consisting of assemblies of multiple shapes. As long as each input assembly has a single generic glue type on every surface of its perimeter, an infinite number of assemblies will be produced with the exact shape of each. Our main result, Theorem \ref{thm:exp_rep}, has advantages over prior replication results in tile assembly. It works for all two-dimensional hole-free shapes and requires no scale factor (a previous version of this result \cite{STAMshapes} required a scale factor of 2), while \cite{SchulYurWinfEvolution} replicates one-dimensional patterns, \cite{STAMPatternRep} replicates two-dimensional patterns on rectangular assemblies, and \cite{RNaseSODA2010} replicates a much more constrained set of shapes. Additionally, our construction uses a universal tile set and requires only a single stage where \cite{RNaseSODA2010} requires differing tile sets and multiple stages, and while \cite{SelfReplicationDNA} is capable of replicating three-dimensional shapes, the model a more complex extension to the STAM and the constructions make use of tiles of multiple shapes as well as glues which can be flexible and allow for the reconfiguration of assemblies, neither of which is required by our construction in the basic STAM (which itself has experimental motivation \cite{SignalTilesExperimental}). Thus, we present a single, universal shape replicating tile set in the STAM such that a system using its tiles will perform the parallel and exponentially increasing, unbounded replication of all sets of two-dimensional hole-free input shapes.

As a supporting result, in Lemma \ref{lem:rectangle-layer} we also prove that our construction is capable of a form of distributed ``leader election'' in which a nondeterministic process grows a rectangular ``frame'' around each input assembly in a bounded amount of time (based on the shape). Although the perimeter of the input assembly is uniformly marked by a single, generic glue type, once the rectangular frame is completed a distinct corner is identified as the basis for sending a signal that marks a single location on the perimeter of the input assembly and allows the construction to proceed to the next stage with a guarantee of order of growth and correctness. 
This is especially difficult given the distributed, parallel, and asynchronous nature of tile assembly in the STAM combined with the facts that absolutely no assumptions can be made about the input shapes and that there can be multiple input shapes. This process of uniquely identifying a single point on the perimeter of an assembly of arbitrary shape and uniform glue type requires the bulk of the complexity of our construction, and may perhaps be of use for other constructions in future work with ultimately different goals.

\update{A universal shape-replicating tile system has also been demonstrated for the 2-Handed Tile Assembly Model with negative glues (-2HAM) \cite{chalkUniversalShapeReplicators2017}.
 Similar to the tileset presented in this paper, the -2HAM tileset in \cite{chalkUniversalShapeReplicators2017} takes hole-free shapes as input and generates copies of the input shape via a common process of generating a frame around
 the target shape. 
 The overall process of -2HAM is similar to that of the STAM replicator presented in this work; create a frame around the input shape, detach the frame from the input, fill in frame with new tiles, and eject the replicated shape from the frame. 
The performance of the -2HAM replicator varies from the STAM replicator in this paper in both the types of shapes which can be replicated and the information provided to the replicator, in the form of exposed glues.
 First, the set of shapes which can be replicated in this paper is not limited, whereas the gadgets required for the -2HAM replicator are limited by the `feature size' metric (seen also in \cite{RNaseSODA2010}).
 Second, the need for ``leader election'' as carried out by the STAM replicator is bypassed due to the fact that the input shapes encode the specific location for binding of the tiles which initiate the replication process.
 The leader election process necessitates the large number of signals required for tiles to communicate between one another.
 Finally, the STAM replicator presented in this paper is shown to replicate \emph{exponentially}, wheas the -2HAM replicator is posited to be a \emph{quadratic} replicator.}


The topic of shape replication is interesting to study for multiple reasons. First is the potential for practical applications, since having even a much more restricted shape replication system could be very useful. A nanoscale ``copying machine'' could ease and speed up the production of targeted structures and materials. Additionally, we hope that the study of these simpler shape replicators can lead to studies of more complex systems capable of evolution directed by carefully structured selective pressures, which may give more insight into the process of evolution in general and provide additional avenues toward biomimicry and utilizing the power of evolution to design structures.

This paper is organized as follows. In Section \ref{sec:prelims} we present a high-level introduction to the models used in this paper and several definitions used throughout. In Section \ref{sec:frame-building} we present our leader election construction (which we refer to as \emph{frame building}), and in Section \ref{sec:rep} we present the shape replication result.  Note that a preliminary and greatly shortened version of this paper was published in \cite{STAMshapes}.

\section{Preliminaries}\label{sec:prelims}

Here we provide informal descriptions of the models and terms used in this paper.  Formal definitions can be found in \cite{jSignals3D}.

\subsection{Informal definition of the 2-Handed Assembly Model}

The 2-Handed Assembly Model (2HAM) \cite{AGKS05g,DDFIRSS07} is a mathematical model of tile-based self-assembly systems, where the basic components are square \emph{tiles} that can have \emph{glues} on their edges allowing them to bind together. The 2HAM is a generalization of the abstract Tile Assembly Model (aTAM) \cite{Winf98} that allows for two assemblies, both possibly consisting of more than one tile, to attach to each other. 
We now give a brief, informal, sketch of the 2HAM,

A \emph{tile type} is a unit square with each side having a \emph{glue} consisting of a \emph{label} (a finite string) and \emph{strength} (a non-negative integer).   We assume a finite set $T$ of tile types, but an infinite number of copies of each tile type, each copy referred to as a \emph{tile}.
A \emph{supertile} is (the set of all translations of) a positioning of tiles on the integer lattice $\Z^2$.  Two adjacent tiles in a supertile \emph{interact} if the glues on their abutting sides are complimentary and have positive strength.
Each supertile induces a \emph{binding graph}, a grid graph whose vertices are tiles, with an edge between two tiles if they interact.
The supertile is \emph{$\tau$-stable} if every cut of its binding graph has strength at least $\tau$, where the weight of an edge is the strength of the glue it represents.
That is, the supertile is stable if at least energy $\tau$ is required to separate the supertile into two parts.

Given a set of tiles $T$, define a \emph{state} $S$ of $T$ to be a multiset of supertiles, or equivalently, $S$ is a function mapping supertiles of $T$ to $\N \cup \{\infty\}$, indicating the multiplicity of each supertile in the state. We therefore write $\alpha \in S$ if and only if $S(\alpha) > 0$.

A \emph{(two-handed) tile assembly system} (\emph{TAS}) is an ordered triple $\mathcal{T} = (T, S, \tau)$, where $T$ is a finite set of tile types, $S$ is the \emph{initial state}, and $\tau\in\N$ is the temperature.  For notational convenience we sometimes describe $S$ as a set of supertiles, in which case we actually mean that $S$ is a multiset of supertiles with one count of each supertile. We also assume that, in general, unless stated otherwise, the count for any singleton tile in the initial state is infinite.

Given a TAS $\calT=(T,S,\tau)$, a supertile is \emph{producible}, written as $\alpha \in \prodasm{T}$, if either it is a single tile from $T$, a supertile from $S$, or it is the $\tau$-stable result of translating two producible assemblies without overlap. 
A supertile $\alpha$ is \emph{terminal}, written as $\alpha \in \termasm{T}$, if for every producible supertile $\beta$, $\alpha$ and $\beta$ cannot be $\tau$-stably attached.

\subsection{Informal description of the STAM}

The STAM is an extension of the 2HAM which is intended to provide a model based on experimentally plausible mechanisms for glue activation and deactivation via \emph{signals} caused by glue binding events, but to abstract them in a manner which is implementation independent.
Therefore, no assumptions are made about the speed or ordering of the completion of signaling events (i.e. the execution of the transition functions that activate and deactivate glues and thus communicate with other tiles via binding events).  This provides a highly asynchronous framework in which care must be made to guarantee desired results, but which then provides robust behavior independent of the actual parameters realized by a physical system.
A detailed, technical definition of the STAM model is provided in \cite{jSignals3D}.

In the STAM, tiles are allowed to have sets of glues on each edge (as opposed to only one glue per side as in the aTAM and 2HAM).  Tiles have an initial state in which each glue is either ``$\texttt{on}$'' or ``$\texttt{latent}$'' (i.e. can be switched $\texttt{on}$ later).  Tiles also each implement a transition function which is executed upon the binding of any glue on any edge of that tile.  The transition function specifies, for each glue $g$ on a tile, a set of glues (along with the sides on which those glues are located) and an action, or \emph{signal} which is \emph{fired} by $g$'s binding, for each glue in the set.  The actions specified may be to: (1) turn the glue $\texttt{on}$ (only valid if it is currently $\texttt{latent}$), or (2) turn the glue $\texttt{off}$ (valid if it is currently $\texttt{on}$ or $\texttt{latent}$).  This means that glues can only be $\texttt{on}$ once (although may remain so for an arbitrary amount of time or permanently), either by starting in that state or being switched $\texttt{on}$ from $\texttt{latent}$ (which we call \emph{activation}), and if they are ever switched to $\texttt{off}$ (called \emph{deactivation}) then no further transitions are allowed for that glue.  This essentially provides a single ``use'' of a glue (and any signals sent by its binding).  Note that turning a glue $\texttt{off}$ breaks any bond that that glue may have formed with a neighboring tile. Also, since tile edges can have multiple active glues, when tile edges with multiple glues are adjacent, it is assumed that all matching glues in the $\texttt{on}$ state bind (for a total binding strength equal to the sum of the strengths of the individually bound glues).  The transition function defined for a tile type is allowed a unique set of output actions for the binding event of each glue along its edges, meaning that the binding of any particular glue on a tile's edge can initiate a set of actions to turn an arbitrary set of the glues on the sides of the same tile $\texttt{on}$ or $\texttt{off}$.

As the STAM is an extension of the 2HAM, binding and breaking can occur between tiles contained in pairs of arbitrarily sized supertiles.  It was designed to model physical mechanisms which implement the transition functions of tiles but are arbitrarily slower or faster than the average rates of (super)tile attachments and detachments.  Therefore, rather than immediately enacting the outputs of transition functions, each output action is put into a set of ``pending actions'' which includes all actions which have not yet been enacted for that glue (since it is technically possible for more than one action to have been initiated, but not yet enacted, for a particular glue). Any event can be randomly selected from the set, regardless of the order of arrival in the set, and the ordering of either selecting some action from the set or the combination of two supertiles is also completely arbitrary.  This provides fully asynchronous timing between the initiation, or firing, of signals (i.e. the execution of the transition function which puts them in the pending set) and their execution (i.e. the changing of the state of the target glue), as an arbitrary number of supertile binding (or breaking) events may occur before any signal is executed from the pending set, and vice versa.  

A STAM system consists of a set of tiles and a temperature value.  To define what is producible from such a system, we use a recursive definition of producible assemblies which starts with the initial (super)tiles and then contains any supertiles which can be formed by doing the following to any producible assembly:  (1) executing any entry from the pending actions of any one glue within a tile within that supertile (and then that action is removed from the pending set), (2) binding with another supertile if they are able to form a $\tau$-stable supertile, or (3) breaking into $2$ separate supertiles along a cut whose total strength is $< \tau$ (due to one or more glues along that cut having been deactivated).



\subsection{Additional definitions}\label{sec:frame-defns}

Throughout this paper, we will use the following definitions and conventions. 
We define a \emph{shape} as a finite, connected subset of $\Z^2$.  Following~\cite{ShapeIdentAlgo}, we say that a shape $s$ is \emph{hole-free} if the complement of $s$ is an infinite connected subset of $\Z^2$. We say that an assembly $\alpha$ is hole-free if $\dom \alpha$ is hole-free.  Then, an \emph{input assembly} is a non-empty, $\tau$-stable, hole-free assembly $\alpha$ such that every glue on the perimeter of $\alpha$ is strength 1 and of the same type. Throughout this section, we denote an input assembly by $\alpha$ and the type of the glue exposed on the perimeter of $\alpha$ by $x$.  Furthermore, a \emph{side} of a shape is any segment of the perimeter which connects two vertices (each of which can be convex or concave).

The algorithm that we present here will make use of many different gadgets and refer to various features of an input assembly and the shape of this input assembly. For convenience and brevity, we use the following conventions.

\begin{enumerate}

    \item[-] When a tile initially binds to an assembly, we call the sides which have glues that participate in that initial binding event the \emph{input} side(s), and the other sides the \emph{output} sides.
    \item[-] Respectively for each of northeast, southeast, southwest, and northwest, define convex corners with that pair of incident edges as $\necon$, $\secon$, $\swcon$, and $\nwcon$
    \item[-] Respectively for each of northeast, southeast, southwest, and northwest, define concave corners with that pair of incident edges as $\necav$, $\secav$, $\swcav$, and $\nwcav$
    \item[-]  CW: clockwise, CCW: counterclockwise
    \item[-] Let $R$ be the smallest bounding rectangle for $\alpha$ (i.e. the smallest rectangle that completely contains $\alpha$)
    \item[-] Let $PERIM(\alpha)$ be the set of all perimeter edges of $\alpha$.

	\item[-] Let $C$ be the convex hull of the set, $S$ say, of points in $\R^2$ defined by the corners of each tile. (That is, $(x,y)\in S$ iff $(x - \frac{1}{2},y - \frac{1}{2})$, $(x + \frac{1}{2},y - \frac{1}{2})$, $(x - \frac{1}{2},y + \frac{1}{2})$, or $(x + \frac{1}{2},y + \frac{1}{2})$ is in $\dom \alpha$.)
 	\item[-] We define $EXT(\alpha)$ to be the set of edges in $PERIM(\alpha)$ that are completely contained in the boundary of $C$.
 	\item[-] We define $INT(\alpha)$ to be $PERIM(\alpha) / EXT(\alpha)$.
	\item[-] Then, a \emph{concavity} of $\alpha$ is defined to be a subset, $E$, of edges of $INT(\alpha)$ such that each edge of $E$ is incident to some other edge in $E$.
\end{enumerate}
For an intuitive picture of the last five definitions above, see Figure~\ref{fig:convex-hull}.

\begin{figure}[htp]
\centering
    \includegraphics[width=3in]{images/convex-hull}
    \caption{An input assembly $\alpha$ consisting of the gray tiles. (Strength-1 $x$ glues on the exterior of $\alpha$ are not shown.) The blue plus red highlighted edges make up $PERIM(\alpha)$. The dashed line depicts the boundary of the convex hull $C$. The edges highlighted in blue make up the edges of $EXT(\alpha)$, and the edges highlighted in red make up the edges of $INT(\alpha)$. \vspace{-10pt}}
  \label{fig:convex-hull}
\end{figure}

\vspace{-10pt}
\section{Frame Building}\label{sec:frame-building}

Throughout this paper, we provide constructions which take as input assemblies of arbitrary 2D hole-free shapes. 
All input assemblies have completely uniform perimeters in terms of glue labels, meaning that no location on a perimeter is marked anyhow differently from the others.  Given the local nature of the self-assembly process, namely that tiles bind based only on local interactions of matching glues, and also with the order and locations of tile attachments being nondeterministic and growth of assemblies massively in parallel, a distributed problem such as ``leader election'' can be quite difficult, and similarly so is the problem of uniquely identifying exactly one point on the perimeter of an input assembly when no assumptions can be made about the shape other than the facts that it: (1) is connected, and (2) has no interior holes which are completely surrounded by the assembly.  Therefore, in this section we provide a construction which is a single universal STAM system, with temperature parameter equal to 2, capable of forming \emph{frames}, or simply borders composed of tiles, completely surrounding input assemblies in such a way that the growth of the frames performs a distributed algorithm which uniquely identifies exactly one perimeter location on each input shape.  We leverage this by using the tile at this location to verify that a completed frame surrounds the shape to be replicated and then detach the frame from the shape.
While this algorithm and STAM system, as well as several of the novel techniques, are of independent interest, they also play integral roles in the remaining constructions of this paper and will potentially also provide a useful toolkit for future constructions in others.

At a very high-level, the frame building construction can be broken into three main components.  First, a series of layers of tiles attach to the input assembly $\alpha$, each slowly helping to fill in the openings to any concavities, until eventually $\alpha$ is enclosed in an assembly which has a rectangular outer layer.  Second, that rectangular layer is able to detach after its unique southeast corner tile ``gadget'' initiates the propagation of a signal inward through all of the layers to the easternmost of the southernmost tiles which have attached directly to $\alpha$.  The tile that is immediately to the left of this tile is elected as the ``leader'' of the frame.  Third, the leader initiates a signal which propagates in a counterclockwise direction around $\alpha$, carefully ensuring that the entire perimeter of $\alpha$ is surrounded by tiles which have bound to it and made a complete ``mold'' of the shape.  After this is accomplished, the entire frame detaches from $\alpha$.  The result is a perfect mold of $\alpha$, with generic glues exposed around its entire interior surface except for one specific location, the leader, which exposes a unique glue.  It is from this unique glue that the frame assembly will then be able to initiate the growth which fills in the frame, making a replica of $\alpha$.


\vspace{-5pt}
\subsection{Building layers of the frame}

We now give an extremely high-level sketch of the formation of the \emph{frame}.  (See Sections~\ref{sec:layer1} and \ref{sec:frame-layers} for more details.)  Essentially, the frame grows as a series of layers of tiles which begin on (possibly many) southeast convex corners of $\alpha$ (depending on its shape) and grow counterclockwise (CCW) around $\alpha$.  A greatly simplified example of the basic tiles which form layers of the frame can be seen in Figure~\ref{fig:basic-convex-frame}.  Each path which forms a layer can grow only CCW, and therefore, depending on $\alpha$'s shape, may crash into a concavity of $\alpha$ (or one formed by a previous layer that the current layer is growing on top of).  Such \emph{collisions} are detected by a specialized set of \emph{collision detection} tiles, and an example of a collision and its detection is shown in Figure~\ref{fig:new-frame2-concave-collision1}.  The need for collision detection tiles is technical and related to the need for the exposed glues on all parts of the growing frame assembly to be minimized and carefully controlled so that multiple shapes and copies of shapes can be replicated in parallel, without separate assemblies interfering with each other.

The growth of frame layers is carefully designed so that they are guaranteed to proceed until all external openings to concavities of $\alpha$ have been filled in by partially completed layers, resulting in layers which are more and more rectangular, and eventually an exactly rectangular layer.  At this point, and only at this point, we are guaranteed to have a layer which has exactly one convex southeast corner. Due to the distributed and asynchronous nature of the assembly process, and the fact that each tile only has local information, throughout layer formation it is necessary for some layers to make local ``guesses'' that they are rectangular, and in order for that not to cause errors, a mechanism of layer detachment is used.  Basically, layers which guess they may be rectangular attempt to disconnect, but only the first truly rectangular layer can successfully detach.  At this point it activates glues on the layer immediately interior of it, which it has \emph{primed} to receive a signal (i.e. a glue which can bind to receive the signal has been turned on) from a tile which will now be free to attach since the covering exterior layer dissociates. This is then used in the unique leader election.  It is by the careful use of the ``global'' information provided by the layer detachment that the construction can proceed correctly.

\subsection{Overview of frame building tiles and signals}
\update{
A key benefit of STAM over other models of self-assembly (particularly aTAM) is the ability of STAM tiles to reuse physical space; in essence, each tile can carry out multiple computations via glue activation and deactivation. 
We leverage this advantage by providing tiles that are able to communicate by the successive activation of glues. 
Our design methodology revolves around two key aspects - the set of tiles along with their initial binding conditions to a growing frame, and the sets of glue activation signals which are necessary to grow the frame around $\alpha$ such that each frame uniquely maps to the input shape provided, and is able to detach only upon guarantee of completion. 
In this section we provide the descriptions about what tiles are present in the system, and the various functionality provided by signals. We do not provide exact descriptions of the signals in this section; these are present in the following sections with details of the frame assembly and leader election processes.
}

\subsubsection{Overview of frame building tiles}\label{sec:base-tiles}

The local information which tiles utilize to bind to the growing structure comprises of both the shape of the growing frame, and the active glues presented. We first provide the set of tiles which serve as the template for the signals to be added. We omit single-use tiles which are context specific to a single phase of the frame building process: collisions (Section~\ref{sec:frame-layers}), the leader election process (Section~\ref{sec:leader-details}), and mold creation (Section~\ref{sec:outlining}). Figure~\ref{fig:basic-convex-frame} includes all tiles included in frame building which will have additional modifications to incorporate signal passing.

\begin{figure}[htp]
\centering
    \includegraphics[width=\textwidth]{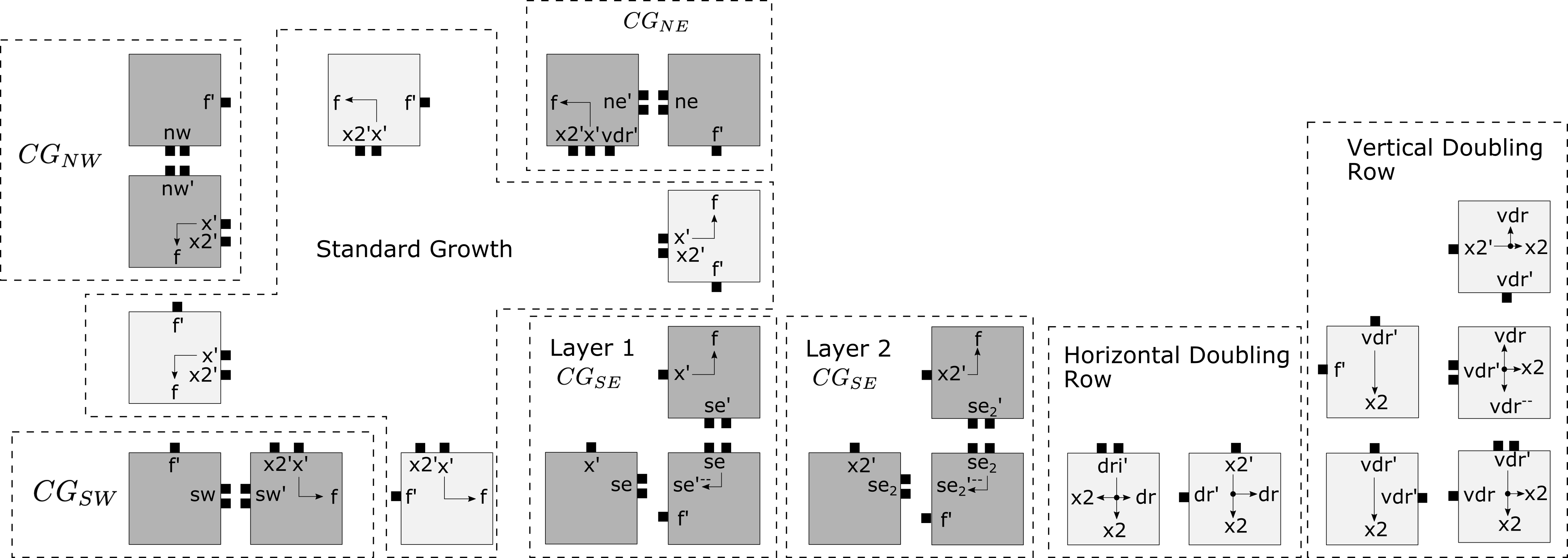}
    \caption{A simplified version of the tile set which grows the  frame around $\alpha$.  
    The darker grey tiles represent \emph{corner gadgets} ($CG$) which form as duples (or a triple for the southeast corner gadget, $\cgSE$).
     Additionally, $\cgSE$ is able to either form the triple, or remain a duple and bind overhangs as shown in Figure~\ref{fig:frame-corner-cases1}(e). 
     All tiles with both $x2'$ and $x'$ glues are able to bind to $\alpha$ and any subsequent layer, and also initiate the $f$ signal. Horizontal doubling row and vertical doubling row tiles are included.
    Additionally, strength 2 glues which begin in the latent state are identified with a superscript of two dashes;
    for example, $se'^{--}$
    Note that this is a basic, beginning tile set to which we will add additional signals and tiles throughout the construction.\vspace{-10pt}}\label{fig:basic-convex-frame}
\end{figure}

\update{
\subsubsection{Overview of frame building signals}\label{sec:frame-signals}

In this section we provide a high-level overview of the main signals that are propagated through tiles of the layers which grow a frame around $\alpha$. Rather than depicting individual tiles, we show segments of standard paths, $\cgSE$ gadgets, ``doubling rows'', collision detection tiles and gadgets, and the signals which propagate through them to control their growth and the growth of additional layers when necessary. The logical functionality of each signal is explained, as well as the (main) cases in which they are initiated. In order to provide a clear and relatively succinct description of each signal, we reserve explanations of most ``special cases'' for later sections (and a full enumeration of special cases can be seen in Figure \ref{fig:frame-corner-cases1}).

The main focus is to show how, in every case where a frame layer grows but is not rectangular (where a rectangular frame layer has exactly one of each of the $4$ convex corner types, and it grows out from the north of the same $\cgSE$ which it eventually collides with from the west) a signal will be propagated that allows for a new frame layer to grow immediately outside of it, and such a layer will always be able to be initiated by the attachment of a new $\cgSE$. Additionally, it is important that special ``doubling'' rows are added in certain situations so that rows are guaranteed to eventually become rectangular.

A listing of the relevant signals for this process, as well as an overview of when  each is initiated and the logical function of each follows:
}

\begin{figure}
    \centering
    \includegraphics[width=5.2in]{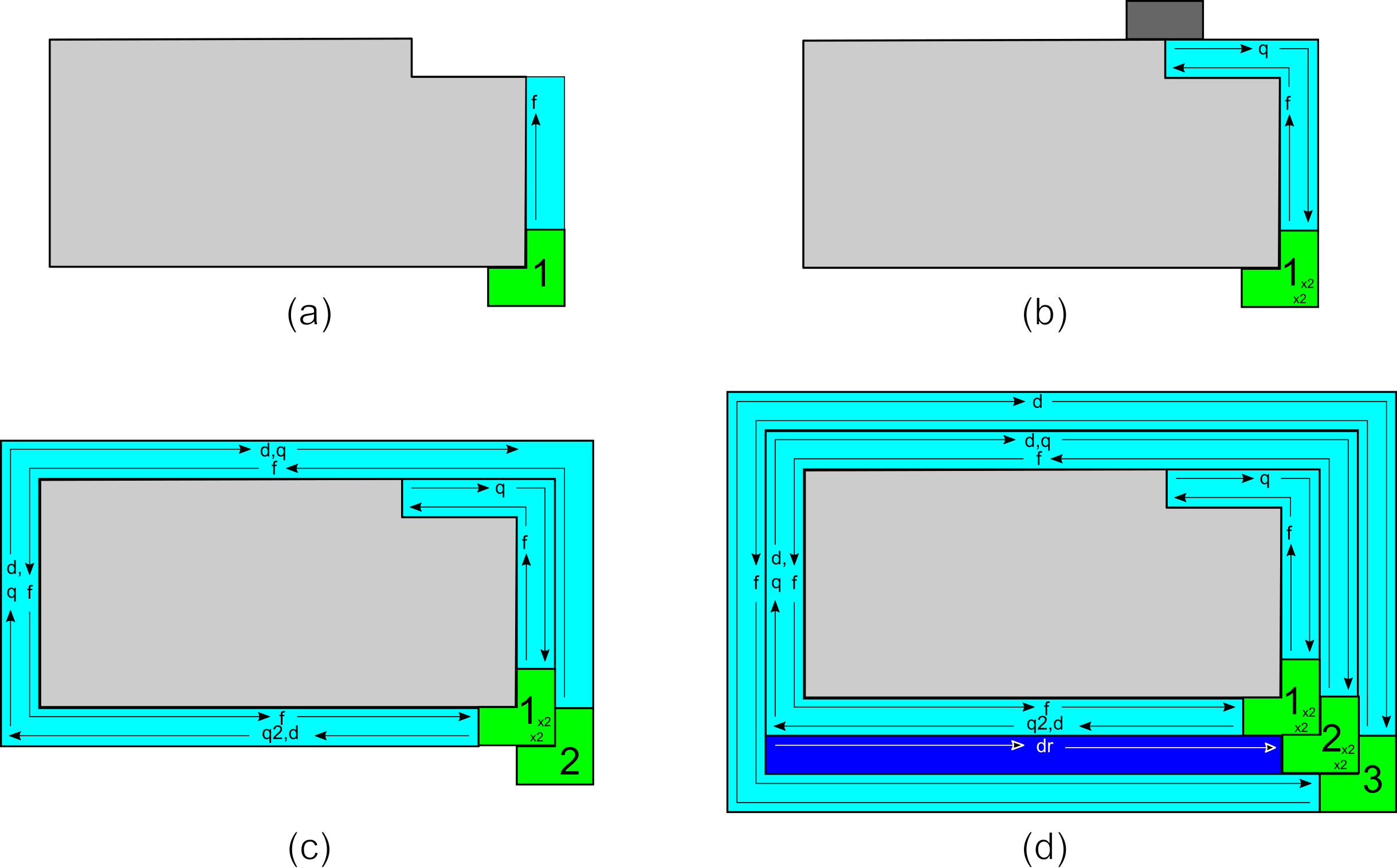}
    \caption{Simple example demonstrating several signals used during frame construction. In all figures light grey represents $\alpha$, each $\cgSE$ is shown in green, standard paths are light blue, and the horizontal doubling row is dark blue. (a) The first $\cgSE$ attaches to $\alpha$ and begins northward growth of a standard path, propagating the $f$ signal CCW. (b) The first standard path collides with $\alpha$ and the collision detection gadget initiates the CW propagation of the $q$ signal. This causes all tiles of the standard path and the $\cgSE$ to activate $x2$ glues that allow for a new layer to piggyback. (c) A second layer grows a standard path which eventually collides with the first $\cgSE$. This collision initiates the CW passing of a $d$ signal, as well as a $dri$ signal (since the $\cgSE$ received the $q$ signal from its north). When the $dri$ signal reaches the west side, it is turned into a $q$ signal that travels along with the $d$ signal. (d) The $dri$ signal initiates the growth of the horizontal doubling row once. The $q$ signal causes the third $\cgSE$ to attach and the next layer to grow. This standard path collides with the $\cgSE$ which initiated it. Note that if a horizontal doubling row was not forced to grow before that layer, it would not be rectangular and able to collide with the third $\cgSE$. Without horizontal doubling rows, in a situation like this layers would continue to grow infinitely far without ever becoming rectangular.}
    \label{fig:horiz-doubling-example}
\end{figure}

\begin{figure}
    \centering
    \includegraphics[width=6.2in]{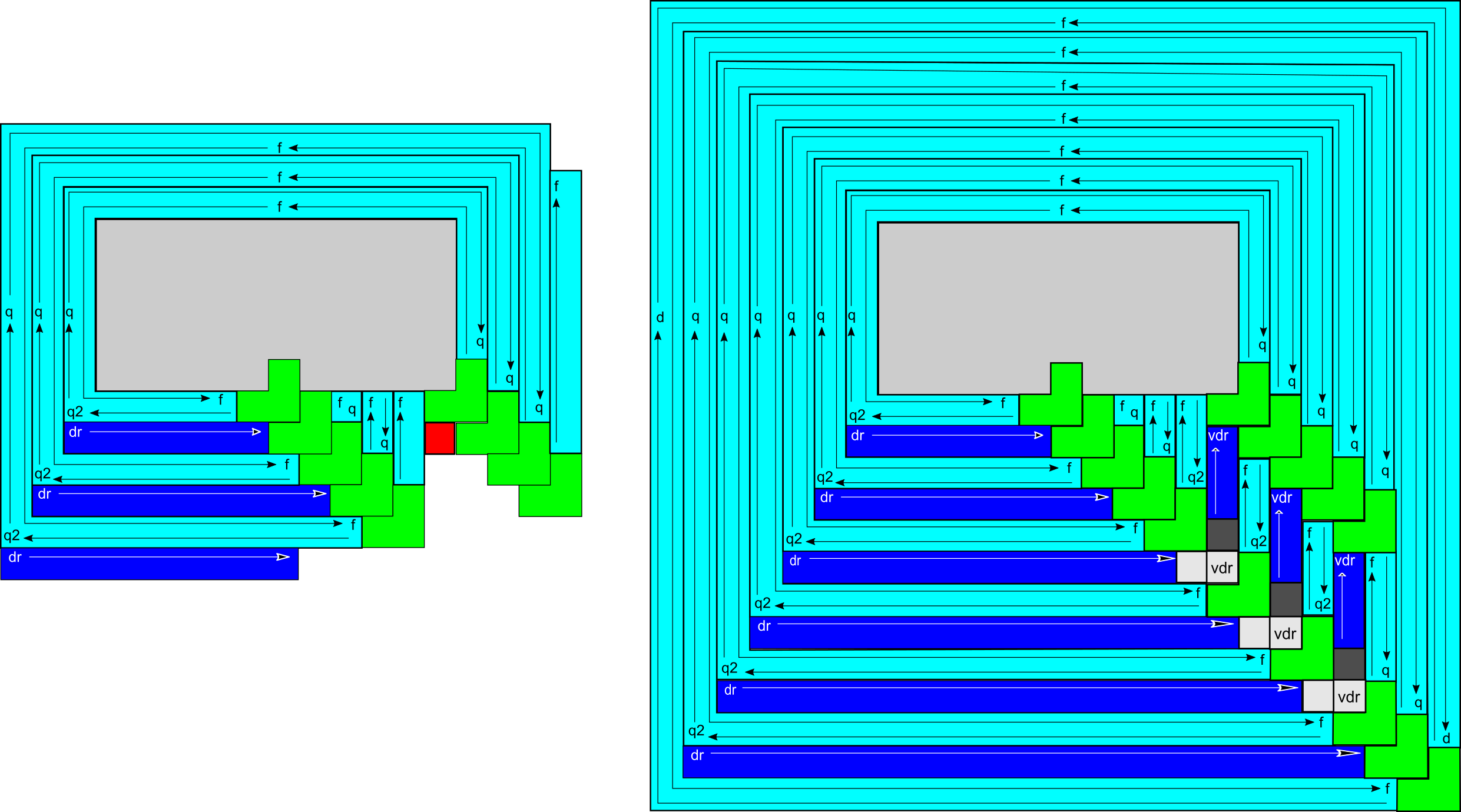}
    \caption{Example frame growth around an input assembly $\alpha$ (grey) which is nearly rectangular (with the exception of a single-tile concavity on the south). (Left) The first few layers of growth from two stacks of exterior $\cgSE$'s.  $\cgSE$'s are green, standard paths in light blue, and doubling rows in dark blue.
    Doubling rows are able to form from the first $\cgSE$ in the stack, as both layer 1 and layer 2 $\cgSE$ have an $x$ glue on the northernmost tile which sense if it is adjacent to $\alpha$.
    Note that the doubling row signal is activated for the second $\cgSE$ in the left stack, however since it is adjacent to a doubling row (and not a standard path), the message will not be transmitted as the doubling row tiles do not contain $q2$ signals.
    Also highlighted in red is the location of the first collision detection tile attachment which detects that the upward growing path is next to another $\cgSE$.  (Right) Once the collision shown in red on the left is encountered, every new layer of the left stack of $\cgSE$'s will be spaced by a vertical doubling row which causes each newly attaching $\cgSE$ in the stack to be offset from the previous by $(2,-2)$.  This offset guarantees that the left stack eventually covers the right stack, finally allowing for the growth of a perfectly rectangular layer.}
    \label{fig:stack-catch-up}
    \vspace{-10pt}
\end{figure}
\update{
\begin{enumerate}
    \item $f$: the ``forward'' signal propagates CCW through a standard path, allowing it to continue until it crashes into (1) $\alpha$, (2) some portion of a standard path on a different layer, or (3) a $\cgSE$.
    
    \item $q$: the ``quit'' signal which begins from the point of a collision of a standard path and propagates CW and CCW through that standard path. As it propagates, the tiles of the standard path turn on $x2$ glues on their right sides. When it returns to the $\cgSE$ from which the standard path initiated, that $\cgSE$ turns on the $x2$ glues needed to allow another $\cgSE$ to piggyback on it and begin the growth of a standard path which piggybacks on the previous.
    
    \item $d$: the ``detach'' signal begins from the collision of a standard path with a $\cgSE$ in standard position (i.e. when the standard path is growing from west to east and collides with the westernmost tile of a $\cgSE$). The $d$ signal is passed by a strength-2 glue and propagates CW through the standard path and each tile receiving it turns off the $x2'$ glue on its left side. A $\cgSE$ only turns off its $x2'$ glues if it receives this signal from its north.
    
    \item $q2$ (horizontal): the horizontal ``quit 2'' signal is used to initiate the growth of a \emph{horizontal doubling row} (i.e. an extra row of tiles on the southern side of a standard path). Similar to the $d$ signal, it begins from the collision of a standard path with a $\cgSE$ in standard position. However, it is only initiated by a $\cgSE$ which has received a $q$ signal from the standard path it initiated to its north. The $q2$ signal is passed CW, east to west along the southern row of the standard path and causes each tile to turn on a strength-1 $x2$ glue on its right (southern) side until reaching the westernmost tile of the row. This tile will be the west tile of a $\cgSW$ and will turn on a strength-2 $dri$ glue on its southern side and an $x2$ glue on its west side.
        \begin{enumerate}
            \item The first tile of a horizontal doubling row will be able to attach to the strength-2 $dri$ glue, and cooperative attachments will allow the rest of the row to grow from west to east until attaching to the westernmost tile of the $\cgSE$ which initiated the $q2$ signal. Each of the tiles of this horizontal doubling row will expose an $x2$ glue on its south side.
            
            \item The west tile of the $\cgSW$, which activates its own south-facing $x2$ glue, initiates a $q$ signal which travels CW through the remaining portion of the standard path.
            
            \item The logical behavior of the $q2$ signal is to continue the CW propagation of the $q$ signal received by the $\cgSE$ from its north through a standard path which grows to its west, while also causing one extra row to grow on the south of that standard path. The necessity for this extra row, depicted in Figure \ref{fig:horiz-doubling-example}, is that if one or more layers of standard paths had collisions before some standard path successfully grew around to return to a $\cgSE$, there would be fewer standard paths arriving at $\cgSE$s than $\cgSE$s. The inclusion of horizontal doubling rows allows the standard paths growing on the southern side to ``catch up'' to the outermost $\cgSE$ and eventually form a rectangular layer.
        \end{enumerate}
        
    \item $q2$ (vertical): the vertical ``quit 2'' signal is used to indicate that two (diagonally growing) stacks of $\cgSE$ gadgets are about to overlap, or have already overlapped (i.e. the western stack has grown to a location where its easternmost tile is at an $x$-coordinate one less than that of a $\cgSE$ of the eastern stack). It is initiated by a special collision tile that detects that a tile of an eastern (i.e. north growing) portion of a standard path is diagonally adjacent to the west tile of a $\cgSE$. (Figure \ref{fig:stack-catch-up} depicts a situation requiring this behavior at the location in red, and Figure \ref{fig:collision-detection-corner-cases} depicts the tile that detects this collision.)
        \begin{enumerate}
            \item When the collision detection tile attaches, it causes the standard path to pass the $q2$ signal south to the $\cgSE$, and for the standard path to activate glues on the east to allow for attachment of vertical doubling row tiles.
            
            \item Once the $q2$ signal reaches the $\cgSE$ at the south end of the standard path, it initiates the propagation of a $q2$ signal to the incoming standard path on the west side of the $\cgSE$. 
            
            \item The $\cgSE$ which recieves a $q2$ signal also activates an eastern strength-2 glue that allows the first vertical doubling row (`$vdr$') tile to attach, and the rest of the vertical doubling row tiles attach to complete the row's growth to the north via cooperation. Each tile of the vertical doubling row exposes an $x2$ glue to its east, allowing another layer to grow over them.
            
            \item A tile attaches to the south of the first $vdr$ tile and allows a new $\cgSE$ to piggyback on it.
            
        \end{enumerate}
\end{enumerate}
}
\begin{figure}
    \centering
    \includegraphics[width=4.2in]{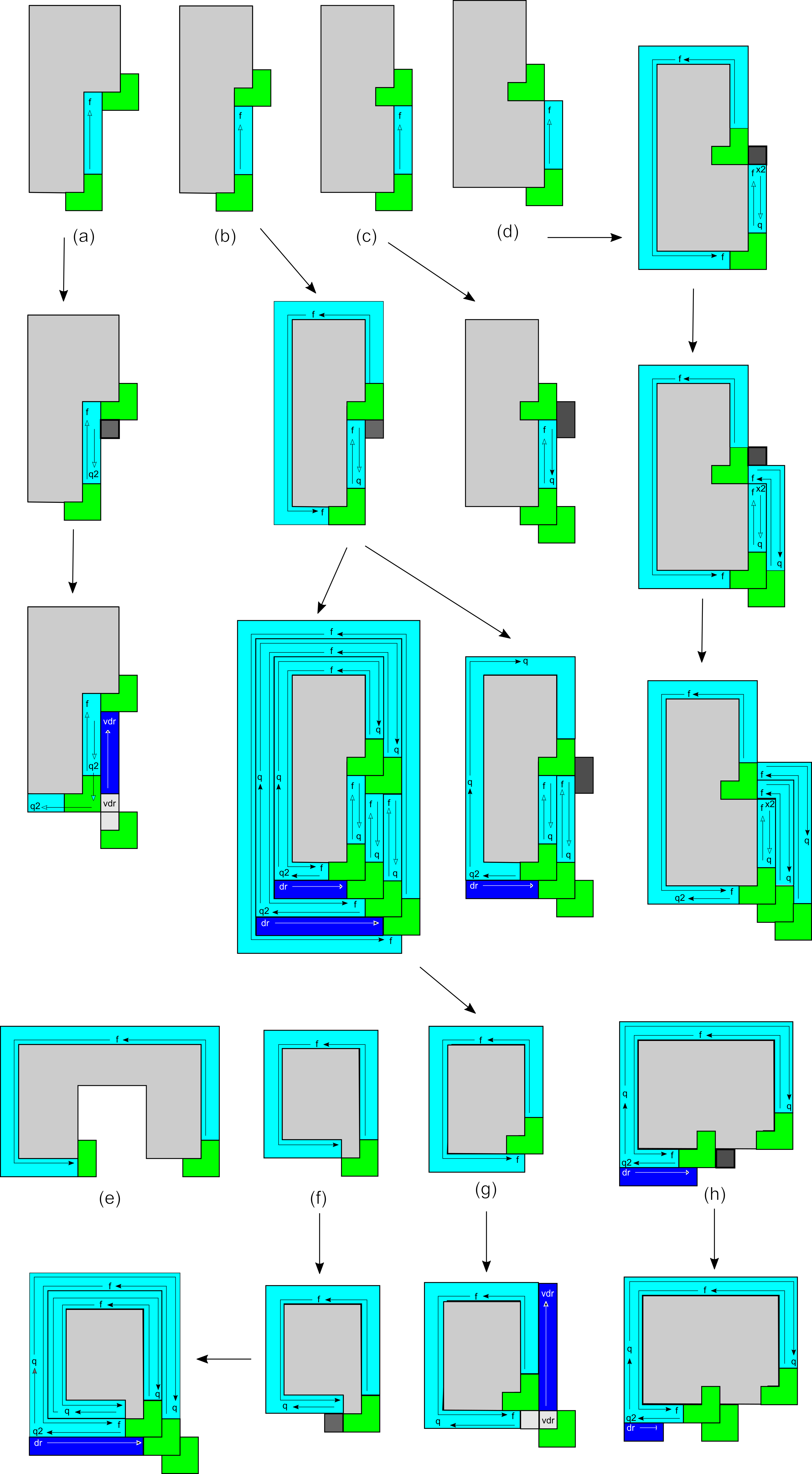}
    \caption{The figures labeled (a) through (h) show simple assemblies displaying each of the ``corner cases'' dealing with every possible situation where a standard path interacts with a $\cgSE$. Assemblies following arrows depict stages of subsequent growth. Relevant signals are shown (and the $d$ signal is omitted since it does not influence the behaviors depicted). Grey portions in the centers represent $\alpha$, $\cgSE$ gadgets are green, standard paths segments are light blue, doubling rows dark blue, collision tiles and gadgets are dark grey, and a few special case vertical doubling row tiles are light grey. For each scenario, the initial interaction between a standard path and a $\cgSE$ is shown, and then one or more subsequent assemblies show additional growth depicting how that particular scenario is handled (except for (e) which simply shows the situation of a standard path growing into the westernmost location which could possibly be occupied by a $\cgSE$ and the 2-tile version of the $\cgSE$ gadget attaching instead).}
    \label{fig:frame-corner-cases1}
\end{figure}

\subsubsection{Prevention of faulty tile addition and faulty signal generation}
\update{
The wide variety of signals which are present have the potential to generate spurious tile additions and unexpected (and potentially invalid) final assemblies.
These outcomes are prevented by utilizing 3 main categories of constraints in our construction.

First, each frame tile (or frame gadget) placed is done so deterministially with only one exception. Specifically, for any frame assembly, at each point of growth there exists only a single
valid tile (or gadget) placement which can occur except for specific instances at southeast corners of shapes in all but a single case.
This determinism is \emph{not} in reference to the frames which can possibly be generated from an input shape $\alpha$; due to the asynchronous nature of STAM, 
certain tile placements may be carried out before others and enable (or block) the placement of valid tiles.
Indeed, asynchronous signal activation and tile placement causes a number of edge cases as demonstrated in Figure~\ref{fig:frame-corner-cases1}.
In this construction, the single case where nondeterministic tile placement is possible in the case where either a frame growth tile may bind to the southern $x$
glue of a southeast corner of $\alpha$, or a $\cgSE$ may be placed. To counteract the possible blocking of the full 3-tile $\cgSE$ gadget, it is possible for the two eastern tiles that make a vertical rectangle
to attach to the frame growth tile and $\alpha$ via both $f'$ and $x'$ glues, instead of typically binding to its western tile with strength 2 $se$ glues.

Second, we prevent exposure of unnecessary glues during the signal passing process of frame building by utilizing collision detection tiles. 
Only a subset of the possible glues are presented on tile sides which are potentially locations for tile growth. 
In particular, those are the collision detection glues ($c, \: cg, \: cr, \: cq, \: c_{dir}, \:sr,$), duple attachement glues of strength 2, and
frame growth glues ($f, \: x, \: x2, \:vdr, \:vdi$). 
The remaining glues are activated only on sides of tiles which are adjacent to tiles of the same layer or within the same gadget.
As such, we can strictly control which internal signals are activated by placement of new tiles and attachment of collision detection tiles

Third, signal deactivation which leads to tiles breaking apart from assemblies is strictly limited to collision detection and the outermost frame tiles.
The only time that any frame tile detaches from the assembly is in the case of the entire outermost, rectangular frame detaching from the remainder of the structure.
In this case, all the $x2$ glues which bind a tile to the frame layer prior are deactivated, and no possible tile types can attach to the frame.
Additionally, collision tiles themselves have no signals - their attachment and detachment is controlled by the frame growth tiles. Thus, they will bind in the same
manner, regardless if they have been joined to a frame tile prior.
}

\subsection{Assembly of layer 1 tiles}\label{sec:layer1}

First we will describe the tiles that initially attach to $\alpha$, starting the frame building process. We refer to the tiles that assemble the frame as the \emph{frame building tiles}.
\update{
Moreover, as the frame assembles, any tile (or tile as part of a gadget) that is initially bound to a tile of $\alpha$ is referred to as a \emph{layer 1} tile. 
A \emph{layer 2} tile is initially bound to a layer 1 tile or another layer 2 tile. 
We note that layer 2 tiles (gadgets) can in some cases be bound to a tile of $\alpha$, but their initial placement is not caused by such interactions. 
As shown in Figure~\ref{fig:basic-convex-frame}, many of the frame building tiles have the ability to be either layer 1 or layer 2 tiles; the only tiles which are restricted to layer 1 are those part of the $\cgSE$ gadget with $x'$ glues, and the only tiles which are restricted to layer 2 are $\cgSE$ gadget tiles with $x2'$ glues and doubling row tiles. 
}

We now describe how the the frame assembly begins.

\begin{observation} \label{obs:secExists}
Every input assembly has at least one convex corner of each type.  This follows directly from the definition of a shape formed by an assembly. 
For instance, starting from an arbitrary point on the perimeter of any valid $\alpha$, follow the perimeter in a CCW direction until returning at the original 
location, after having traversed the entire perimeter of $\alpha$.  
This must be possible because $\alpha$ is a connected 2-D shape and the perimeter of $\alpha$ must be a continuous line.  
In order to move in a full circle around the perimeter which is composed of unit squares, it is necessary to make at least one turn in each direction (N, W, S, and E), and each such turn occurs at a convex corner $\secon$, $\necon$, $\nwcon$, and $\swcon$, respectively.
\end{observation}

\begin{enumerate}
    \item The first tile attachment to $\alpha$ must consist of a $\cgSE$ attaching to a $\secon$ corner of $\alpha$.  (See Figure~\ref{fig:basic-convex-frame}.) Since $\alpha$ may be non-rectangular, it may have multiple $\secon$ corners, but by Observation~\ref{obs:secExists} it must have at least one, and by the fact that every side of every tile on the perimeter of $\alpha$ exposes an $x$ glue, a $\cgSE$ can attach to the single tile of each $\secon$ by binding to its two exposed $x$ glues of strength $1$ each (for a total binding strength of $2$ in this temperature-$2$ system).  Therefore, one or more such attachments occur.
The tiles of a $\cgSE$ are described in Figure~\ref{fig:layer1-SEcg}. By Observation~\ref{obs:secExists}, at least one $\cgSE$ binds to $\alpha$, initiating the assembly of a frame. (Note that no other tiles or corner gadgets of layer 1 tiles are able to attach directly to $\alpha$ since they could only bind with at most strength 1.)

\begin{figure}[htp]
\centering
    \includegraphics[width=4in]{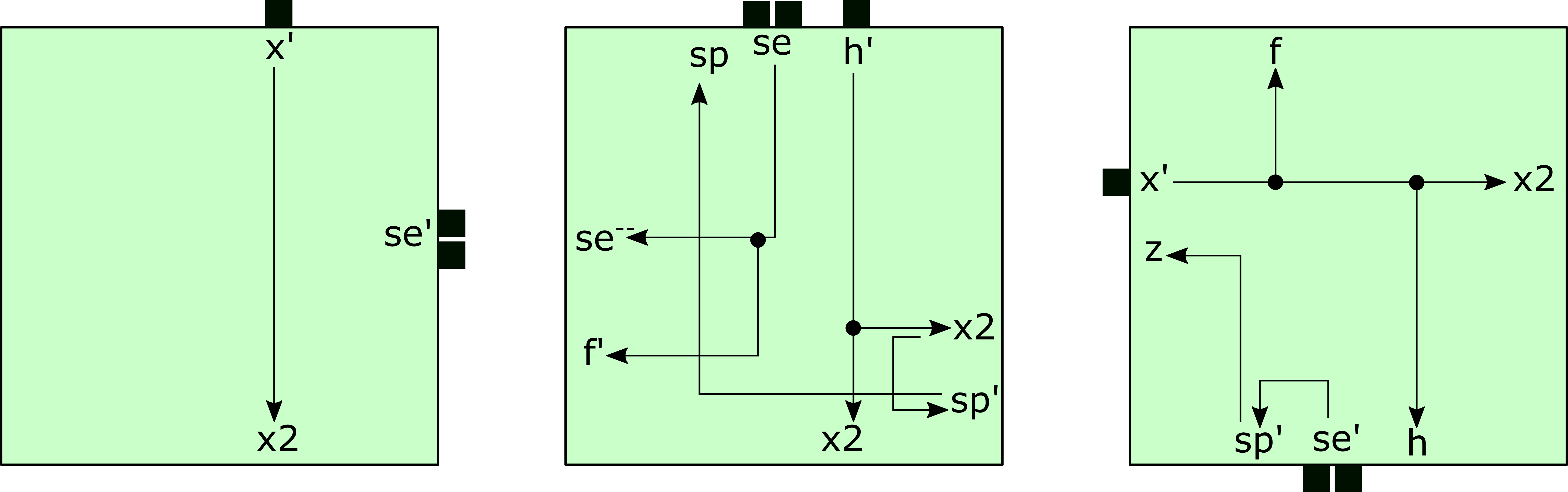}
    \caption{The 3 tiles that make up a layer 1 $\cgSE$ (a $\cgSE$ consisting of layer 1 tiles). 
    These tiles form a $\cgSE$ that can bind to two exposed $x$ glues. 
    (Left) A $\cgSE$ tile shown labeled $C$ in Figure~\ref{fig:secgCases}. 
    (Middle) A $\cgSE$ tile shown labeled $B$ in Figure~\ref{fig:secgCases}. 
    Note the $f'$ glue, which allows for the $\cgSE$ to bind as a duple in cases where it may be hindered by the growth of a standard path.
    (Right) A $\cgSE$ tile shown labeled $A$ in Figure~\ref{fig:secgCases}.  Note that after the tile a $\cgSE$ binds to $x$ glues, it triggers $x2$ glues (among others) to turn $\on$.}\label{fig:layer1-SEcg}
\end{figure}

    \item From each attached $\cgSE$, a single-tile-wide path grows CCW with each layer 1 tile binding to the $x$ glue on $\alpha$ and the output $f$ glue (logically representing ``forward'') of the immediately preceding tile, forming what we call a \emph{standard path}.  In this way, we say that the $f$ message is passed CCW through a standard path of the frame.  The \emph{forward} direction for the path follows the direction of its growth, and its \emph{backward} direction is the opposite.  The \emph{back} of a frame tile is the side which served as input during its initial attachment with the glue $f'$.  The \emph{front} is the side opposite to the back.  Similarly, the \emph{left} is the input side with the input glue $x'$, and the \emph{right} side is the side opposite to the left.  Note that we use these terms to talk about directions relative to individual tiles, and the cardinal directions N, E, S, W to talk about absolute directions. The tiles that assemble these standard paths are described in Figure~\ref{fig:layer1-east-edge-crawler}.

\begin{figure}[htp]
\centering
    \includegraphics[width=1in]{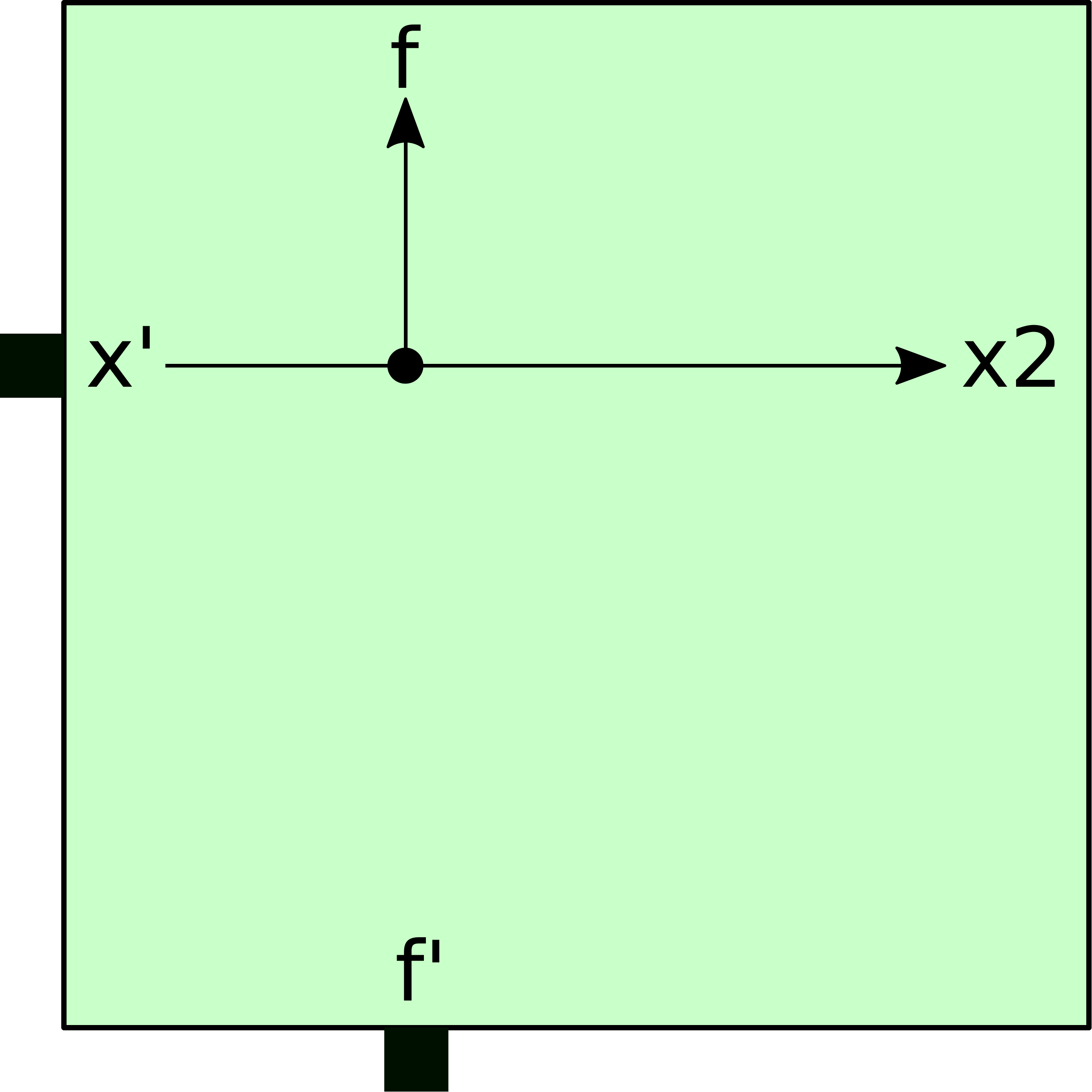}
    \caption{A tile that binds to an $x$ glue and an $f$ glue. The standard path that this tile assembles begins growth from a $\cgSE$. Then, as the standard path grows one tile attachment at a time, each consecutive tile that binds is located one tile location above the previously binding tile. That is, these tiles grow standard paths along the east edges of $\alpha$. Tile types that can grow standard paths along north, west, or south edges of $\alpha$ are obtained by appropriately rotating the tile type depicted in this figure. \vspace{-10pt}}\label{fig:layer1-east-edge-crawler}
\end{figure}

    \item A standard path of layer 1 can continue growing along straight edges of $\alpha$, with each tile using an $x'$ glue as input on its left side, and an $f'$ glue as input on its back side.  As each tile binds, it activates an $f$ glue on its front side, thus allowing the path to continue along straight edges.

    \item When arriving at a NE, NW, or SW convex corner, a simple 2-tile corner gadget (denoted by $\cgNE$, $\cgNW$, and $\cgSW$ respectively) allows for the initial layer to continue growing around the convex corner. 
    The tiles of a $\cgNE$ are described in Figure~\ref{fig:layer1-NE-gadget}.
    If a path arrives at an SE convex corner before a $\cgSE$ attaches, the middle and top tiles of Figure~\ref{fig:layer1-SEcg} are able to attach as a duple.

\begin{figure}[htp]
\centering
    \includegraphics[width=3in]{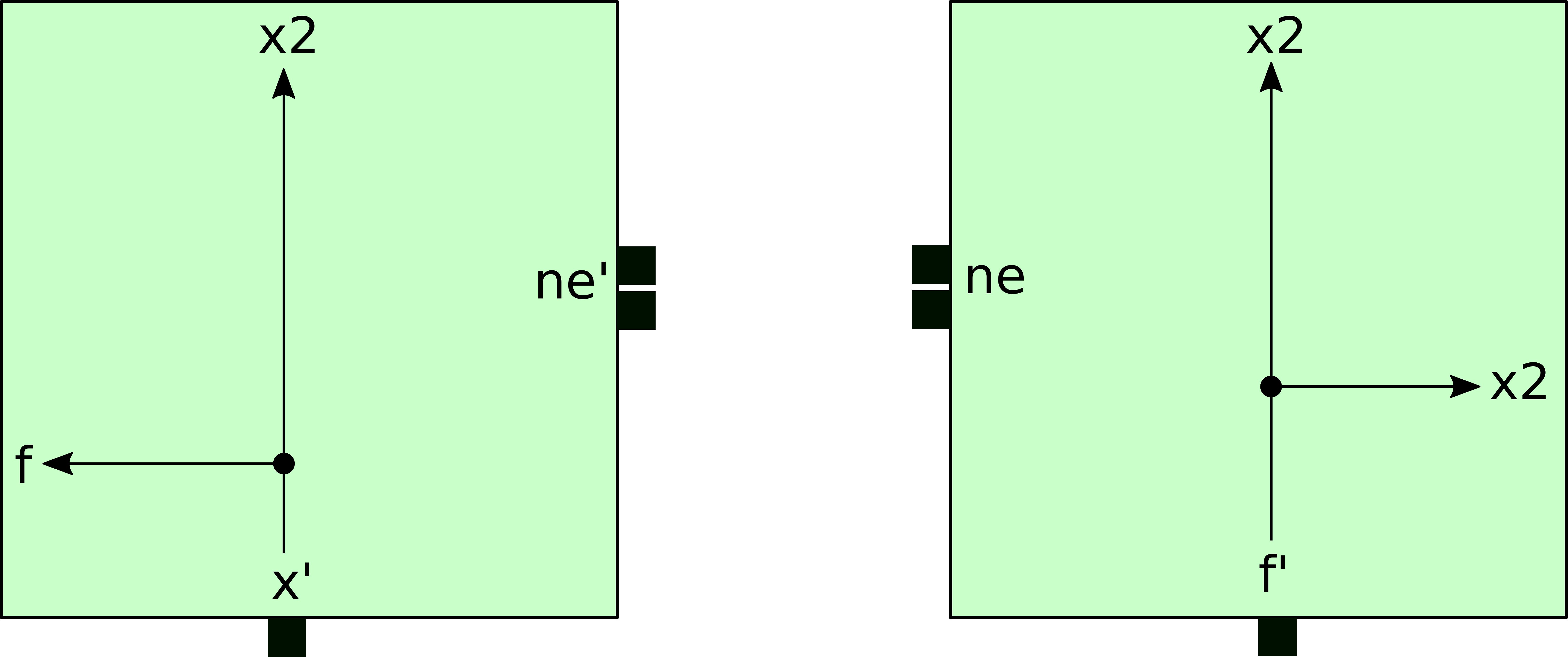}
    \caption{The 2 tiles that make up a $\cgNE$ that can bind to an $x$ glue and an $f$ glue. The 2-tile gadget formed by these tiles binds to an assembly effectively turning a $\necon$. The $f$ message now propagates from right to left. Moreover, once a $\cgNE$ binds, signals also trigger $x2$ glues to turn $\on$. To obtain a $\cgNW$ or a $\cgSW$, rotate the tiles shown in this figure appropriately CCW.\vspace{-10pt}}\label{fig:layer1-NE-gadget}
\end{figure}

    \item A standard path of tiles must eventually terminate in one of three situations:
        \begin{enumerate}
        \item It collides with (i.e. a tile of the path attaches so that its front side is immediately adjacent to another tile) some tile of $\alpha$ or a frame tile that is not part of a $\cgSE$. In this case, the standard path simply terminates.

        \item It grows from the west to a position adjacent to a $\cgSE$ at position 7 as shown in Figure~\ref{fig:secgCases}.  In this case, the standard path terminates since the tile $C$ in Figure~\ref{fig:secgCases} does not expose an $x$ glue on its south edge. \label{item:layer1-terminates-special-case}

\begin{figure}[htp]
\centering
    \includegraphics[width=1in]{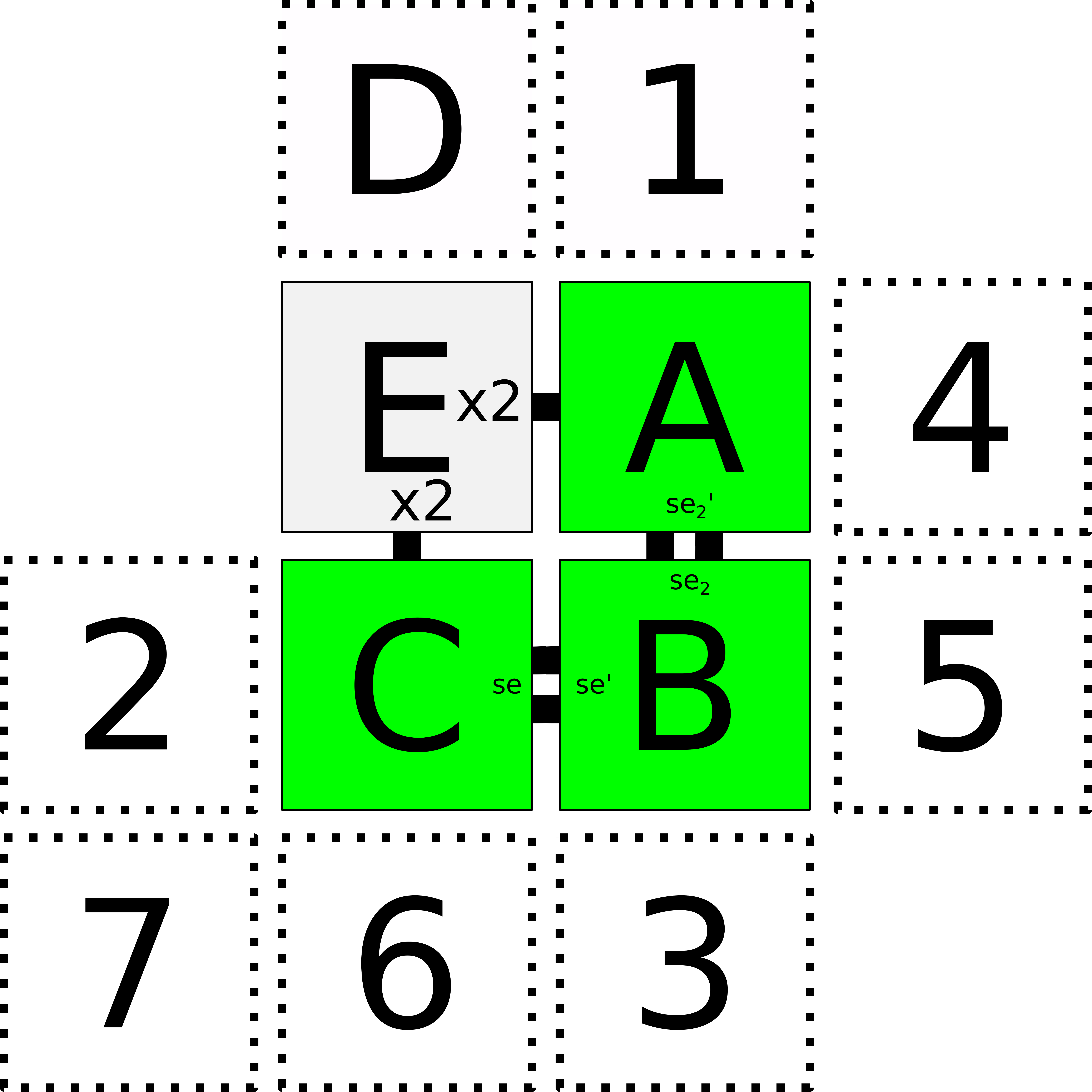}
    \caption{The tile locations in relation a $\cgSE$ (where the $\cgSE$ is represented by the green tiles).\vspace{-10pt}}\label{fig:secgCases}
\end{figure}

       \item It collides with a $\cgSE$ from position 2 (as shown in Figure~\ref{fig:secgCases}). Note that this collision with the $\cgSE$ is only considered to occur if the path was growing from the west. (The special case of it growing from the south, where the collision would be considered to be with the tile to the north, is discussed in more detail in Figure \ref{fig:vdr2}.)

        \end{enumerate}

\end{enumerate}

Layer 1 tiles assemble a frame by propagating an $f$ signal CCW around $\alpha$. 
Note that when $\alpha$ is not a rectangle, the growth of some standard paths of layer 1 may be blocked by tiles of $\alpha$. 
On the other hand, if $\alpha$ is a rectangle, then layer 1 tiles assemble a single tile wide rectangular frame around $\alpha$. In either case, layer 1 exposes $x2$ glues which will allow for additional frame tiles to attach, forming additional layers of the frame. The assembly of additional layers is given in the next section.

\subsection{Assembly of additional frame layers}\label{sec:frame-layers}

In this section, we will show how additional layers assemble. We will see that these layers assemble in a similar manner to layer 1, with the exception that the tiles that assemble additional layers are enhanced so that as a given layer is assembling, it is capable of ``detecting'' when it is not a complete rectangular frame layer. That is, each tile has been equipped with glues and signals that check to determine if any part of $\alpha$ or some frame layer has a tile located to its right side. When a layer detects that it is not a rectangle, it allows for an additional layer to attach by exposing more $x2$ glues. Moreover, each tile is enhanced with glues and signals so that when a standard path of tiles forms by starting from a $\cgSE$ and ending at a (possibly different) $\cgSE$, the tiles of this portion of the standard path ``guess'' that they belong to a rectangular frame layer and propagate a message that attempts to turn the left hand side $x2'$ glues of this frame layer $\off$, thus attempting to detach this frame layer. Later we will show that such a frame layer cannot detach without actually being a rectangular frame layer.

At a high level, each consecutive layer that grows is able to do so because a previous layer has detected that it is not a rectangular layer, and with each additional layer more and more concavities  are ``filled in''.
(See Figures~\ref{fig:new-frame2-full-example-begin} and \ref{fig:new-frame2-full-example}.)
We will show that after some finite number of layers form, the last layer determines that it is a rectangle (due to the lack of a signal specifying otherwise) and propagates a detach message that eventually causes this last layer to disassociate. Here we see the use of disassociation to detect a global property of a frame layer (namely, that it is a rectangle). We will then use the fact that a rectangle has a single $\secon$ to ``elect a leader tile''. This leader election is described in Sections~\ref{sec:main-leader} and \ref{sec:leader-details}. Now, we present the assembly of additional layers.

\begin{figure}[htp]
    \centering
    \subfloat[Early stages of a full example showing the growth of the frame around $\alpha$.  Light grey: $\alpha$, Green: $\cgSE$, Light blue: standard path.]{%
    \includegraphics[width=2.4in]{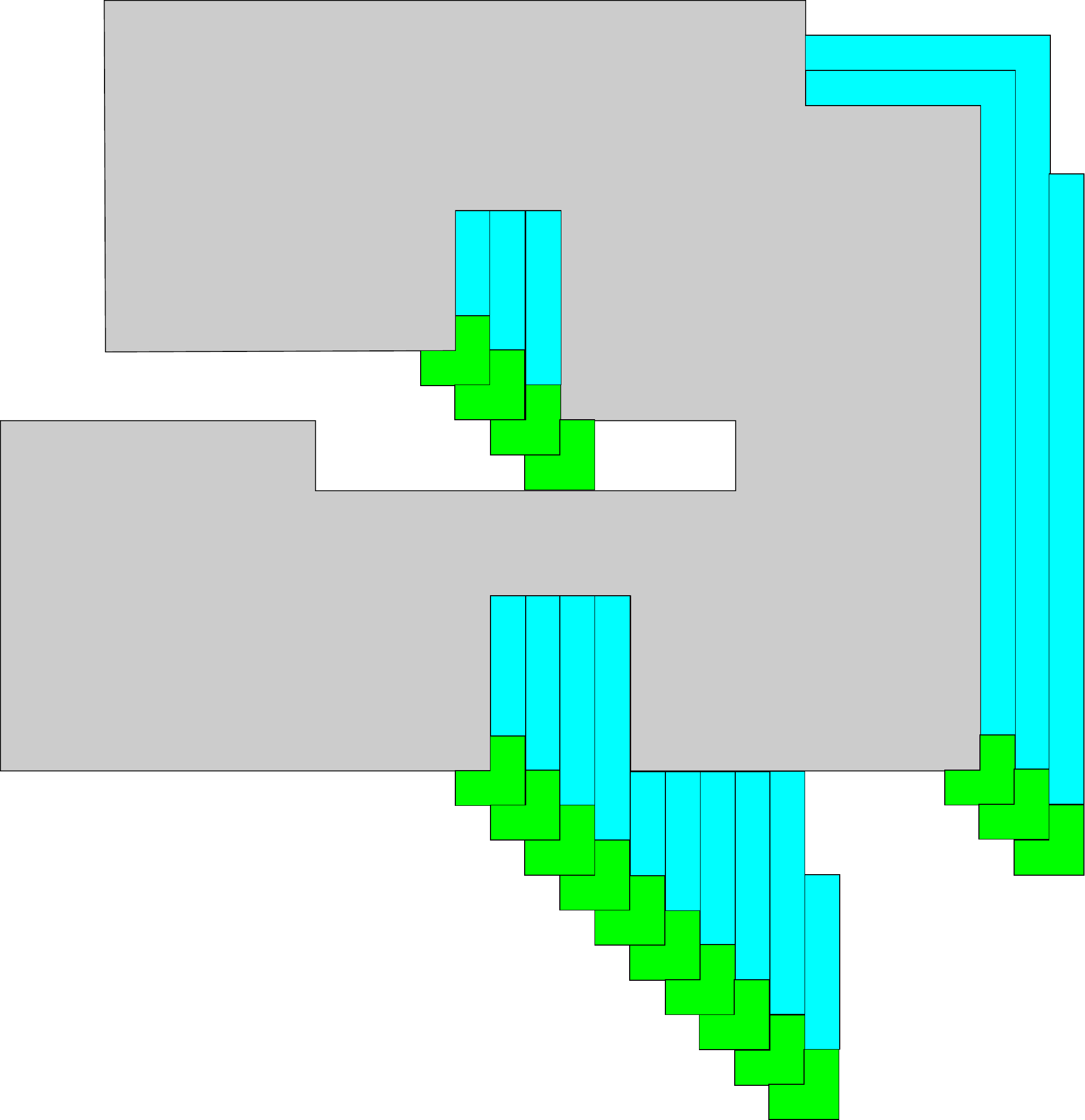}%
    \label{fig:new-frame2-full-example-begin}%
    }%
    \qquad\quad%
    \subfloat[The full frame for the example of Figure~\ref{fig:new-frame2-full-example-begin} showing the growth of the frame around $\alpha$ until it terminates in the first rectangular layer. Doubling rows are shown as dark blue.]{%
    \includegraphics[width=3.0in]{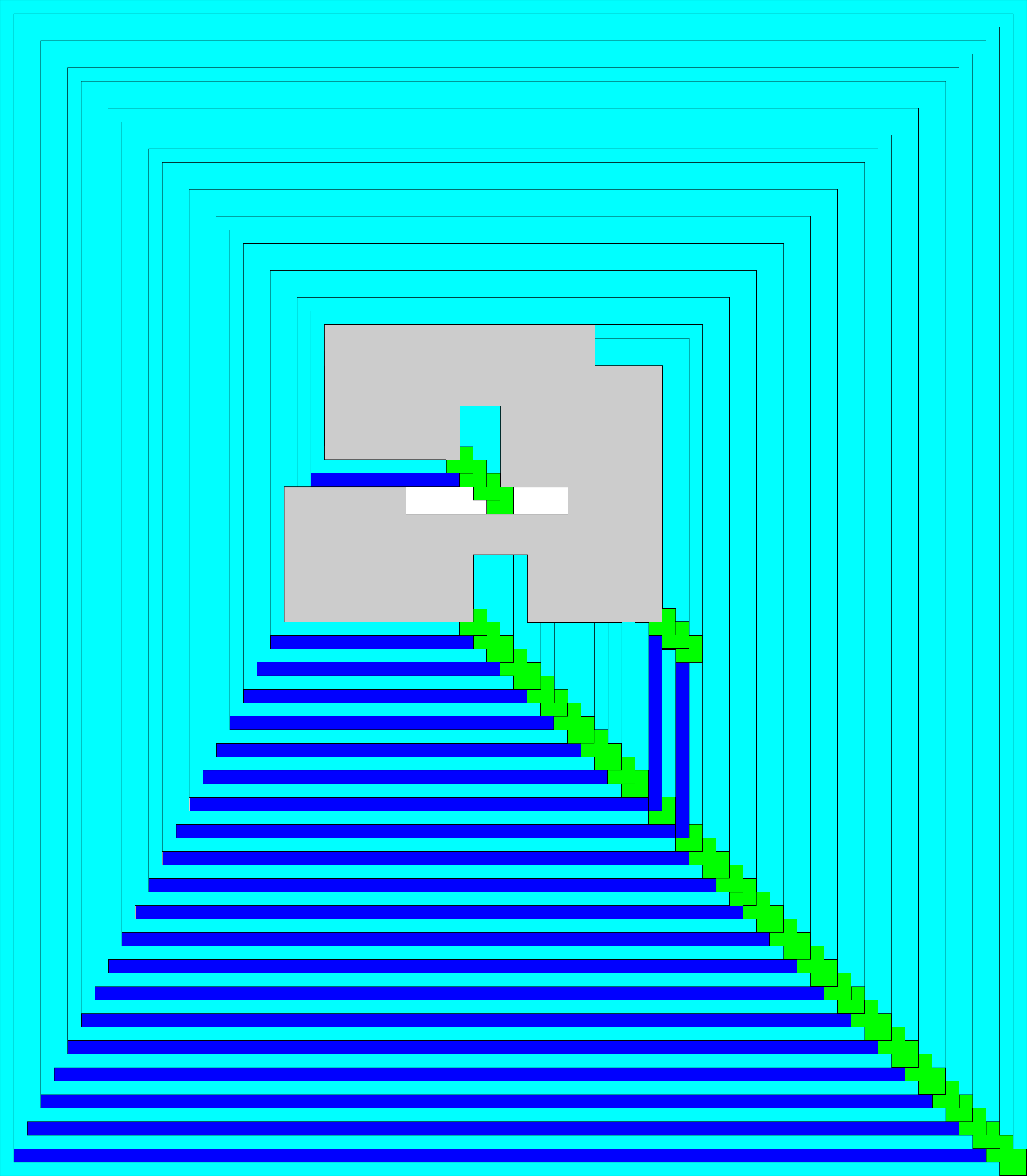}%
    \label{fig:new-frame2-full-example}%
    }%
    \caption{Frame growth examples}
    \label{fig:frame-grow-examples}
\end{figure}



\begin{enumerate}
	\item The first tiles that attach to layer 1 tiles may fall into the two following cases. Note that the first case always happens, and the second case may or may not, depending on the shape of $\alpha$.
	\begin{enumerate}
		\item The tile attachment to layer 1 consists of a 3-tile $\cgSE$ attaching to a $\cgSE$ that exposes $x2$ glues. The signals that pertain to this case are described in Figure~\ref{fig:SE-corner-gadget-tiles}.
		\item Some standard row terminated growth using $x$ glues as described in Case~\ref{item:layer1-terminates-special-case} above. 
        In this case, the tile attachment to layer 1 consists of a singleton tile attaching to a $\cgSE$ at location 6 in Figure~\ref{fig:secgCases} relative to this $\cgSE$ of layer 1. Such a tile is obtained by enhancing an appropriately rotated version of the layer 1 tile shown in Figure~\ref{fig:layer1-east-edge-crawler} with the signals and glues shown in Figure~\ref{fig:south-edge-crawling-tile}.
         For this enhancement, we do not duplicate the $f^\prime$ glue. 
        We rotate the tiles in  Figure~\ref{fig:layer1-east-edge-crawler} and Figure~\ref{fig:south-edge-crawling-tile} so that the active $f^\prime$ glues belong to the same edge. We obtain the a new tile by including all of the signals and glues of both tiles without duplicating glues that the two tiles have in common. For example, the resulting tile will have a single active $f^\prime$ glue and a single latent $f$ glue. 
	\end{enumerate}
	
\begin{figure}[htp]
\centering
    \includegraphics[width=4.5in]{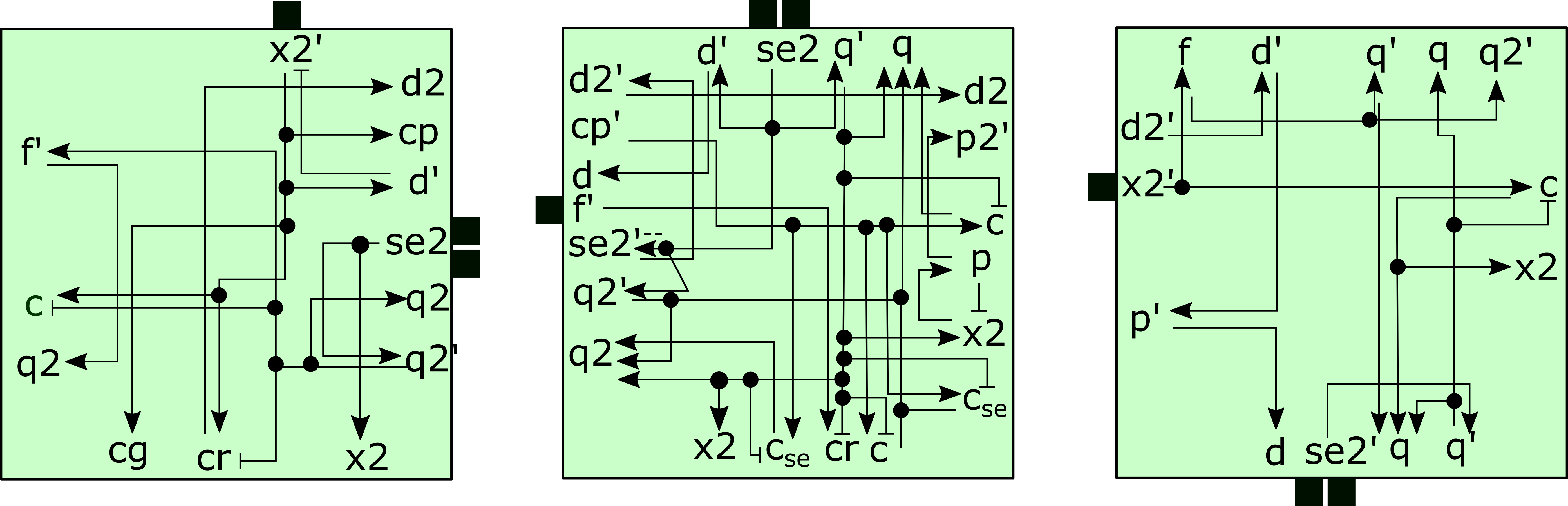}
    \caption{Signals and glues for the 3 tiles that make up a layer 2 $\cgSE$. 
    Here, we again only depict some of the signals of these tiles which correspond to the $d$, $f$, $q2$, $q$.
     We note that these four signals are present on the layer 1 tiles as well.
    Later, additional glues and signals will be added.}\label{fig:SE-corner-gadget-tiles}
\end{figure}


\begin{figure}[htp]
\centering
\subfloat[Signals and glues for a tile that grows standard paths along the south edge of a frame or $\alpha$.
 This includes the glues and transitions necesseary to carry out the $f$, $d$, $q$ signals.
 Note the $sr$ glue, which allows for a detachment duple to attach to this tile and a $\cgSE$]{%
 \label{fig:south-edge-crawling-tile}
\makebox[2in][c]{\includegraphics[width=1.5in]{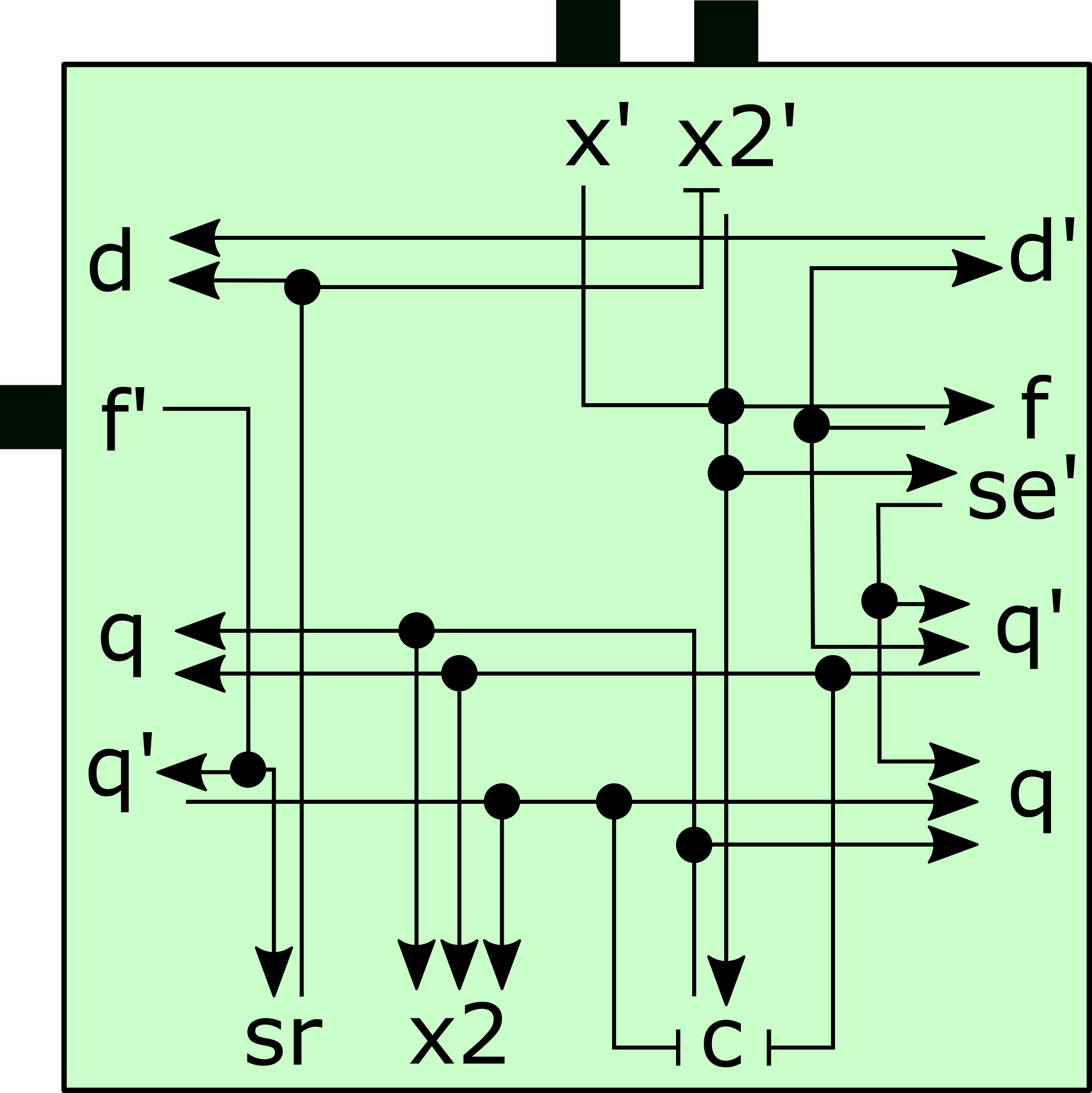}}
}
\qquad%
\subfloat[A frame tile which grows along a northern edge and which has been enhanced to detect collisions via collision detection tiles (see Figure~\ref{fig:collision-detection-tiles}) and pass the \emph{quit} (i.e. $q$) message as needed.
  Also, once it is able to determine it is on a quitting path (via attachment to the $c$ glue or reception of a $q$ message), it deactivates its collision detection $c$ glue and activates the $x2$ glue which will allow another layer to grow over it.
   This tile is also enhanced to pass the \emph{detach} (i.e. $d$) message CW. Note that this tile deactivates the left $x2'$ glue when the $d$ message is received, and also that the $d$ and $d'$ glues are strength-2 glues.\label{fig:new-frame2-tile}]{%
\makebox[3.8in][c]{\includegraphics[width=1.5in]{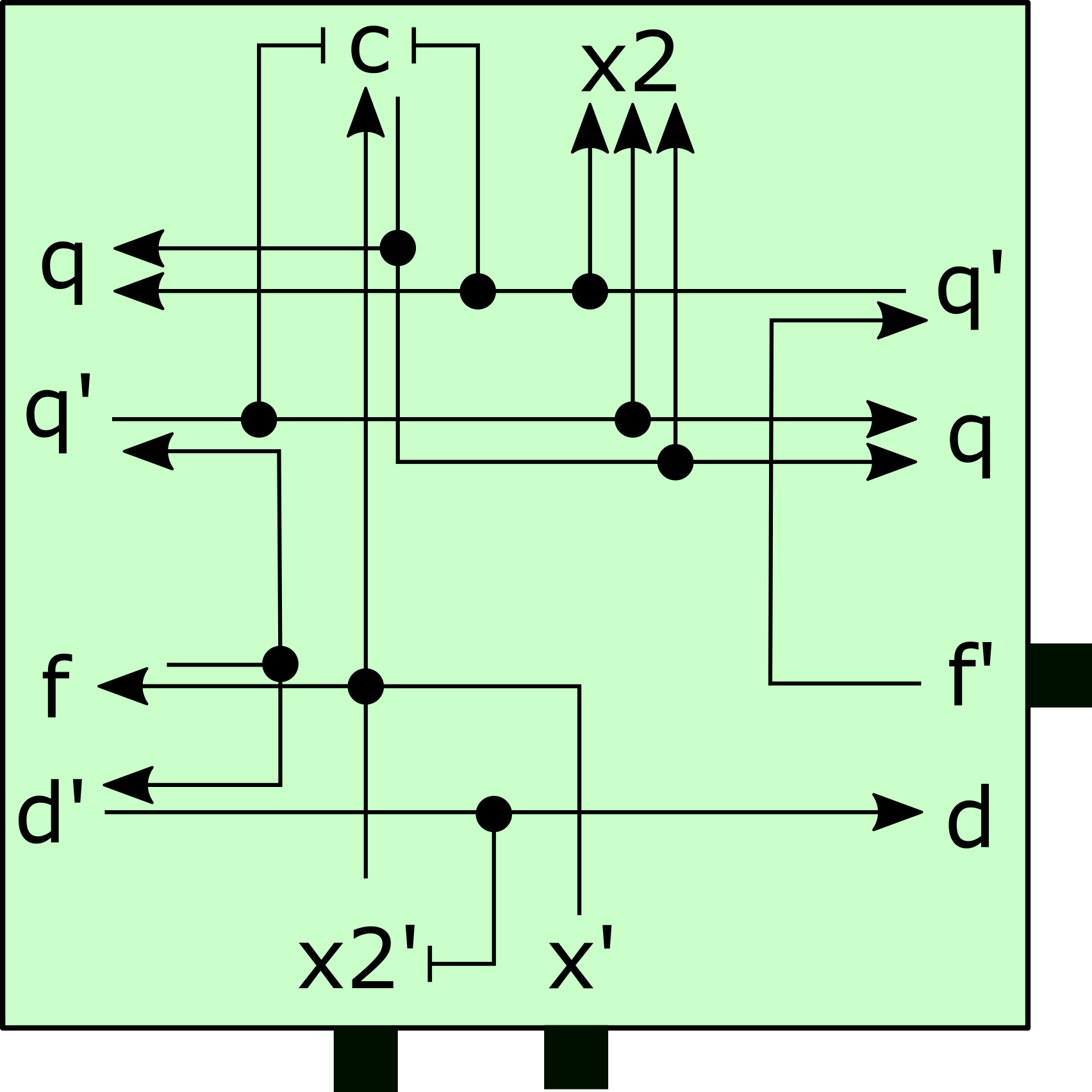}}%
}
\caption{Growth tiles.}
\end{figure}

    \item From each attached layer 2 tile, a single-tile-wide path grows CCW with each tile binding to the $x2$ glue of a previous layer or $x$ glue on $\alpha$ and the output $f$ glue of the immediately preceding tile. As in the assembly of layer 1, we call this the formation of a \emph{standard path} and use analogous notions of back, front, left and right sides. Hence, standard paths grow along straight edges of either $\alpha$ or previous layers that have exposed $x2$ glues on their right sides, with each tile using an $x'$ or $x2'$ glue as input on its left side, and an $f'$ glue as input on its back side.  As each tile binds, it activates an $f$ glue on its front side, thus allowing the path to continue along straight edges. The glues and signals pertaining to the case where a standard path tile binds to an $x2$ glue are described in Figures~\ref{fig:south-edge-crawling-tile} and~\ref{fig:new-frame2-tile}.


	\item When arriving at a NE, NW, or SW convex corner that exposes two $x2$ glues, a simple 2-tile corner gadget (denoted by $\cgNE$, $\cgNW$, and $\cgSW$ respectively) allows for continued growth around the convex corner while continuing to pass the $f$ message. The signals for a $\cgNE$ are described in Figure~\ref{fig:NE-corner-gadget-tiles} and those for $\cgNW$ and $\cgSW$ gadgets are similar but rotated appropriately.

\begin{figure}[htp]
\centering
    \includegraphics[width=4in]{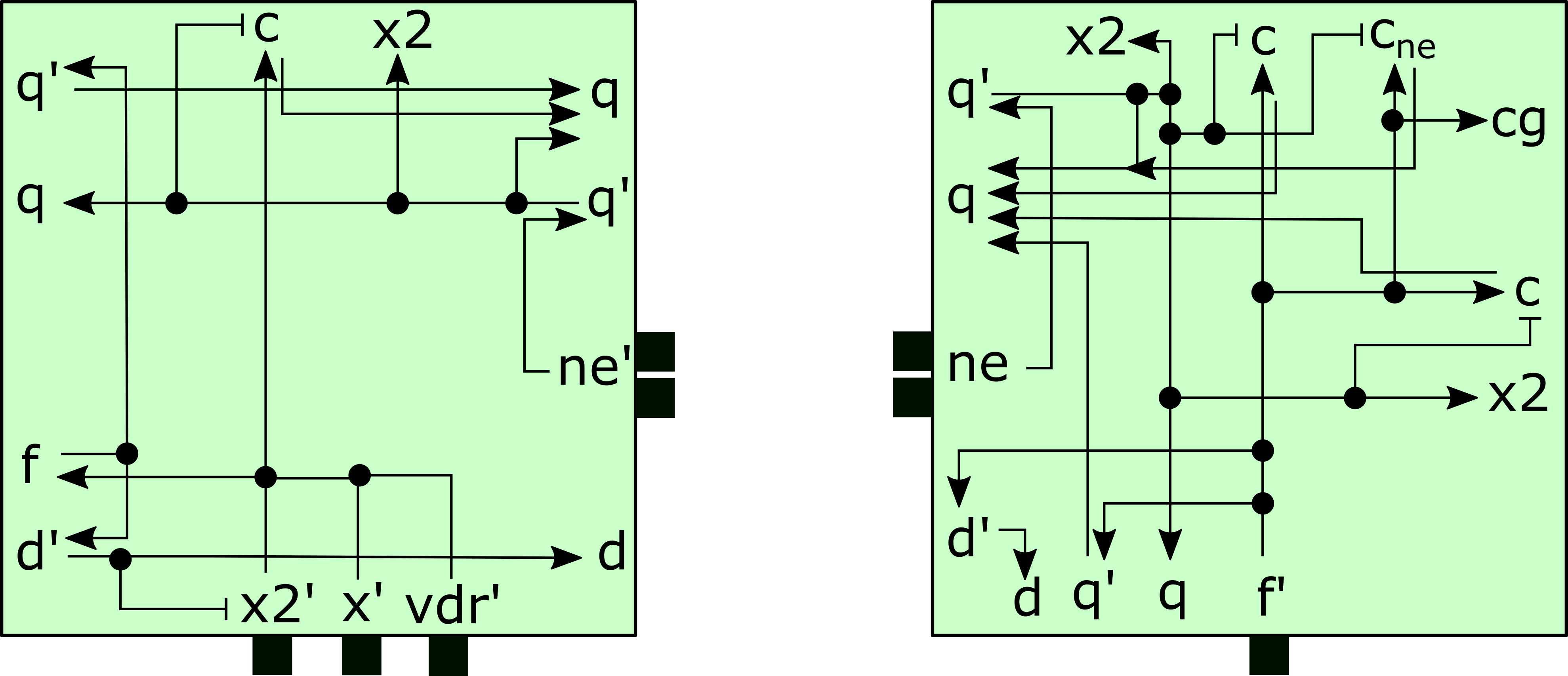}
    \caption{The 2 tiles that make up a $\cgNE$ that can bind to $x2$, $x$, and $vdr$ glues along a path initiated by an $f$ glue.
    Note, the $vdr'$ is unique to the $\cgNE$ to be able to attach to vertical doubling rows.
    }\label{fig:NE-corner-gadget-tiles}
\end{figure}

	\item As the tiles of a layer bind to $x2$ glues, they trigger $c$ glues to turn $\on$. Corner gadgets also trigger $c_{dir}$ glues to turn $\on$, where $dir$ can be $se$, $sw$, $ne$ or $nw$. During the assembly of a layer, if any of the collision detection singleton tiles or duples depicted in Figures~\ref{fig:collision-detection-tiles},~\ref{fig:collision-detection-corner-cases}, or~\ref{fig:standard-row-nonquitting-cgse} bind, then the frame layer, containing the tile to which they bound, cannot be rectangular. These tiles are considered the \emph{collision detection} tiles. For examples of various cases where one of these singleton or duple tiles bind, see Figures~\ref{fig:new-frame2-concave-collision1},~\ref{fig:new-frame2-duple-detect-collision},~\ref{fig:new-frame2-single-detect-collision}, and~\ref{fig:special-collision-detection}.

\begin{figure}[htp]
\centering
    \includegraphics[width=3.5in]{images/collision-detection-tiles.png}
    \caption{The tiles which, individually or after forming duples, are able to bind to any frame layer which  is not rectangular.  Note that there are also analogous versions of all of these tiles where the $x2'$ glues are replaced by $x'$ glues so that collisions between frame layers and $\alpha$ as well as between frame layers can be detected. The analogous versions of these tiles are not depicted. They are obtained by replacing all of the $x2^\prime$ glues in this figure with $x^\prime$ glues.}\label{fig:collision-detection-tiles}
\end{figure}



\begin{figure}[htp]
\centering
\subfloat[While the collision detection duples/tiles of Figure~\ref{fig:collision-detection-tiles} detect collisions with $\alpha$ and frame layers that present $x2$ glues, the tiles depicted here detect collisions of frame layers with other frame layers using $c$ and $c_{dir}$ glues where $dir$ is either $se$, $sw$, $ne$, or $nw$. See Figure~\ref{fig:special-collision-detection} for examples of how these detection duples/tiles are used.\label{fig:collision-detection-corner-cases}]{%
\makebox[4in][c]{\includegraphics[width=3.5in]{images/collision-detection-corner-cases.png}}
}%
\qquad%
\subfloat[Tiles which detects the case where a standard row terminates before overlapping with $\cgNE$, $\cgNW$, and $\cgSW$ gadgets. The analog for this case in a $\cgSE$ causes a vertical doubling row to form; its analog of this tile is shown in Figure~\ref{fig:vdr-tiles1}.\label{fig:standard-row-nonquitting-cgse}]{%
\makebox[1.8in][c]{%
    \includegraphics[width=.85in]{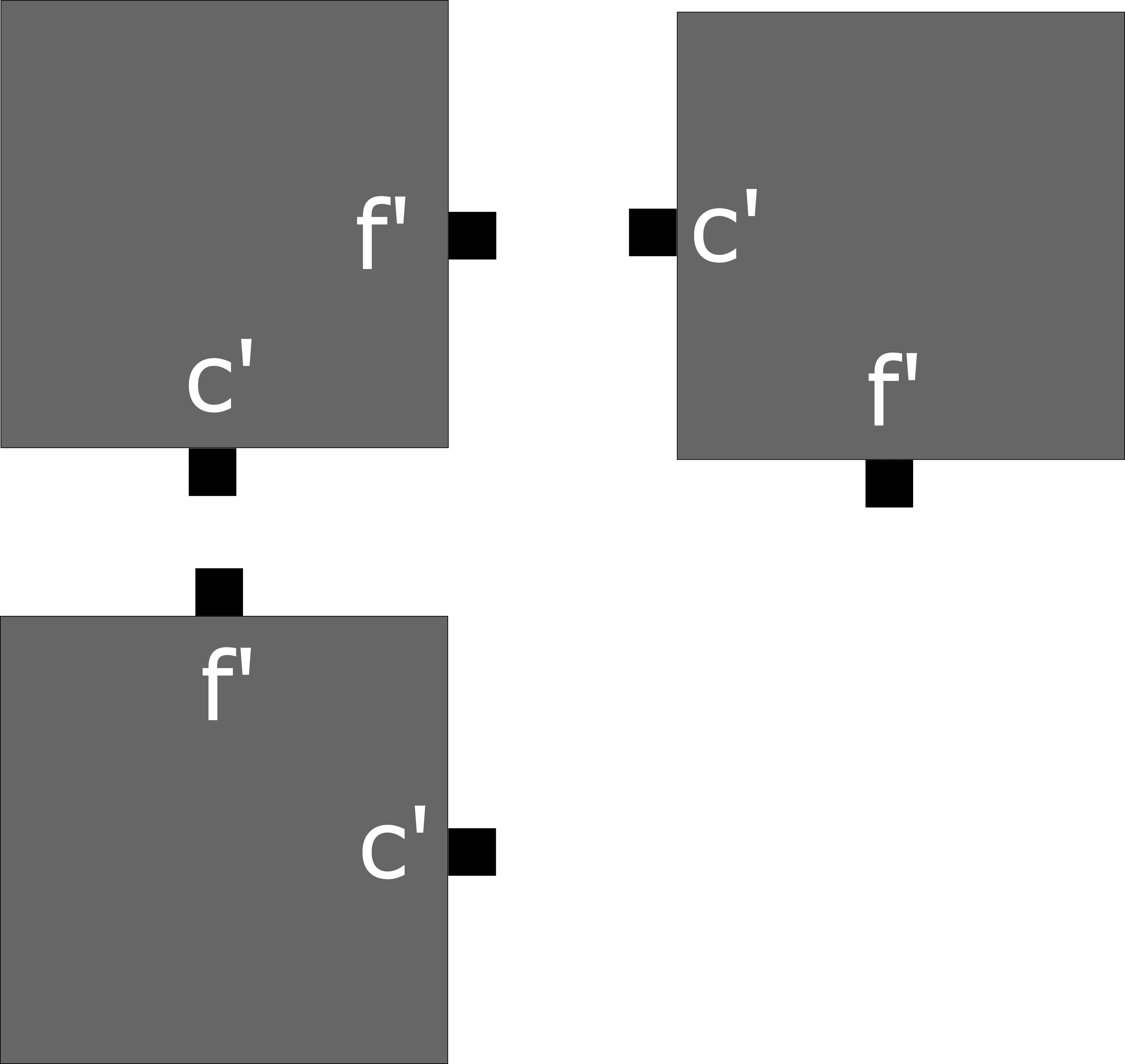}}
}
\caption{Collision detection tiles.}
\label{fig:collision-detection-multi}
\end{figure}


	\item Any tile belonging to a frame layer that binds via a $c$ or $c_{dir}$ glue begins the forward and backward propagation of a \emph{quit} ($q$) message.  It also turns off its $c$ glue to which the collision detection tile attached (causing it to detach) and turns on an $x2$ glue on that side, thus allowing for another frame layer tile to bind via this $x2$ glue. Every tile in the path which receives the $q$ message turns off the $c$ glue on its right side and turns on an $x2$ glue there, and continues the propagation of $q$ in the CW and CCW directions.

\begin{figure}[htp]
\centering
    \includegraphics[width=5.0in]{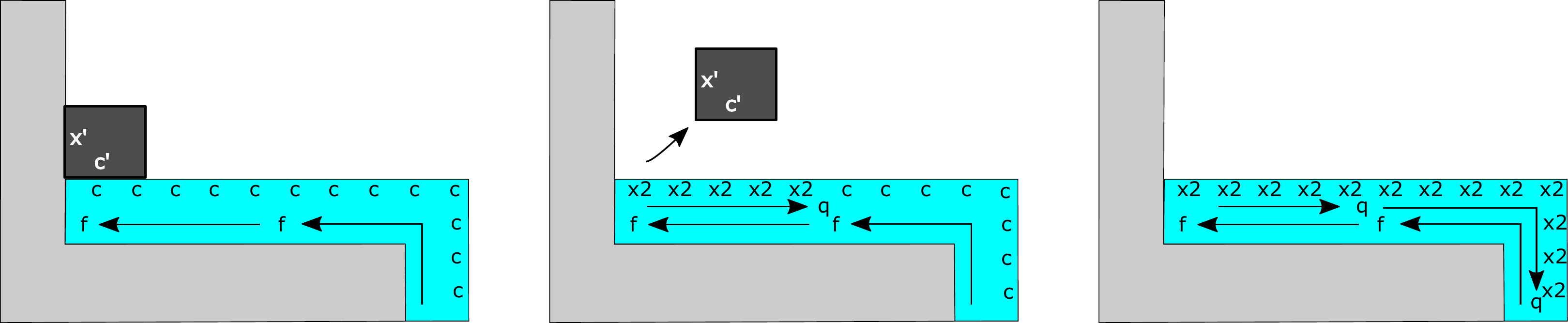}
    \caption{Depiction of a path colliding into a concavity, which is detected by the attachment of a collision detection tile.  Upon connecting, the collision detection tile initiates a $q$ message which causes all outward $c$ glues to deactivate and $x2$ glues to turn on.  This also causes the collision detection tile to fall off.}
    \label{fig:new-frame2-concave-collision1}
\end{figure}

\begin{figure}[htp]
    \centering
    \subfloat[An example of a collision detection by a collision detection duple (which attached, initiated the $q$ message, then detached).]{%
    \label{fig:new-frame2-duple-detect-collision}%
    \includegraphics[width=2.0in]{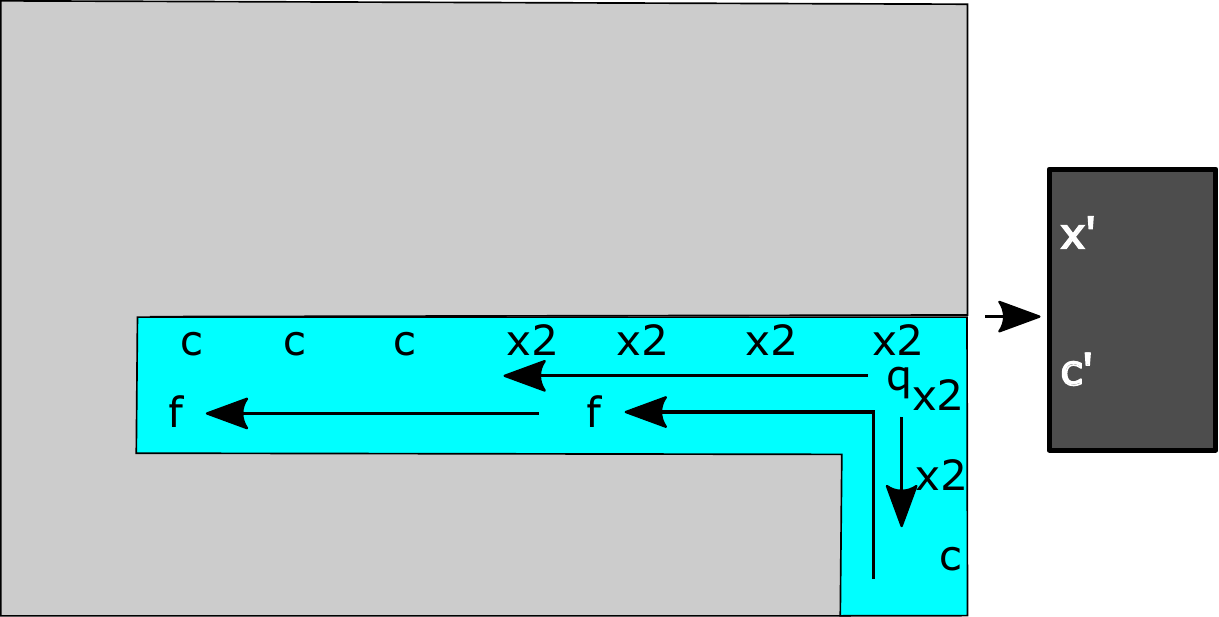}
    }%
    \qquad
    \subfloat[An example of a collision detection by a collision detection tile (which attached, initiated the $q$ message, then detached), while the path is still growing into a concavity.]{%
    \label{fig:new-frame2-single-detect-collision}%
    \includegraphics[width=2.0in]{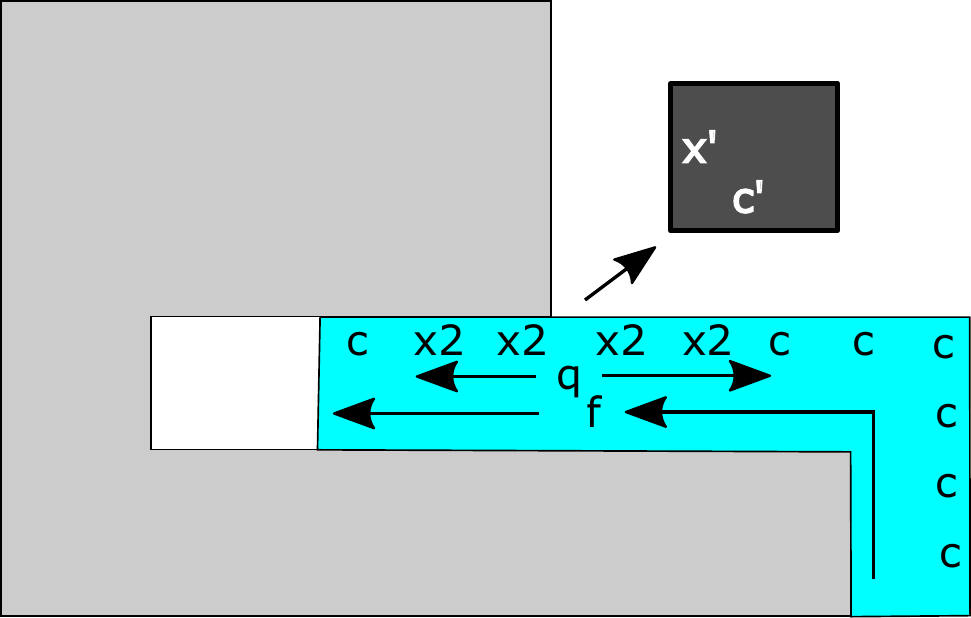}%
    }%
    \caption{Typical collision detection examples}
    \label{fig:col-detect}
\end{figure}

\begin{figure}[htp]
\centering
  \subfloat[][Collision detection enabled by a collision detection tile like the one in Figure~\ref{fig:standard-row-nonquitting-cgse}.]{%
        \label{fig:new-frame2-special-detect-collision}%
	        \includegraphics[width=1.5in]{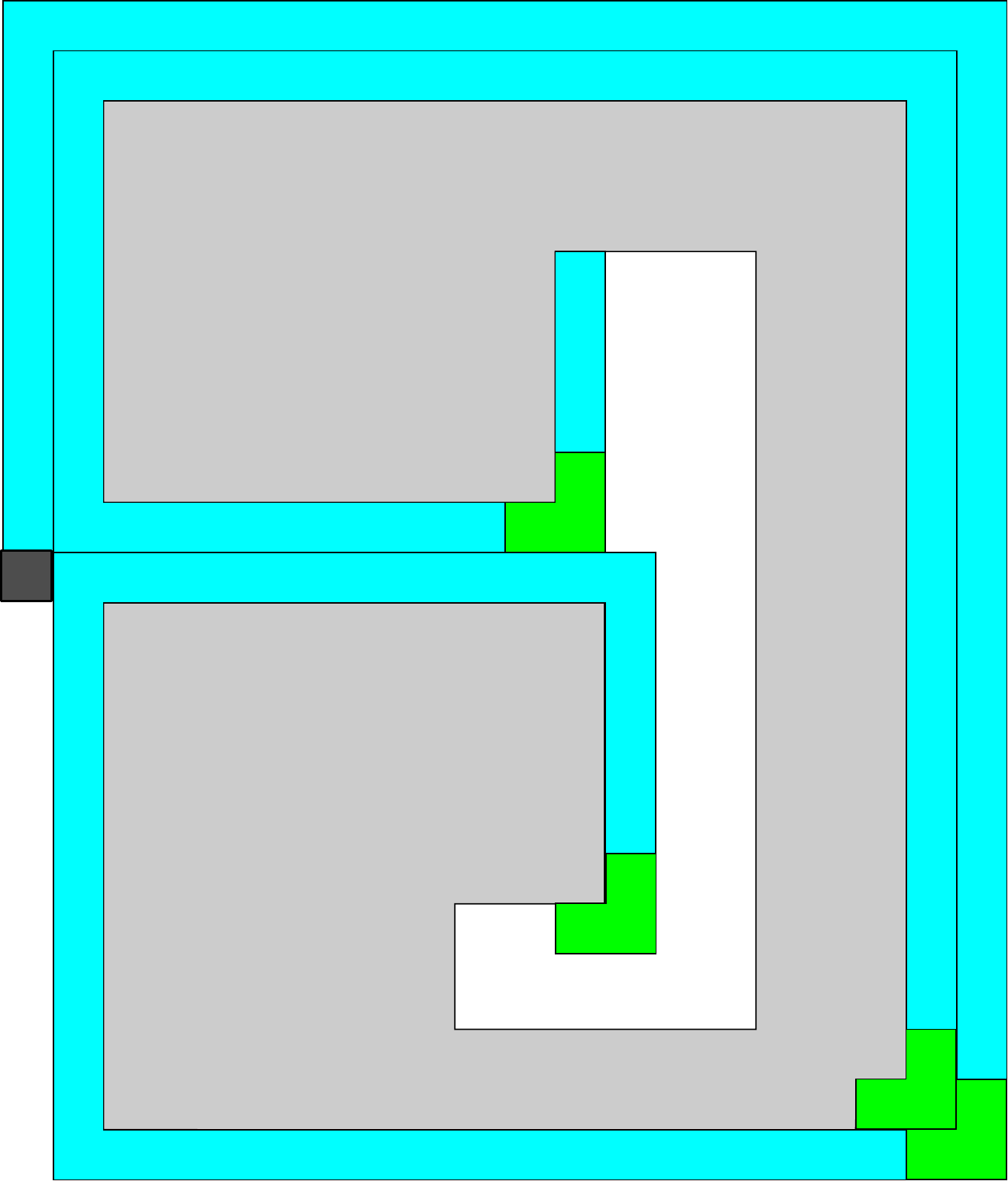}
        }%
        \quad\quad
  \subfloat[][Collision detection enabled by a collision detection duple which attaches to a $\cgNW$, similar to those shown in Figure~\ref{fig:collision-detection-corner-cases}.]{%
        \label{fig:new-frame2-special-detect-collision2}%
        		\includegraphics[width=1.5in]{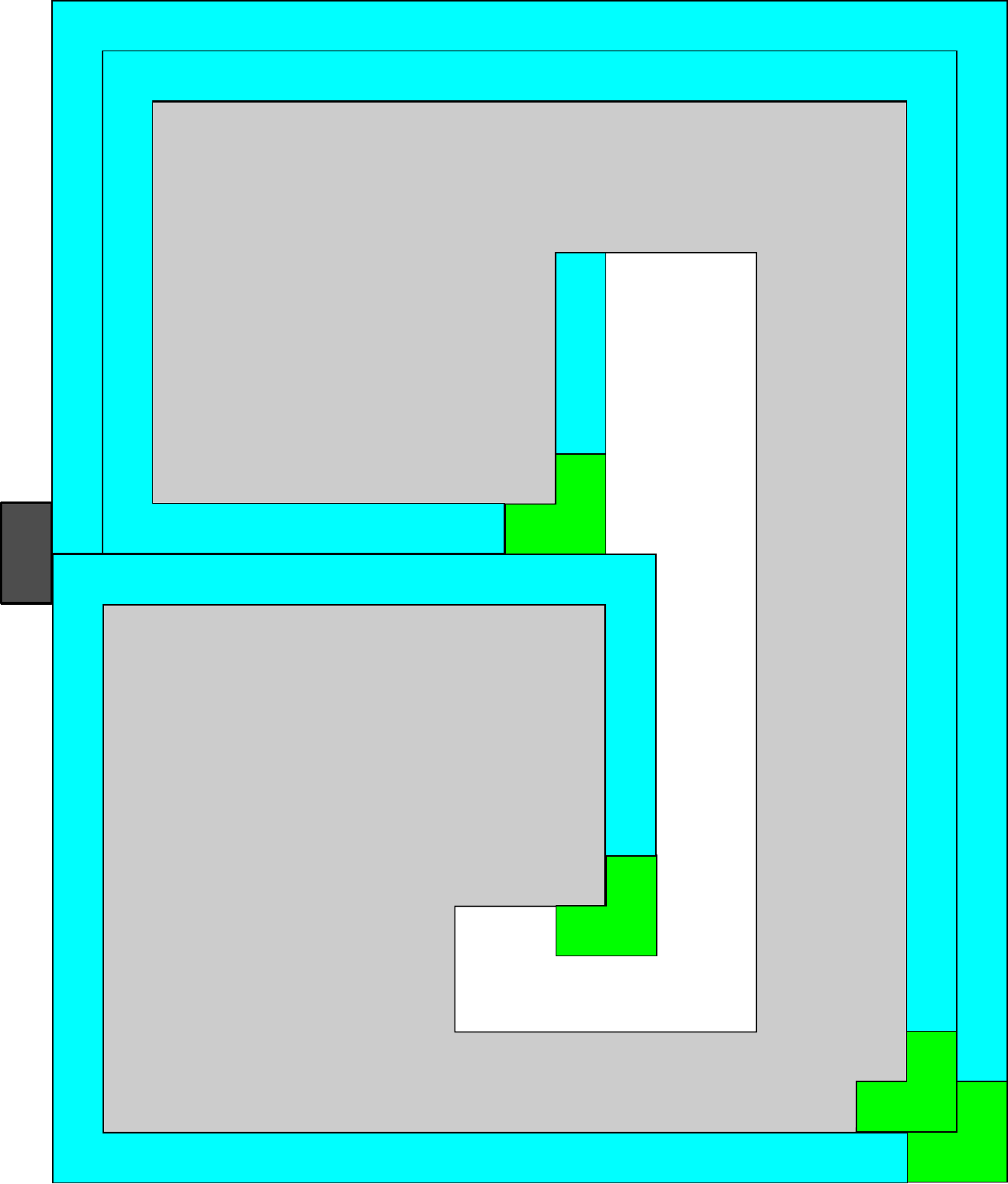}
        }%
    \qquad%
  \subfloat[A schematic depiction of the case where a doubling row forms.\label{fig:doubling-row-example}]{\includegraphics[width=1.5in]{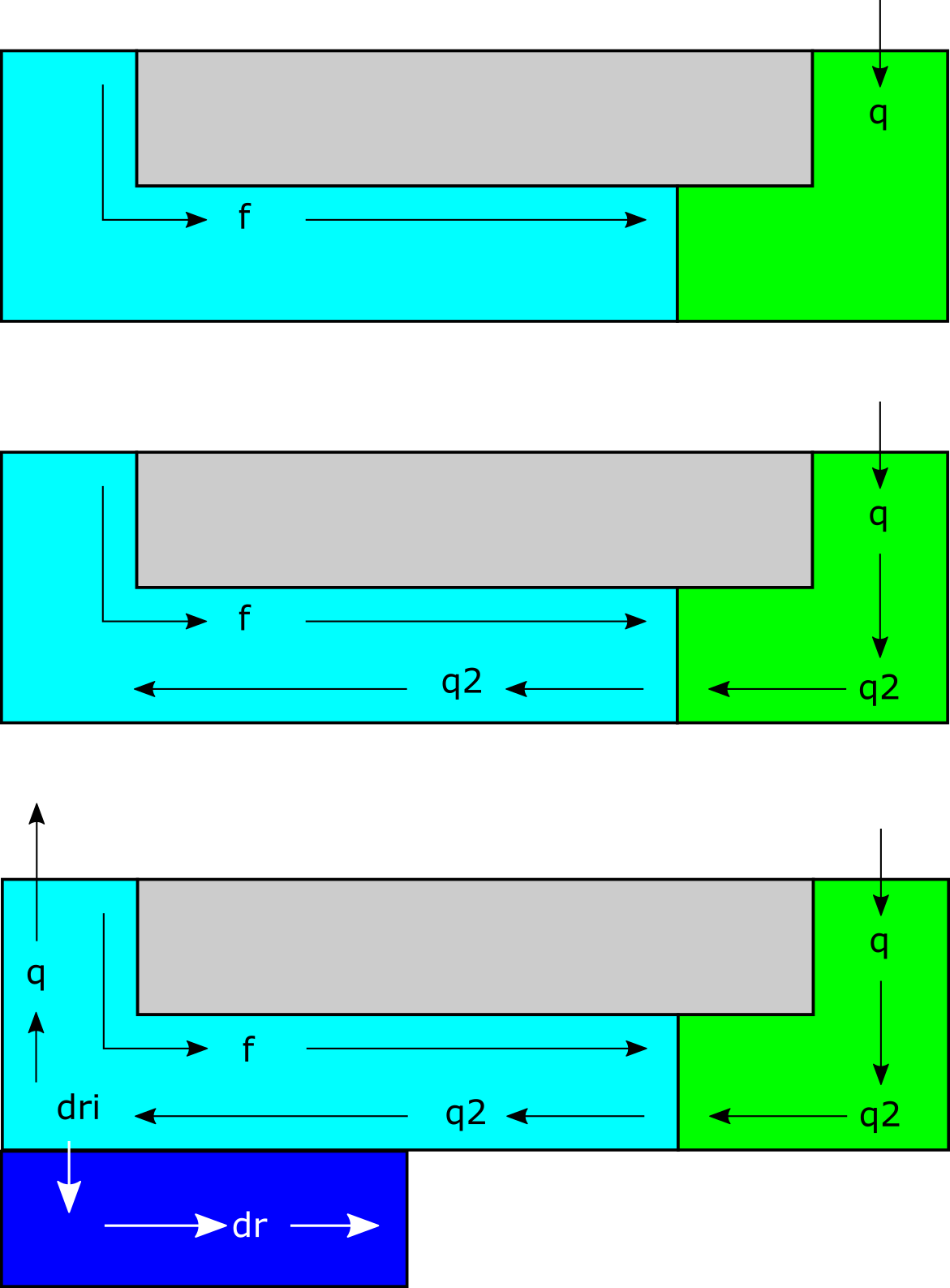}}%
  \caption{Examples of special cases of collision detection and doubling row formation.}
  \label{fig:special-collision-detection}
  
\end{figure}

	\item The $q$ message will then pass through all tiles on that path back to the $\cgSE$ which initiated the path.  
    That $\cgSE$ will now have the information that the path was not a rectangle (i.e. it encountered some concavity), 
    so it will activate glues which allow another $\cgSE$ to attach to its back and begin growth of a new layer.
    \update{
    If a standard path grows into the westernmost tile of the $\cgSE$, then this standard path will propagate the $q$ message via the $q2$ signal  
    to initiate growth of a \emph{doubling row} to its south. In Figure~\ref{fig:doubling-row-example}, a doubling row is shown in dark blue.
    The tiles which comprise a doubling row are shown in Figure~\ref{fig:doubling-row-tiles}.
     
    To initiate the growth of a doubling row, when a $\cgSW$ receives a $q2$ message it turns $\on$ a strength-2 $dri$ (\emph{doubling row initiator}) glue.
    This allows for the growth of a doubling row which will grow along the south border of the layer.
    This in effect makes the south of the layer 2 rows wide, aiding in the production of a rectangular frame.
    The effect of horizontal doubling rows are further explained in Section~\ref{sec:frame-proof}.}

\begin{figure}[htp]
\centering
    \subfloat[A depiction of the tiles and signals which enable the growth of a doubling row.]{%
     \label{fig:doubling-row-tiles}    
        \includegraphics[width=2.0in]{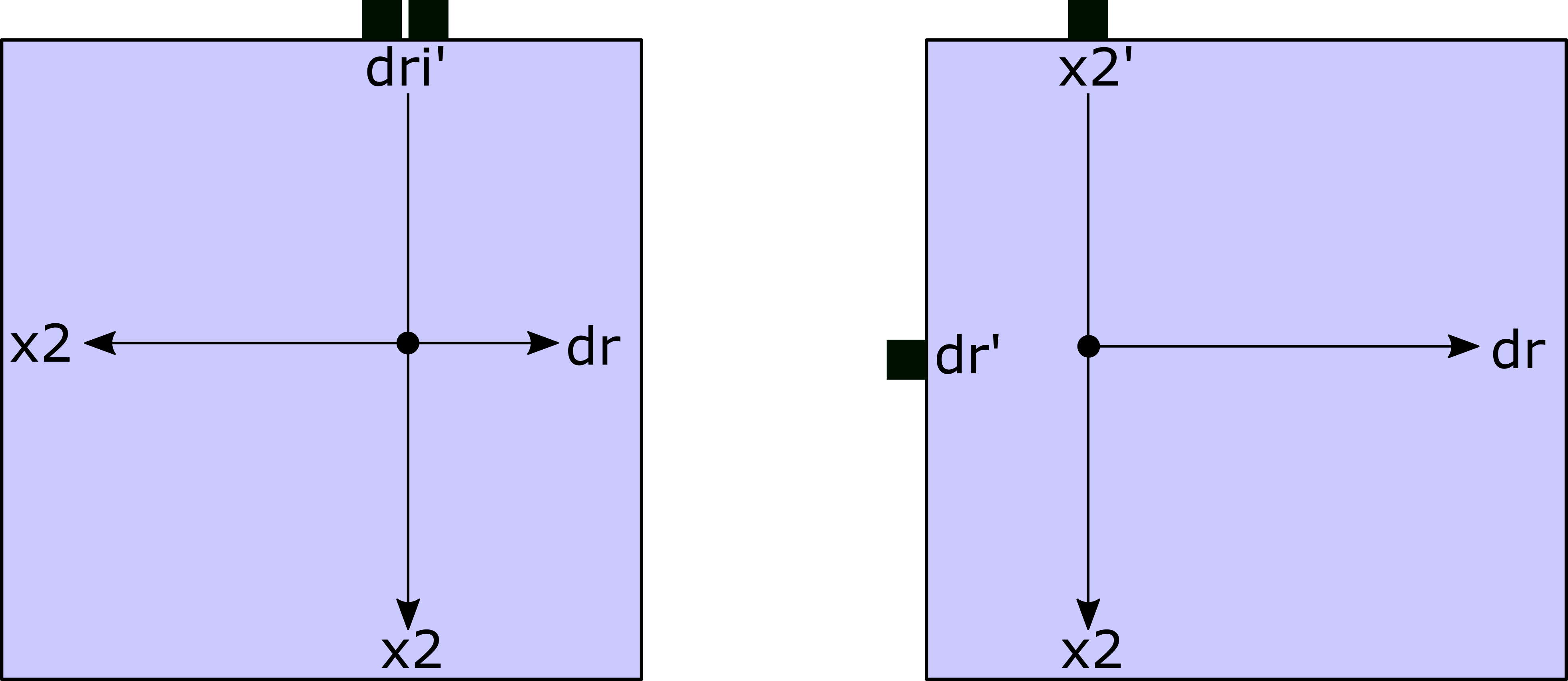}
    }
    \qquad
    \subfloat[The tiles which are part of a $\cgSW$ gadget. 
              Note that the $q2$ glue must activate the doubling row before the $q$ signal is sent north.
              This ensures that the doubling row attaches before any possible layer growing on top of it.]{%
    \label{fig:cgSW}    
    \includegraphics[width=2in]{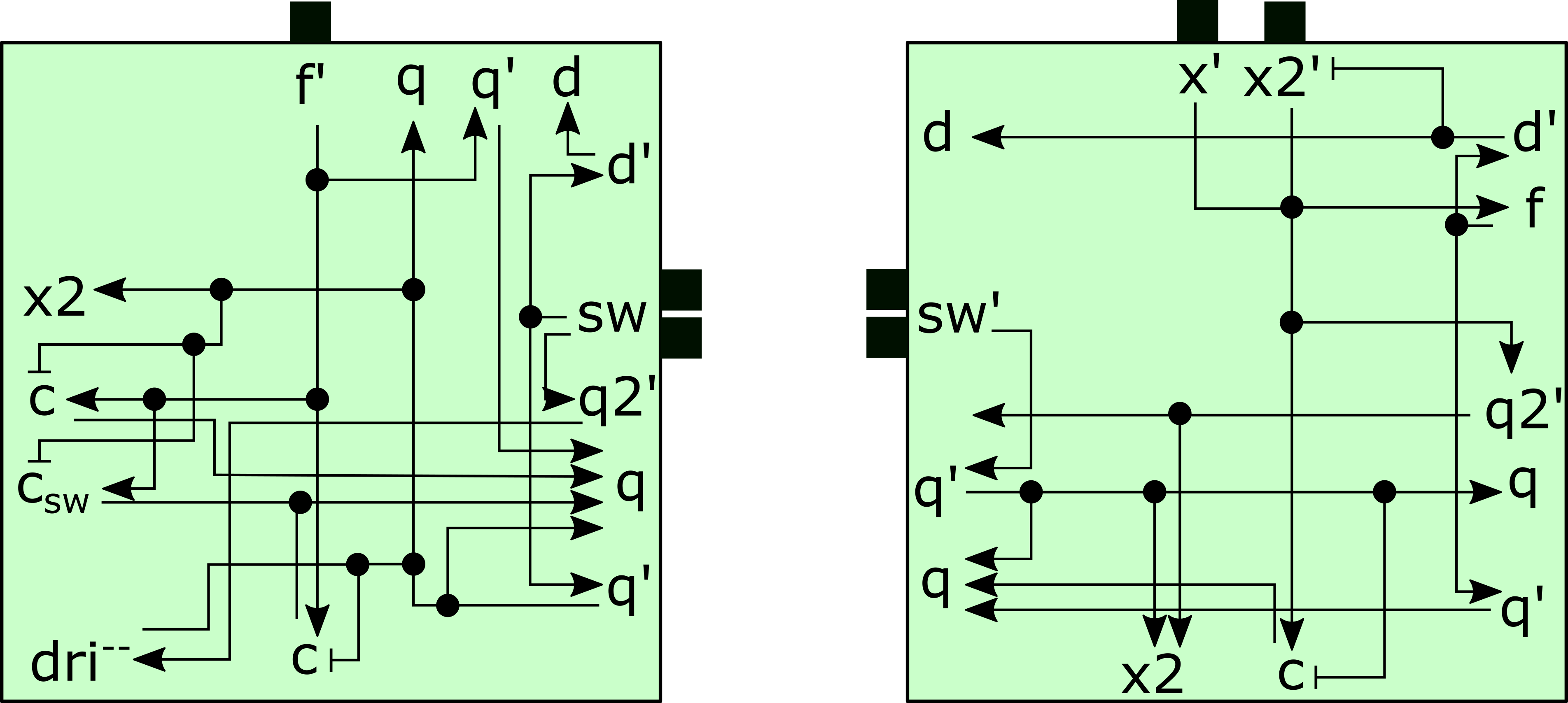}
    }
    \qquad
    \subfloat[The glues specific to passing the $q2$ message CW through south standard growth tiles.]{%
    \label{fig:q2-grow-signals}    
    \includegraphics[width=.85in]{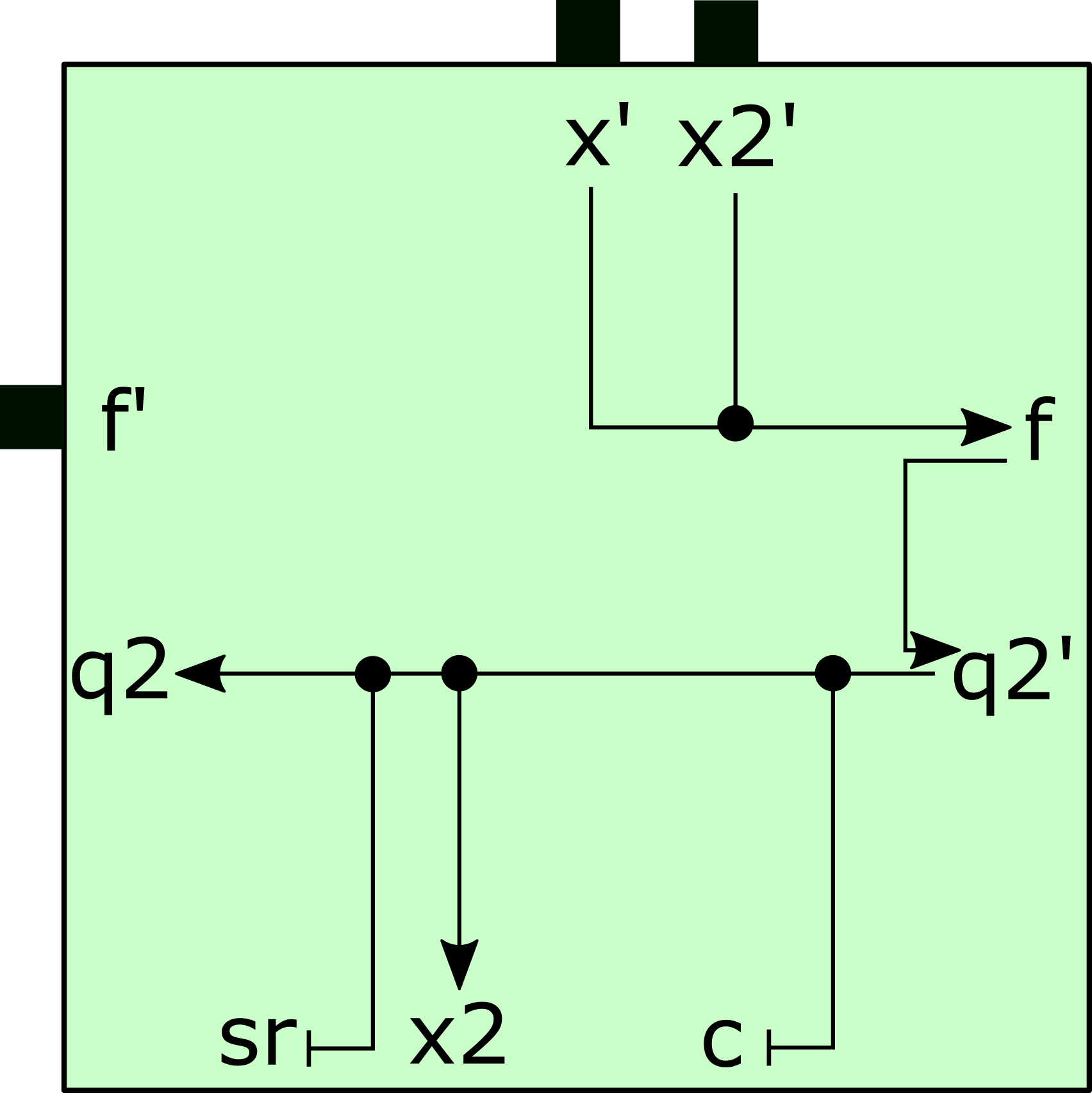}
    }
    \caption{Tiles involved in horizontal doubling row formation}
    
\end{figure}


	\item 
    \update{ 
        A \emph{vertical doubling row} is a layer of tiles which form to the east of some $\cgSE$ in two specific instances.
        First, if a row of tiles grows along eastern edges of some frame layer and places a tile at location 7 of some other $\cgSE$ in Figure~\ref{fig:secgCases},
        then a vertical doubling row forms as described in Figure~\ref{fig:vdr2}.
        Second, if a $\cgSE$ is `overtaken' by a growing southern standard path growing as seen in Figure~\ref{fig:vdr1}, a vertical doubling row forms.
        The tiles which initiate and generate a vertical doubling row are provided in Figure~\ref{fig:vdr-base-tiles}, and the signals which allow for this in a $\cgSE$ are show in Figure~\ref{fig:vdr-growth}.
        The purpose of these vertical doubling rows are to allow for a stack of $\cgSE$ to overtake another to correct a situation where a frame is not rectangular.
        The effect of vertical doubling rows are further explained in Section~\ref{sec:frame-proof}.

    }

     \begin{figure}[htp]
        \centering
          \subfloat[][Tiles which are placed in location 3 (i.e., below south-western tile in $\cgSE$) as part of the vertical doubling rows. 
                        (left) This tile detects if the $\cgSE$ been overtaken by a growing standard path (i.e., case 2 as demonstrated in Figure~\ref{fig:vdr1}).
                        (right) This tile is placed if the $q2$ signal has been recieved from the north of the $\cgSE$ (i.e., case 1 as demonstrated in Figure~\ref{fig:vdr2})]{%
                \label{fig:vdr-tiles1}%
                \makebox[2.8in][c]{\includegraphics[width=2in]{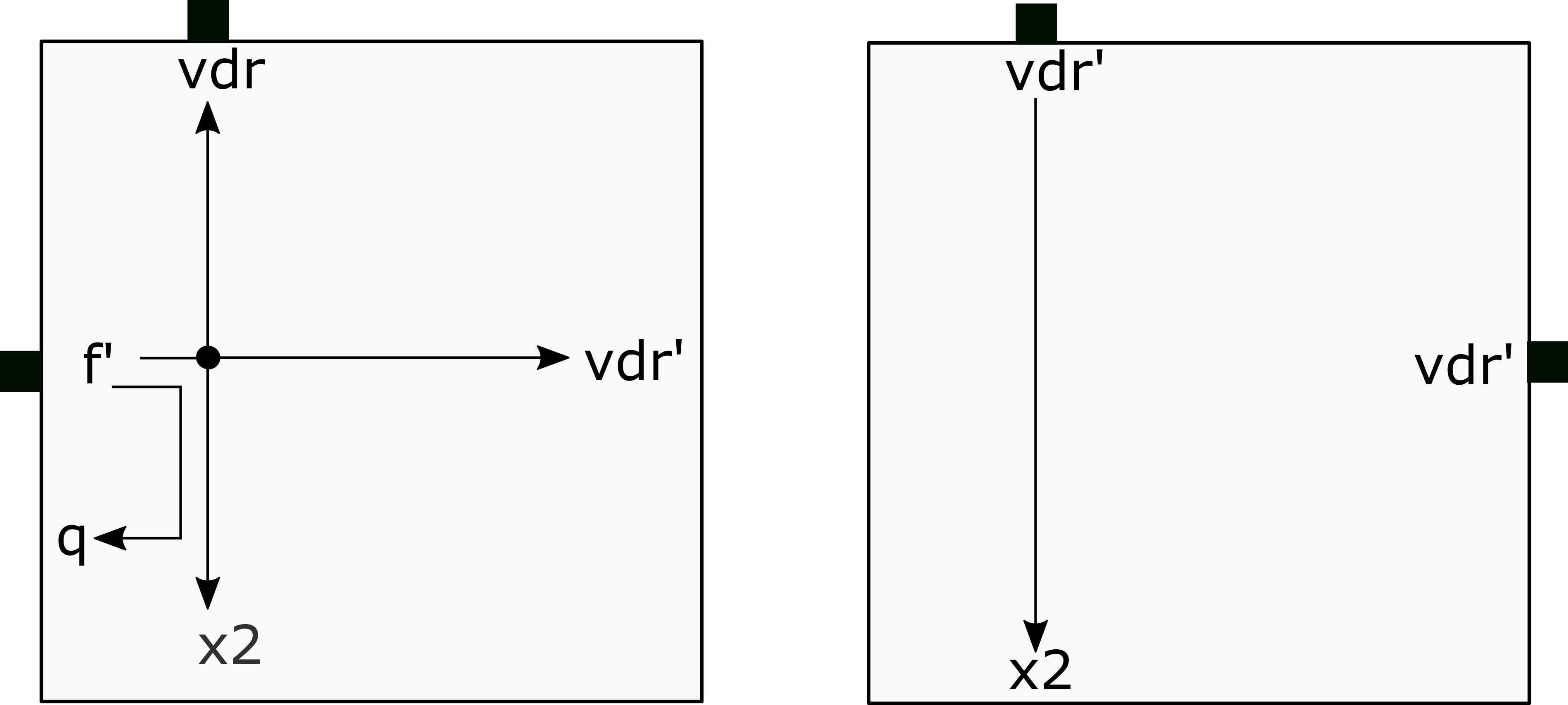}}
                }%
                \quad\quad
          \subfloat[][Tiles which grow the vertical doubling row.
            (left) Tile which attaches via strength 2 to $\cgSE$ initially.
            (center) Tile which allows for new $\cgSE$ to attach, attaches to prior tile via strength-2 glues.
            (right) Tile which grows the north vertical doubling row.]{%
                \label{fig:vdr-tiles2}%
                        \includegraphics[width=3in]{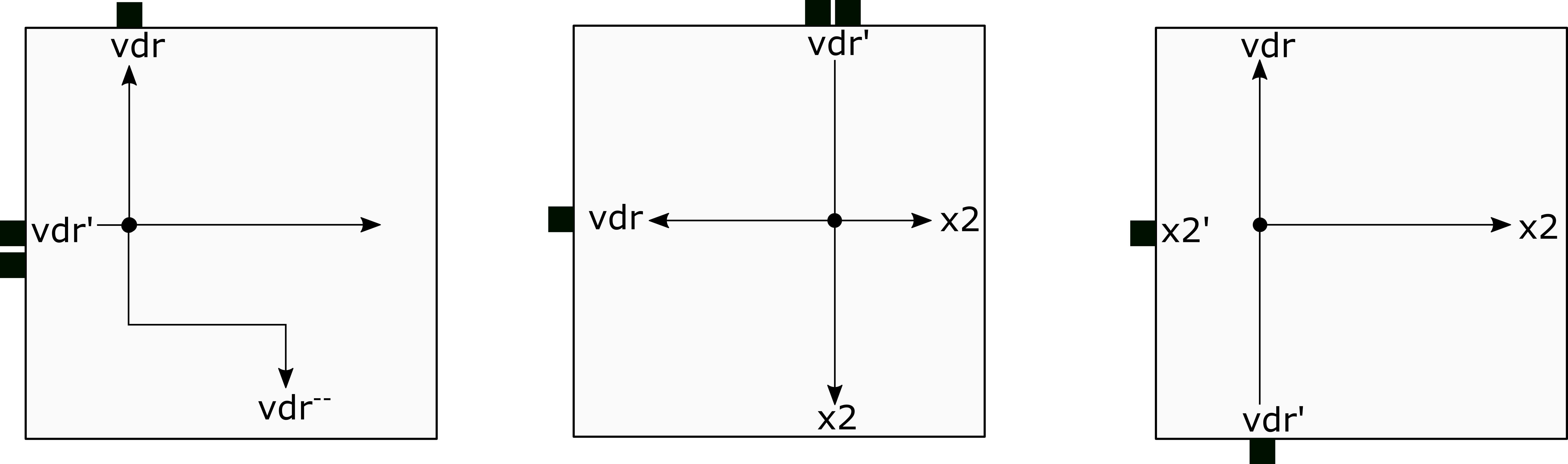}
                }%
          \caption{Vertical doubling row tiles.}
          \label{fig:vdr-base-tiles}
          
        \end{figure}

        \begin{figure}[htp]
            \centering
              \subfloat[][Signals added to a northern growing standard path tile to enable vertical doubling rows.]{%
                    \label{fig:vdr-growth}%
                    \includegraphics[width=1in]{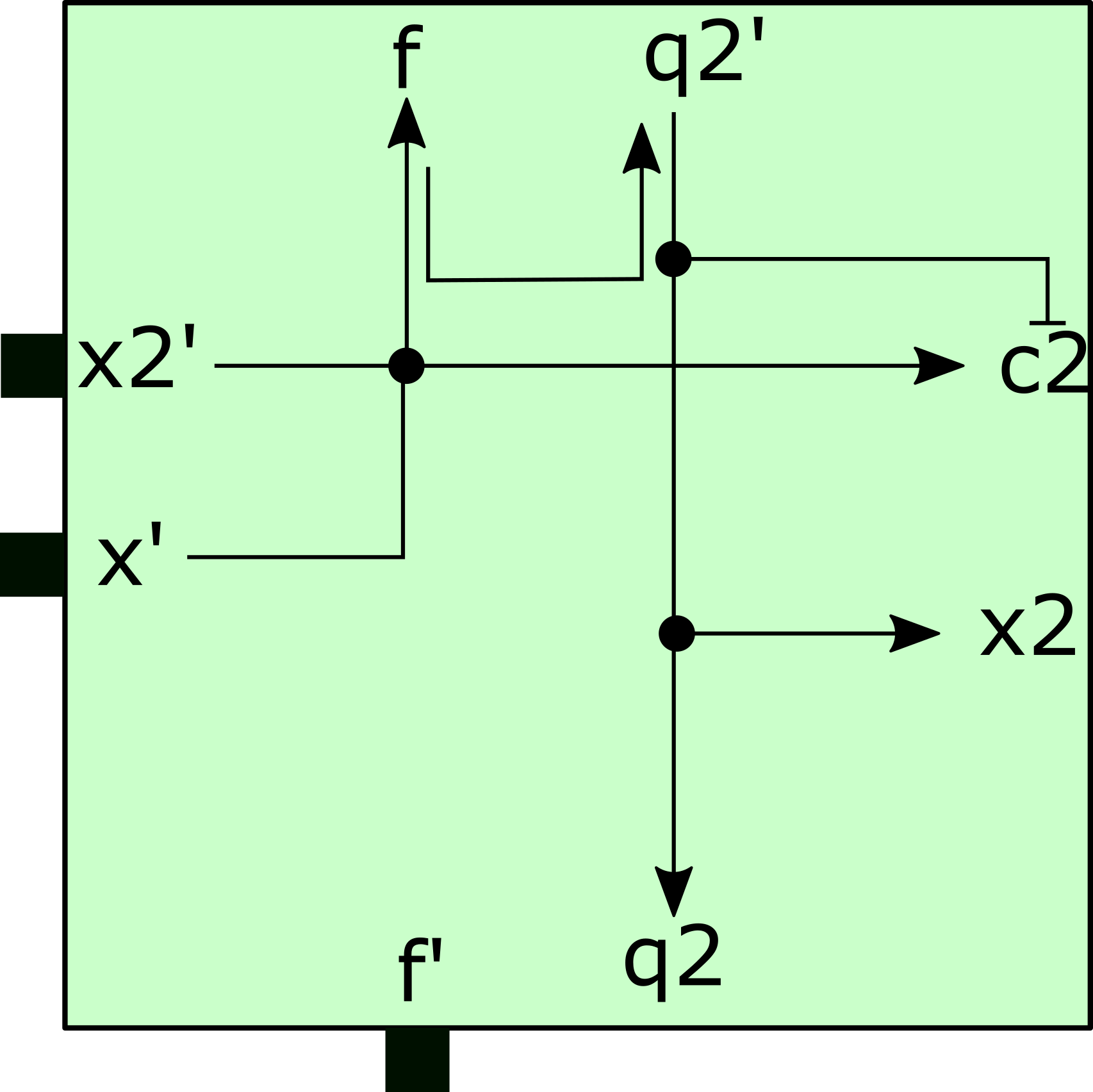}
                    }%
                    \quad\quad
              \subfloat[][$\cgSE$ signals which allow for the growth of the vertical doubling row.
                            Note that the $q$ signal is utilized to indicate a collision in case 2 of generating a vertical doubling row, to account for instances
                            such as shown in Figure~\ref{fig:collision-detection-corner-cases2}
                            This does not cause any spurious activation of signals which could cause the $\cgSE$ to activate its $x2$ glues,
                            as the strength-2 $vdr$ glue is must be bound (and would block placement of a $\cgSE$) to initiate the $q$ signal north]{%
                    \label{fig:vdr-cgSE}%
                    \makebox[4in][c]{\includegraphics[width=3in]{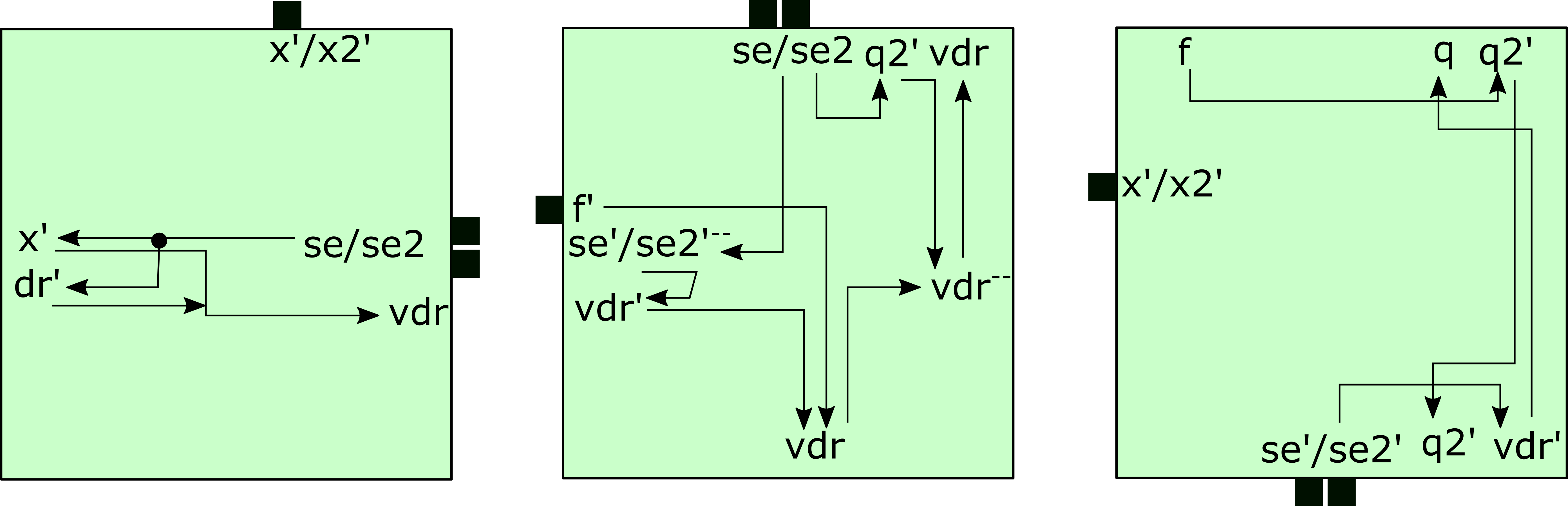}}
                    }%
              \caption{Signals to enable vertical doubling rows}
              \label{fig:vdr-frame-tiles}
              
            \end{figure}

\begin{figure}[htp]
\centering
    \includegraphics[width=4.5in]{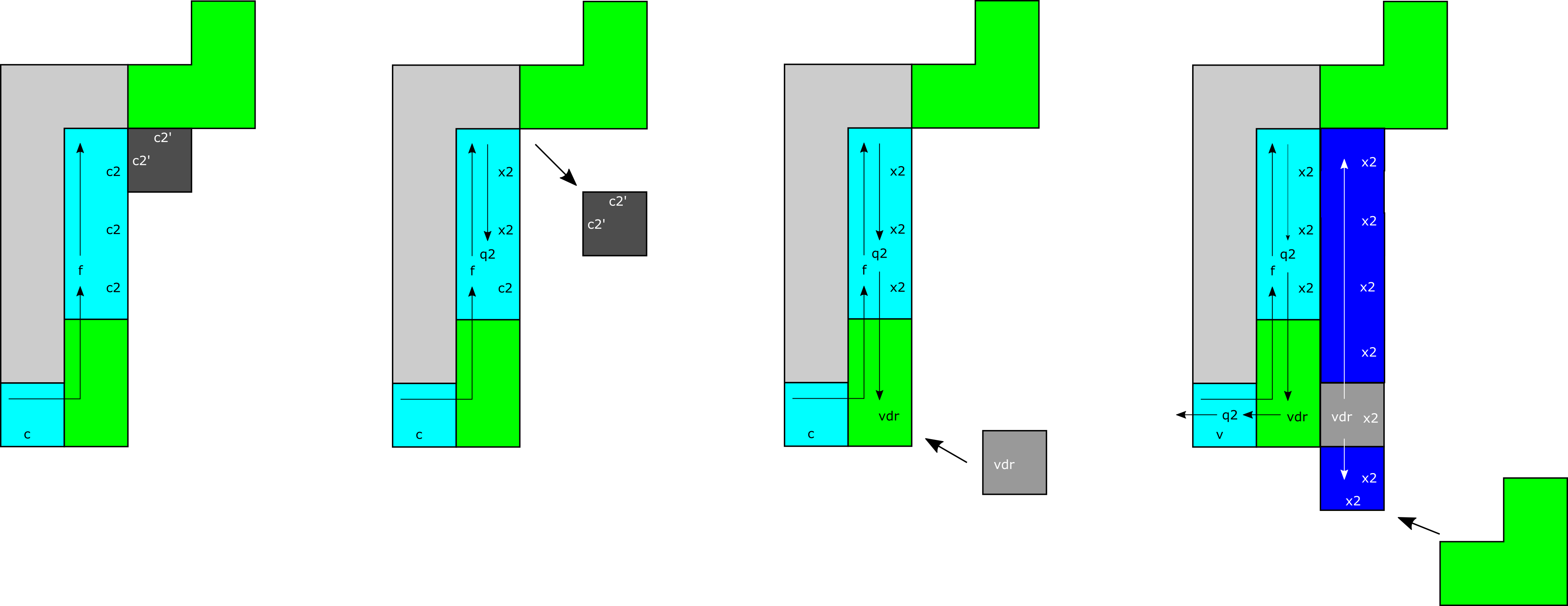}
    \caption{First case of vertical doubling row. When a standard growth places a tile in location 7, a singleton tile collision tile detects this by cooperatively 
             binding to $c2$ glues with which we enhance our north frame growth and $\cgSE$ tiles. 
             The binding of this singleton tile initiates a $q2$ message that is passed CW through a layer that turns $\off$ $c$ and $c2$ glues (allowing for 
             the detachment of $ch$) and turns $\on$ $x2$ glues on east edges. When this message is received by the corner tile of the 
             first $\cgSE$ that it reaches; instead of activating an $x2$ glue on the east edge, it activates a $vdr$ glue. 
             $vdr$ initiates the growth of a vertical doubling row (shown as the blue column) by propagated a $vdr$ message upward. 
             The $vdr$ message is also propagated downward allowing for the attachment of a singleton tile. This singleton tile activates south 
             and east $x2$ glues that allow for a $\cgSE$ to attach.}
    \label{fig:vdr2}
\end{figure}

\begin{figure}[htp]
    \centering
        \includegraphics[width=5.0in]{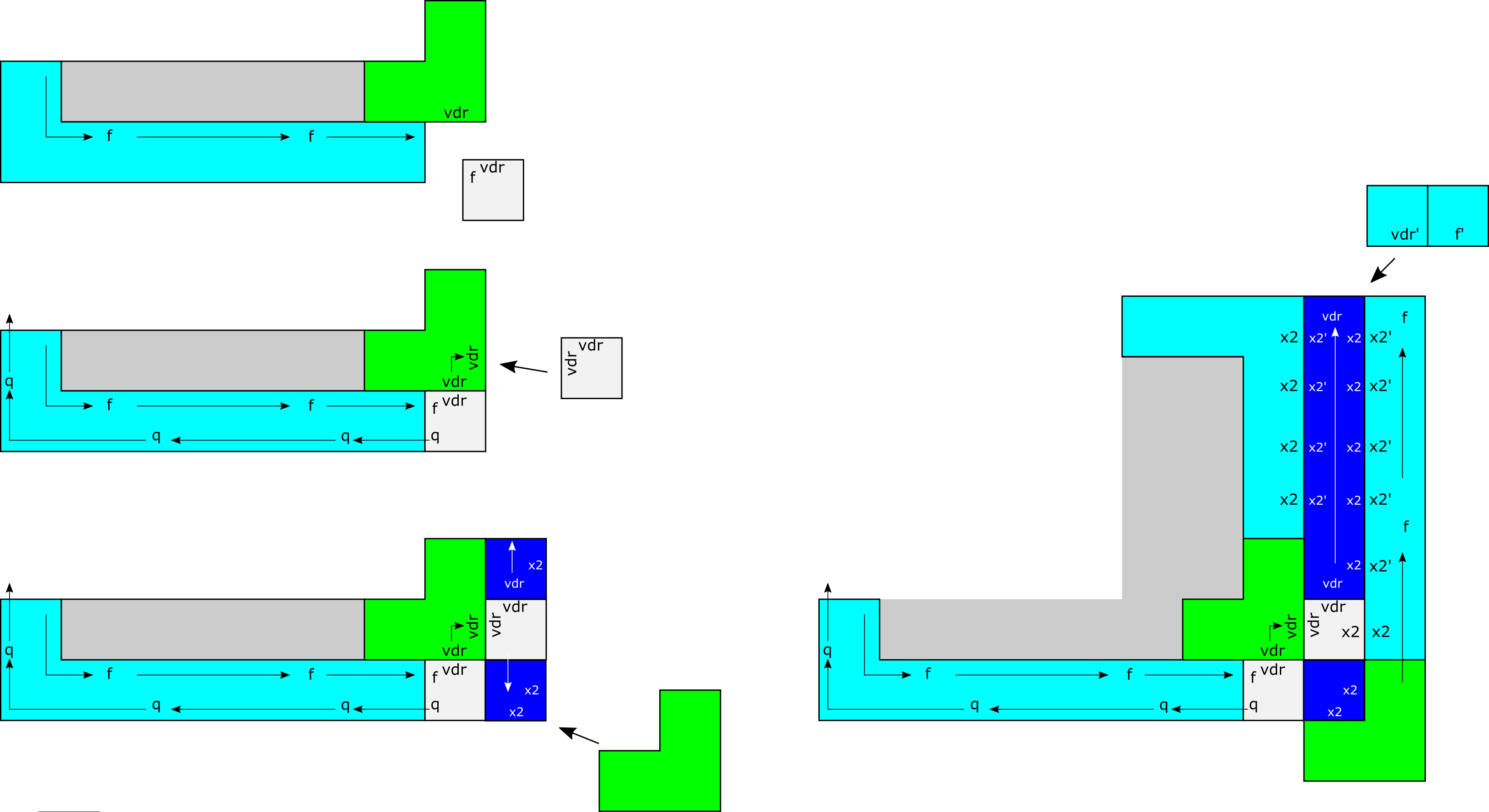}
        \caption{(left) Second case of vertical doubling row. We enhance our frame building tiles so that when a tile of a row growing along the south edges of a frame layer places a 
        tile, $t$, at location 3, the tile of the $\cgSE$ passes a message to activate a $vdr$ glue to allow for cooperative binding. This allows for a similar process as Figure~\ref{fig:vdr2}
        to generate a vertical doubling row.
        A vertical doubling row does not turn a northeast corner, unlike a typical standard growth path tile.
        Moreover, in this case, instead of propagating a $q2$ message CCW from the tiles of $CG$, a $q$ message is propagated.
        If a $q2$ signal was created, the next $\cgSE$ placed at the (2, -2) vector would encounter a standard path in location 3 as before. 
        See Figure~\ref{fig:collision-detection-corner-cases2} for an example of this case.
        (right) Example of frame growth over a vertical doubling row.
         When the vertical doubling row is initiated, the $q$ message is passed through the north side of $\cgSE$, allowing $x2$ glues to be activated
        }
        \label{fig:vdr1}
    \end{figure}


        \item When a frame layer tile is placed adjacent to a $\cgSE$ in position 2 as shown in Figure~\ref{fig:secgCases}, and this $\cgSE$ does not expose $x2$ glues, as long as the $\cgSE$ does not belong to layer 1, the appropriate collision detection duple in Figure~\ref{fig:rectangle-detection-duple} will attach and cause the initiation of a $d$ (\emph{detach}) message which will be propagated backward through the path, starting within the $\cgSE$ so that the $\cgSE$ binds to the colliding path tile using a $d$ glue, all of which are strength-2. (A $\cgSE$ of layer 1 would not have been able to attach to a collision detection gadget and will initiate the growth of layer 2 on its back.)  Each tile/duple in the path that receives the $d$ message turns the $x2'$ glue on its left side $\off$, thus detaching it from the layer that it grew on top of.  When a $d$ message arrives at the north side of a $\cgSE$, that $\cgSE$ deactivates the $x2'$ glues with which it initially attached to the assembly.\label{item:detach}

\begin{figure}[htp]
\centering
    \includegraphics[width=\textwidth]{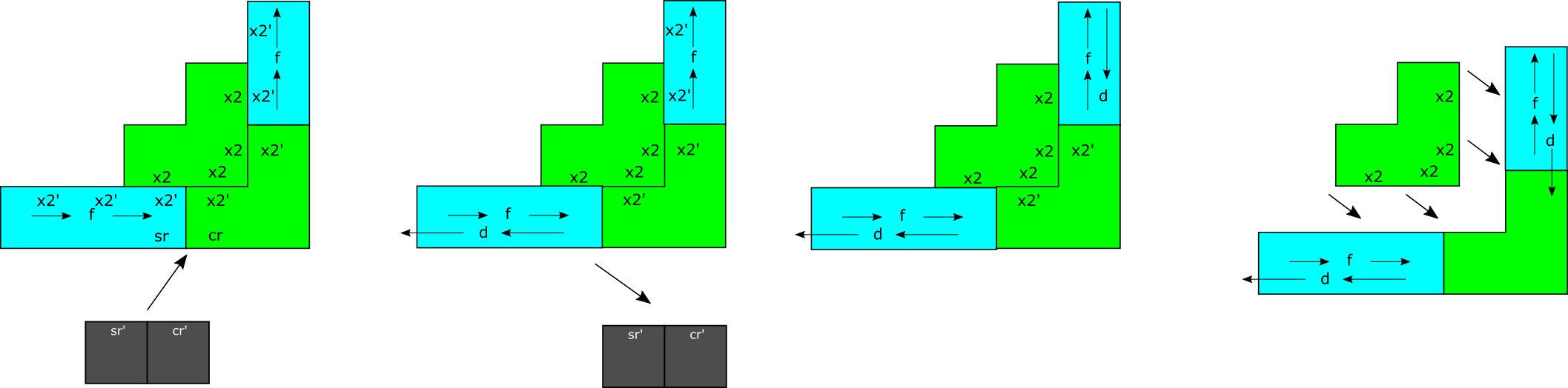}
    \caption{
        (left) A duple attaches to a tile of a standard path and the westernmost tile of a $\cgSE$ that initiates the $d$ (detach) message.
        (center left) This causes the $x2'$ glues of the standard path tiles to deactivate, and a strength two attachment of $d$ glues between tiles of the same layer to form.
        (center right) In the case that the layer which received the detach signal is the outermost later, this signal will return through the north side of the $\cgSE$
        (right) Both $x2'$ glues of the $\cgSE$ of the outermost layer also deactivate, leading to the entire frame being able to detach.
    }
    \label{fig:rectangle-detection-duple}
\end{figure}

\end{enumerate}

In Section~\ref{sec:frame-proof}, we will see that after some finite number of layers have assembled around $\alpha$, the last layer assembled will be a rectangle.
In the case where a frame layer is rectangular, the $\cgSE$ which initiates the $d$ message, call it $CG_d$, will eventually receive this message on its north side, turning off its $x2$ glues. 
Before turning $\off$ this $\cgSE$'s $x2'$ glues, we first \emph{prime} the layer that $CG_d$ is attached to. 
At a high-level, priming exposes a glue on the frame layer that $CG_d$ is attached to which allows a special singleton tile to bind, triggering the signals that will elect a leader (a single special tile that will initiate shape replication). 
The $d$ message described in Case~\ref{item:detach} will eventually turn all left $x2'$ glues exposed by this layer $\off$, allowing $CG_d$ to disassociate and expose the glue turned $\on$ in the priming process. 
Priming a layer and its purpose as well as electing a leader tile belonging to layer 1 are described in the next section. 
Additionally, it is important to note that layers which are not exactly rectangular will be unable to detach. 
This is due to the fact that a $\tau$ stable connected path exists to $\alpha$ for all $\cgSE$ part of a non-rectangular layer. 
An example can be seen in Figure \ref{fig:new-frame2-full-example-disconnect}.


\begin{figure}[htp]
\centering
    \includegraphics[width=3.5in]{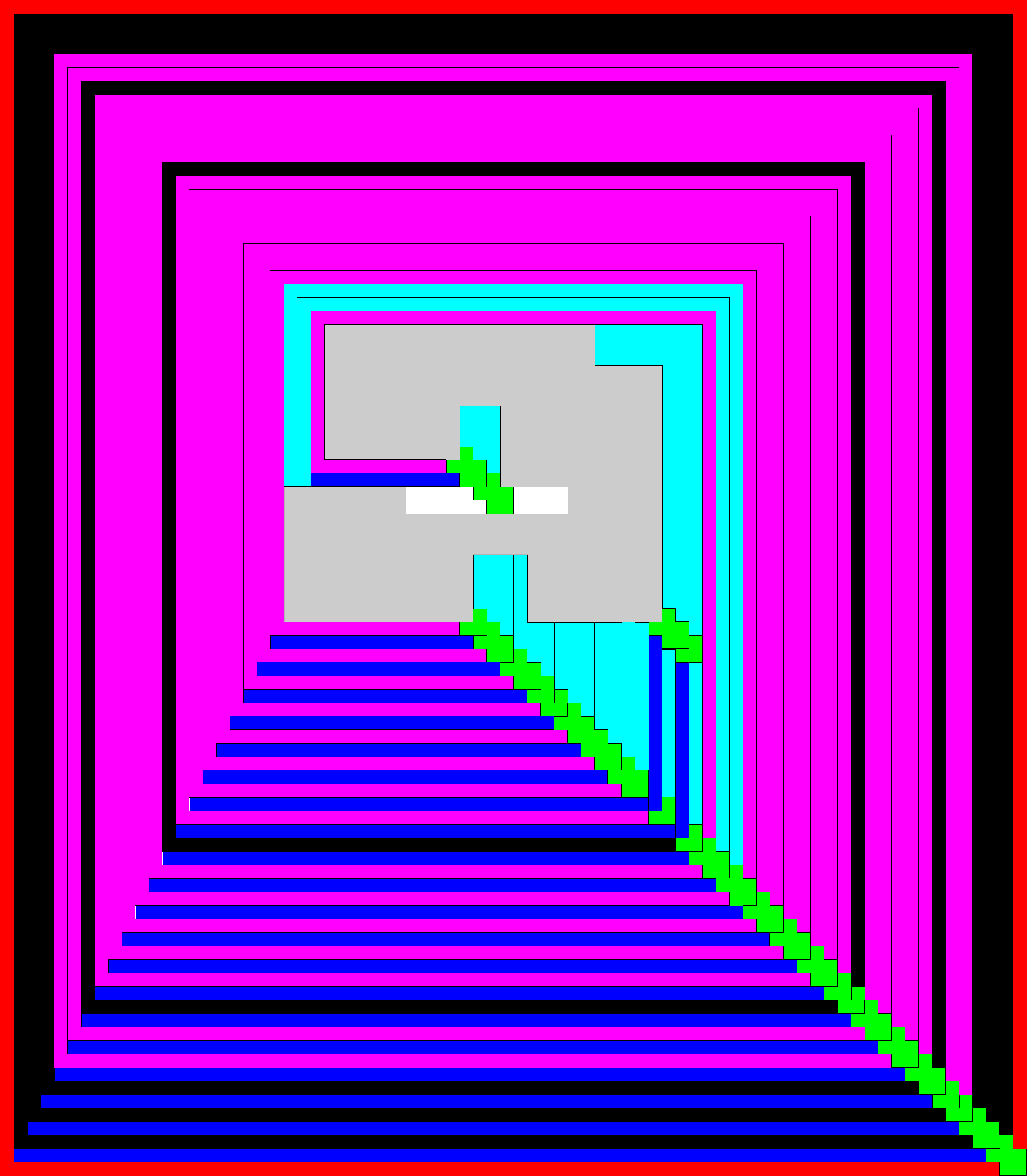}
    \caption{The example from Figure~\ref{fig:new-frame2-full-example} with pink portions highlighting paths which will disconnect from the layers inside of them (because each of those layers grew to the west side of a $\cgSE$), but will ultimately be unable to completely detach since the innermost $\cgSE$ on each connected path will not disconnect (as it won't receive the detach message from the north), and those $\cgSE$'s will create a $\tau$-stable connected path all the way to $\alpha$.  The red outer layer denotes the rectangular layer which can detach, and the entire segment of the longest connected path which attempts to detach is shown in black for highlighting. This path spirals around the perimeter of $\alpha$ several times before stopping at a segment which does not detach.}\label{fig:new-frame2-full-example-disconnect}
\end{figure}

\subsection{\texorpdfstring{Electing a leader and casting a mold of $\alpha$}{Electing a leader and casting a mold of alpha}}\label{sec:main-leader}

In this section, we assume that the last added layer $F$ of the frame has completed growth and is rectangular. Here we give a high-level description of what it means to elect a leader tile of layer 1 and how this is achieved.  (See Section~\ref{sec:leader-details} for more details.) When $F$ completes the growth of a rectangle, it will pass a \emph{detach} ($d$) message CW through each tile of $F$ back to the only $\cgSE$ belonging to $F$, which we denote by $CG_F$. When the $d$ message is received along the north edge of the northernmost tile of $CG_F$, it initiates a series of signals so that after $F$ detaches, the remaining assembly exposes a strength-2 glue. This strength-2 glue is exposed so that a singleton tile can bind to it. This binding event initiates the signals that ``scan'' from right to left for the first corner tile of a $\cgSE$ belonging to layer 1. Notice that this will be the easternmost tile of the southernmost tiles belonging to layer 1. The tile directly to the west of this tile is called the \emph{leader tile} of layer 1. For an overview see Figure~\ref{fig:leader-election-overview}. 

Now that a leader is elected, note that tiles need not completely surround $\alpha$ in concavities. 
In other words, there may be some empty tile locations adjacent to tiles of $\alpha$. 
Below we show how to guarantee every location adjacent to a tile of $\alpha$ is filled. At a high-level, we describe signals and tiles that ``extend'' the frame so as to completely surround $\alpha$ by passing a $g$ message CCW around $\alpha$.
Starting with the leader $\cgSE$, as this message is passed, it activates glues which we use to replicate $\alpha$. The leader tile exposes a unique glue.  For more detail, see Section~\ref{sec:outlining}.

Once we have exposed glues that will allow for the replication of $\alpha$, we propagate a $br$ message through layer 1 of the frame that deactivates all of the $x'$ glues of layer 1 except for the $x'$ glue on the north edges of the leader tile. This will allow $\alpha$ to disassociate from the frame and allow the frame to be used to replicate $\alpha$.

\subsubsection{Details for priming a frame layer and electing a leader}\label{sec:leader-details}

As discussed in the previous section, $F$ will pass a $d$ message CW through each tile of $F$ until the $\cgSE$ of $F$, which we denote by $CG_F$, is reached. When the $d$ message is received along the north edge of the northernmost tile of $CG_F$ (via the $d'$ glue), as depicted in Figure~\ref{fig:SE-corner-gadget-tiles}, it triggers a $p'$ glue to turn $\on$, which propagates a signal that turns the $p2'$ glue $\on$ on the east edge of a tile belonging a $\cgSE$ tile of the layer that $CG_F$ is attached to. We denote this tile by $CG$ and note that it is in position $B$ of Figure~\ref{fig:secgCases}. Moreover, the binding of the $p$ glue belonging to $CG_F$ fires signals which continue the propagation of the $d$ message and turn the $x2'$ glues of $CG_F$ $\off$. 

\begin{figure}[htp]
\centering
  \subfloat[][Signals and glues used to enhance corner tiles of $\cgSE$'s. Once a \emph{primer tile} (shown in \ref{fig:primerTile}) attaches via the $p2'$ glue, it initiates a $sp$ message. When the $sp$ message begins propagation, the $p2'$ glue turns $\off$, and the primer tile disassociates.]{%
        \label{fig:left-right-leader-election-init}%
        \makebox[2.7in][c]{\includegraphics[scale=.33]{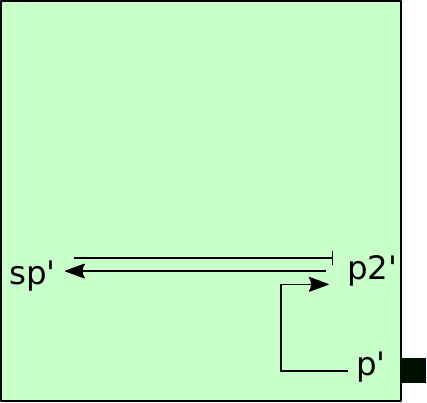}}
        }%
        \quad\quad
  \subfloat[][This tile binds to a primed $frame$ layer to initiate the inward propagation of the $sp$ message which eventually reaches the $\cgSE$ that is attached to $\alpha$. ]{%
        \label{fig:primerTile}%
        \makebox[2.7in][c]{\includegraphics[scale=.33]{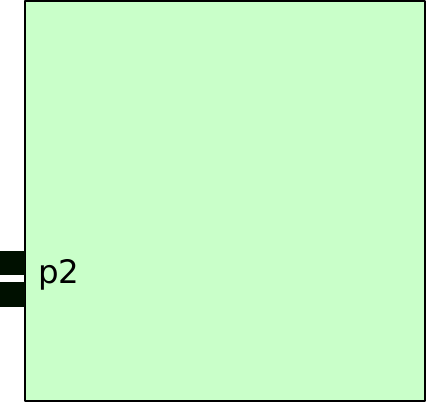}}
        }%
  \caption{Gadgets that initiate the propagation of a $sp$ message which ``marks'' a $\cgSE$ attached to $\alpha$.\vspace{-10pt}}
  \label{fig:frame-priming}
\end{figure}
%

Then, in order to detect a leader $\cgSE$, we enhance the corner tile belonging to a $\cgSE$ with the signals depicted in Figure~\ref{fig:left-right-leader-election-init} and all of the frame building tiles discussed (including those tiles belonging to $\cgSE$) so far with signals depicted in the center figure of Figure~\ref{fig:leader-election-tiles}.

\begin{figure}[htp]
\centering
        \includegraphics[scale=.4]{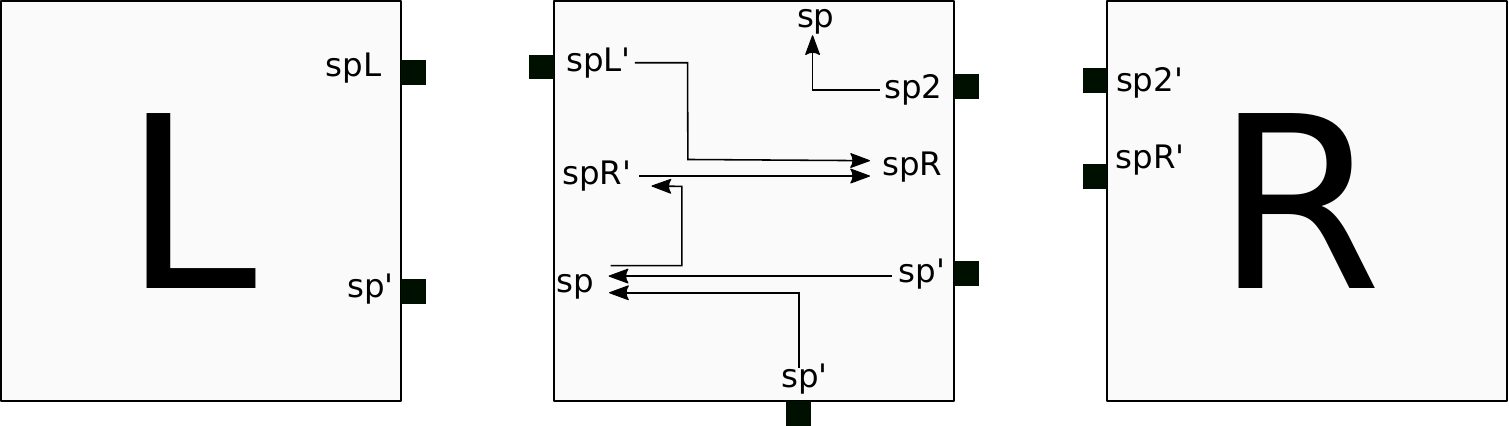}
  \caption{(Left) A tile that binds to the $sp'$ and $spL'$ glues on the west edge of a westernmost tile of a frame that initiate an $spR$ message. (Middle) We enhance each frame building tile with the signals depicted here. The $sp'$ glues, $spL'$ glue, and $sp'$ glue are initially $\latent$ and then triggered to turn $\on$  after an $x2'$ glue of the frame building tile binds. (These $\on$ signals and the $x2'$ glue are not shown here.) (Right) A tile that binds to the $sp2$ and $spR$ glues on the east edge of an easternmost tile of a frame that propagate an $sp$ message one tile to the north of the tile to which it binds.\vspace{-10pt}}
  \label{fig:leader-election-tiles}
\end{figure}

Equipped with these new signals, frame tiles elect a $\cgSE$ by ``scanning'' from right to left for the first corner tile of a $\cgSE$ belonging to layer 1. Notice that this will be the easternmost tile of the southernmost tiles belonging to layer 1. For an overview see Figure~\ref{fig:leader-election-overview}. This scanning is accomplished as follows. Frame tiles can pass an $sp$ message from right to left through tiles of standard paths, doubling rows and corner gadgets. If a tile belonging to layer 1 receives the $sp$ message, it must be a corner tile of some $\cgSE$ of layer 1 and this will be the elected $\cgSE$. If the propagation of the $sp$ message exposes an $sp$ glue on the west edge of a westernmost tile of the frame, then a singleton tile (labeled $L$ in Figure~\ref{fig:leader-election-overview}) may use this glue and an $spR$ glue to bind and initiate the propagation of an $spR$ message which propagates to the right. Notice that the $spR$ message will eventually expose an $spR$ glue on the east edge of an easternmost tile of the frame, allowing for a singleton tile (labeled $R$ in Figure~\ref{fig:leader-election-overview}) to bind via its $spR'$ and $sp2'$ glues. This fires signals which propagate the $sp$ message on tile to the north before continuing to propagate the $sp$ message from right to left, ensuring that eventually a $\cgSE$ of layer 1 will be elected.

\begin{figure}[htp]
      \centering
      \subfloat[][An example of electing a $\cgSE$. The particular $\cgSE$ that is elected has a corner tile that is the easternmost tile of the southernmost tiles of layer 1.]{
            \includegraphics[width=2.8in]{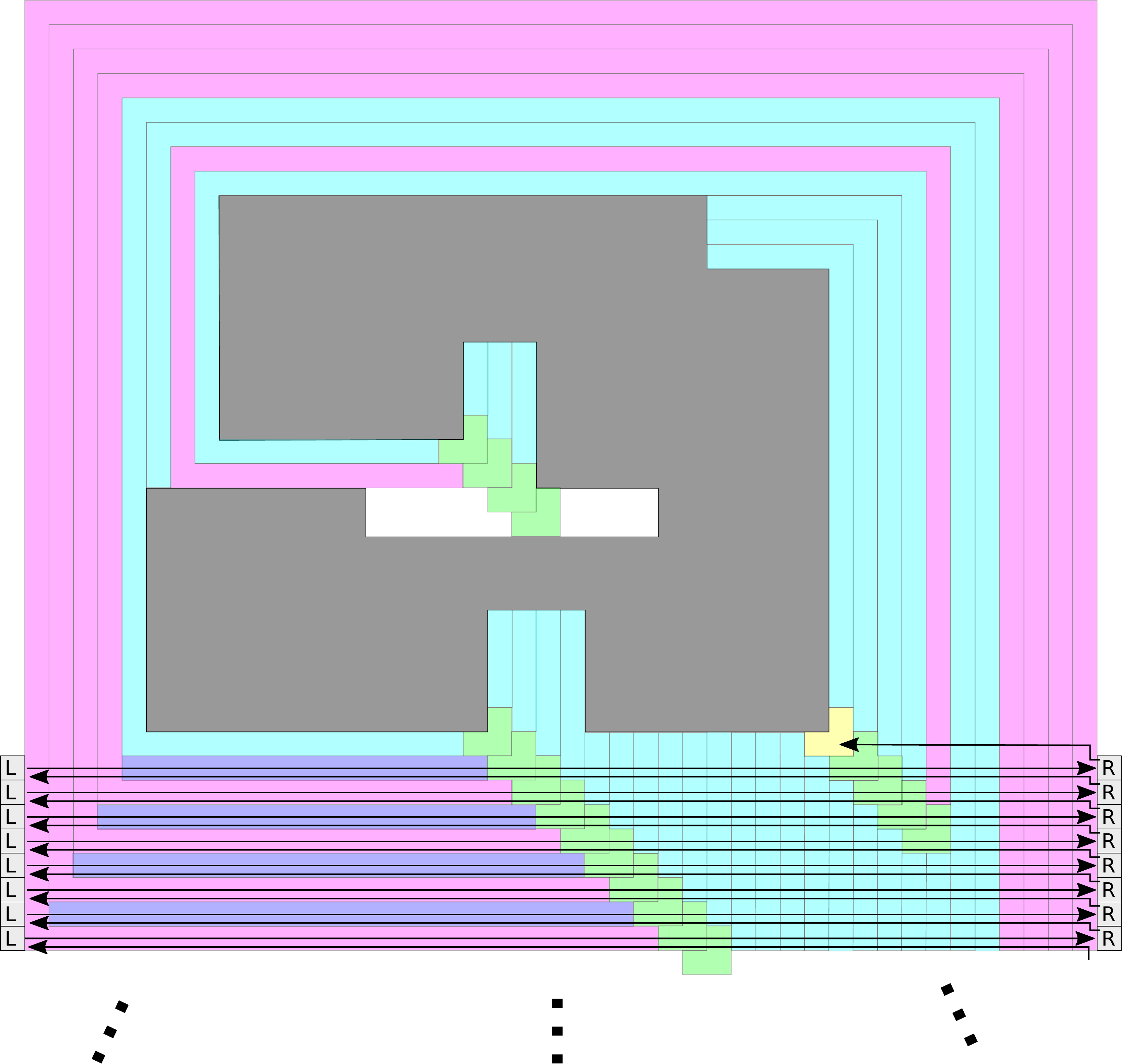}
            \label{fig:leader-election-overview}
      }
      \quad \quad
      \subfloat[][Once the leader $\cgSE$ has been elected, a $g$ signal is passed to its north (the signal closest to $\alpha$). This passes through frame tiles until reaching a void in the concavity. 
      From there, the exposed $g$ glues allow for placement of tiles from Figure~\ref{fig:post-collision-handling} (seen as orange). 
      Note that empty locations still exist, however the newly placed tiles are adjacent to the remaining edges of $\alpha$.
      After the $g$ signal returns to the leader $\cgSE$, the $br$ signal follows $g$.]{
            \includegraphics[width=2.8in]{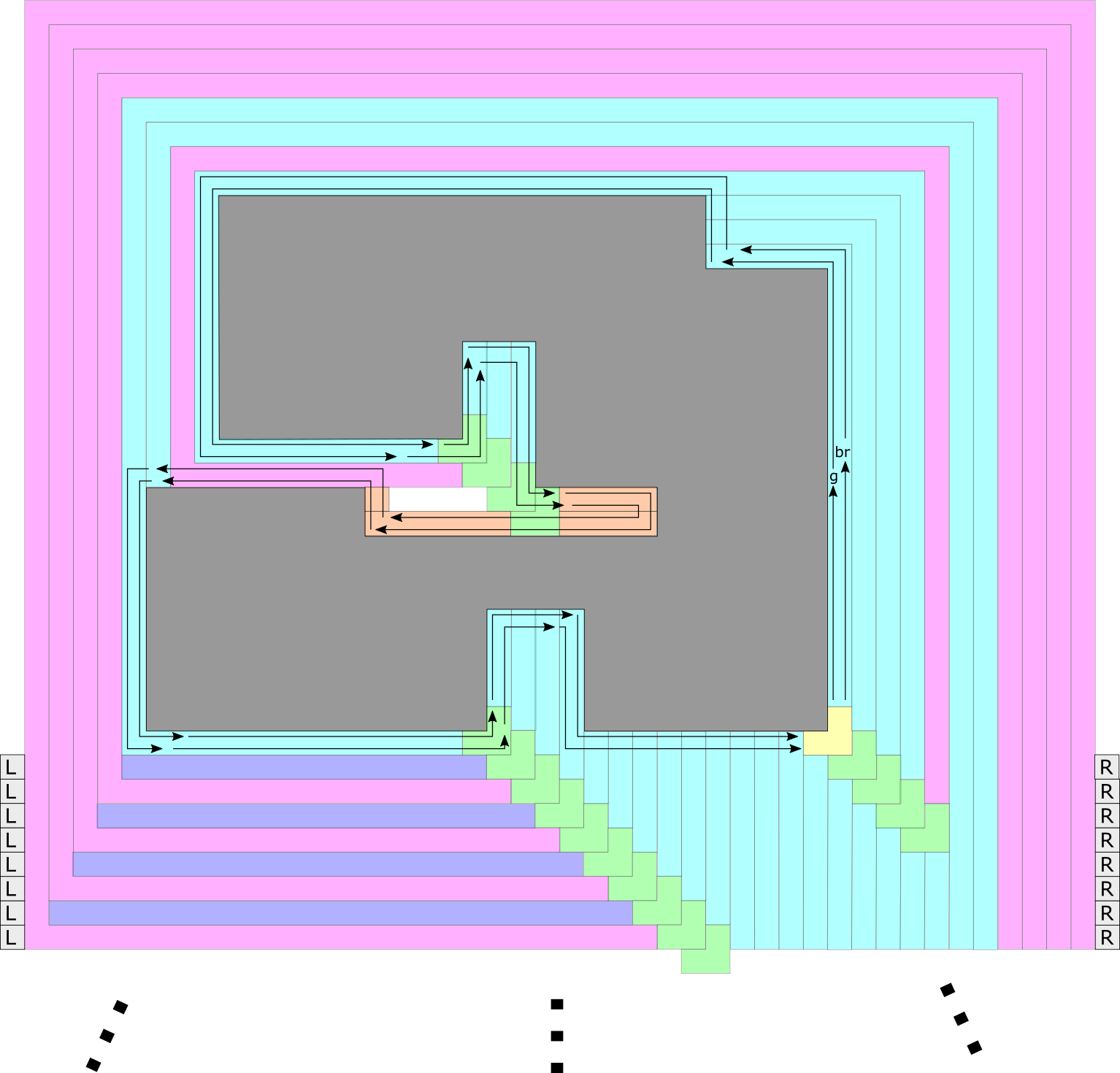}
            \label{fig:mold-creation}
      }
      \caption{Leader election and mold creation processes}
      \label{fig:leader-election-mold-creation}
\end{figure}


From Figure~\ref{fig:layer1-SEcg}, we can see that once a $sp$ message is received by a $\cgSE$ of layer 1, it propagates an $sp$ signal to the westernmost tile of that $\cgSE$. 
We will call this tile the \emph{leader} tile. When the $sp'$ glue of the leader tile binds, it fires signals to turn $\on$ a $z$ glue and a $y$ glue. The purposes of these glues are described in Section~\ref{sec:rep}.

Now that a leader is elected, note that layer 1 need not completely surround $\alpha$. In other words, there may be some empty tile locations adjacent to tiles of $\alpha$. 
In the next section, we show how to ``complete'' layer 1 so that for every tile location adjacent to a tile of $\alpha$, this location contains a tile of layer 1.

\subsubsection{Details for casting a mold of $\alpha$}\label{sec:outlining}

In this section, we show how to ``extend'' layer 1 of the frame to completely surround $\alpha$ so that for every $x$ glue on $\alpha$, there is an $x'$ on some tile of the extended layer 1 that is bound to $x$.

\begin{figure}[htp]
\centering
    \includegraphics[width=2in]{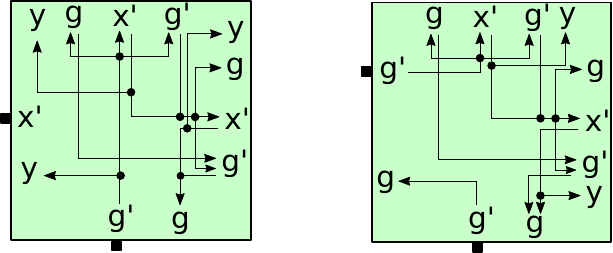}
    \caption{ (Left) A tile capable of propagating a $g$ message around an $\alpha$ in the CCW direction. Notice that this tile can pass a $g$ message once in each direction.
    This particular tiles propagate the $g$ message in the south-to-north direction.
    These tiles (included rotated variants) are added into the tileset, which enable the mold of $\alpha$ to be extended to all locations adjacent to locations of $\alpha$ (as seen in Figure~\ref{fig:mold-creation})    
     By appropriately rotating this figure, we can obtain signals and glues that propagate signals in the east-to-west, north-to-south, and west-to-east directions.
     Additionally, we enhance the frame building tiles that expose $x'$ on their left side with these signals and glues.
     (Right) A tile capable of propagating a $g$ message around a NE corner. 
     This tile is added to the tileset and we enhance the corner tile of the $\cgNE$ tiles with these signals and glues. 
     Moreover, we enhance the corner tiles of $\cgNW$ and $\cgSW$ tiles with these signals, appropriately rotated.\vspace{-10pt}}
  \label{fig:post-collision-handling}
\end{figure}

\begin{figure}[htp]
      \centering
          \includegraphics[width=2.2in]{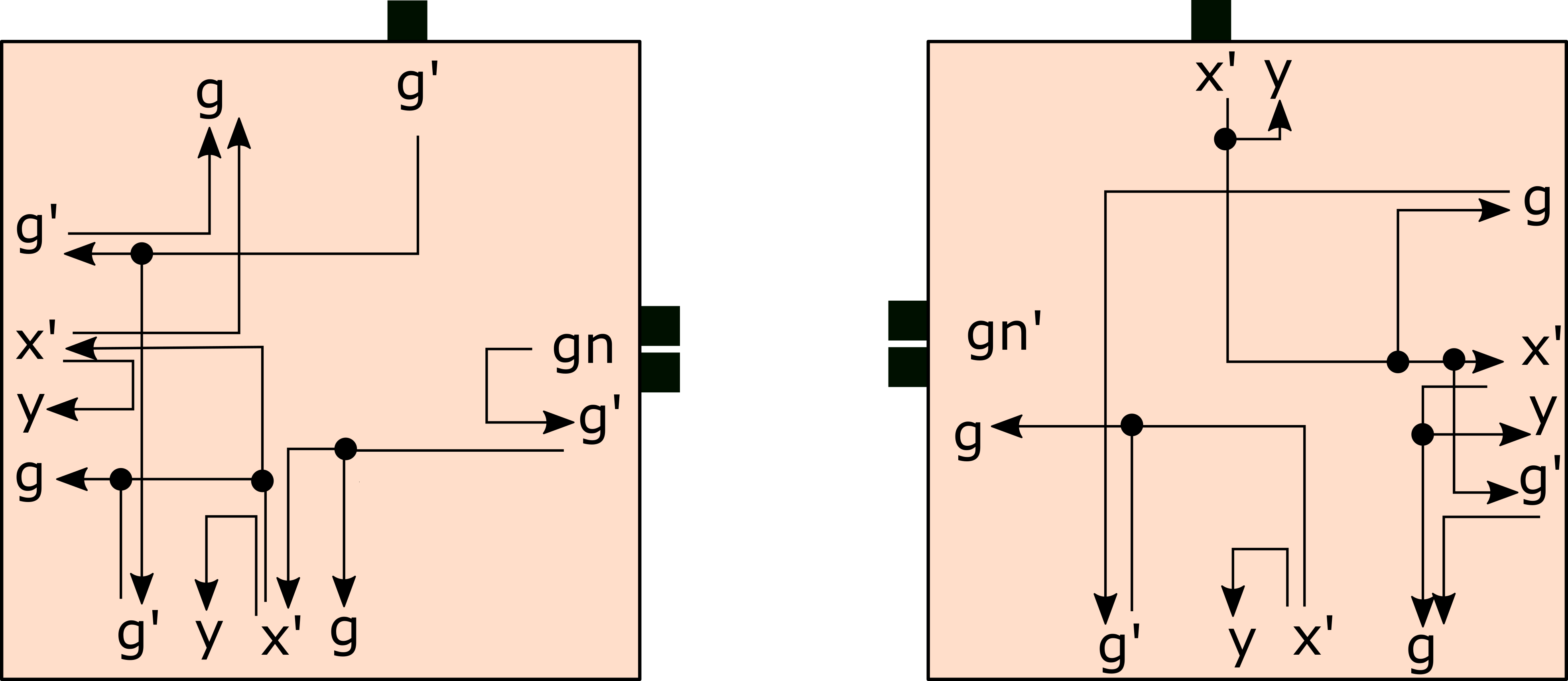}
          \caption{ 
              A duple required to handle the $g$ message being presented to overhangs.
              Both this and a variant rotated 180 CCW are added to the tileset.  
          \vspace{-10pt}}
        \label{fig:post-collision-handling-duple}
      \end{figure}

\begin{figure}[htp]
      \centering
          \includegraphics[width=3.5in]{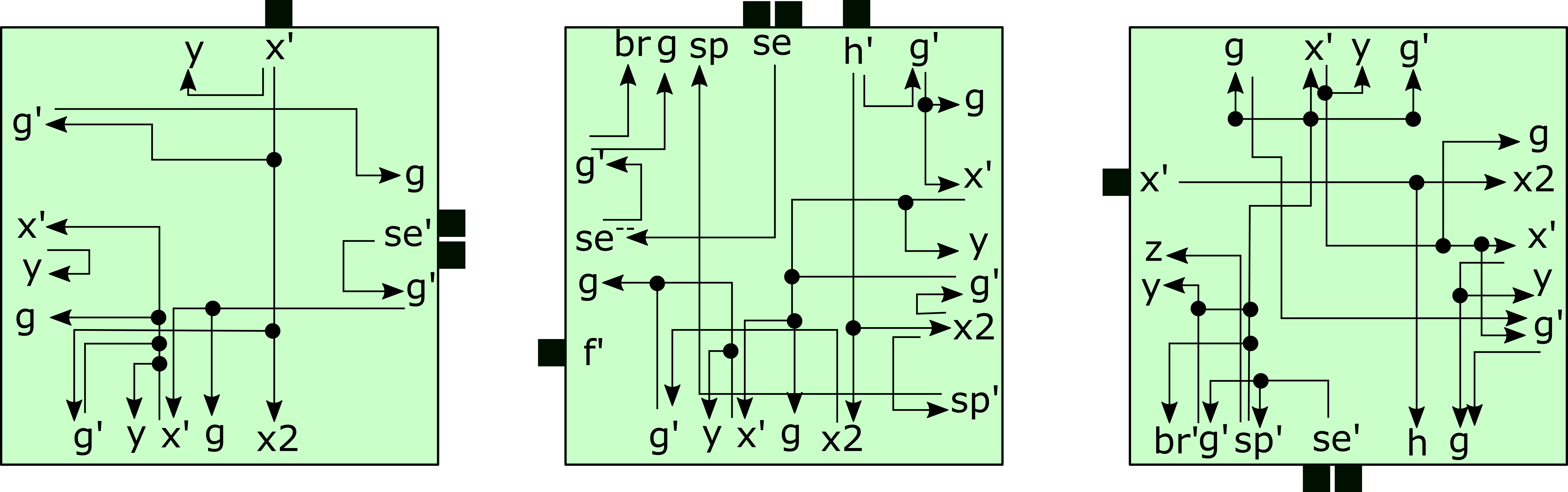}
          \caption{
                We now update the $\cgSE$ tiles of layer 1 with the following signals.
                Note that the only $\cgSE$ that will receive an $sp$ signal is the leader, and as such it is able to activate the $z$ glue signaling the southeasternmost corner.
                The $g$ message is also begun at this point, and continue through the tiles adjacent to $\alpha$, in a CCW manner.
                Once the $g$ signal returns to the leader tile, the $br$ signal is activated.
                The remaining $br$ signals which follow $g$ are not visualized,
                but the description of how they are carried out is detailed in Section~\ref{sec:signal-following}.
                These signals are also applied to the layer 2 $\cgSE$, minus the $br$ activation and leader election of $z$.
          \vspace{-10pt}}
        \label{fig:L1-cgSE-g}
      \end{figure}

Figure~\ref{fig:post-collision-handling} depicts tiles capable of propagating a $g$ message CCW around $\alpha$.
When we enhance frame building tiles with the signals and glues depicted on the left in Figure~\ref{fig:post-collision-handling}, the $g'$ glue (shown in the $\on$ state in the figure) is initially $\latent$ and is turned $\on$ when $x'$ binds (this signal is not shown in the figure). 
\update{
As the $g$ message propagates around the perimiter of $\alpha$, this message is passed through tiles of layer 1 that have previously been bound to $\alpha$ or tiles of layer 2 which are in locations adjacent to $\alpha$.
If no tile is adjacent to $\alpha$ with an exposed $g$ glue, the message propagates by the attachment of a tile (depicted on the left in Figure~\ref{fig:post-collision-handling} up to rotation).
See Figure~\ref{fig:mold-creation} for an example of tile attachment by exposed $g$ glues.

Additionally, we may run into certain concavities with overhangs - that is, it must change its direction from south to east, or north to west.
In both these cases, only a single $g$ glue is available for binding. This is solved by the usage of duple tiles in Figure~\ref{fig:post-collision-handling-duple}.
See Figure~\ref{fig:post-collision-handling-duple-need} for demonstrations of such situations.
}
With these enhancements and additional tiles, the $g$ message propagates CCW starting from the leader tile until it is passed back to the leader tile. At this point, for every $x$ glue on $\alpha$, there is an $x'$ on some tile of the extended layer 1 that is bound to $x$. An example is given in Figure~\ref{fig:mold-creation}.

From Figure~\ref{fig:post-collision-handling}, we can notice that as the $g$ message propagates, it ensures that a $y$ glue is exposed on any edge that binds to an $x$ glue of $\alpha$. Though only a single $y$ glue is signaled to turn $\on$ in Figure~\ref{fig:post-collision-handling}, note that we could turn on any number of glues that we like. In particular, we not only turn a $y$ glue $\on$, but we also turn a $w$ glue $\on$. The purposes of these glues will be explained in Section~\ref{sec:rep}.

\begin{figure}[htp]
      \centering
          \includegraphics[width=4in]{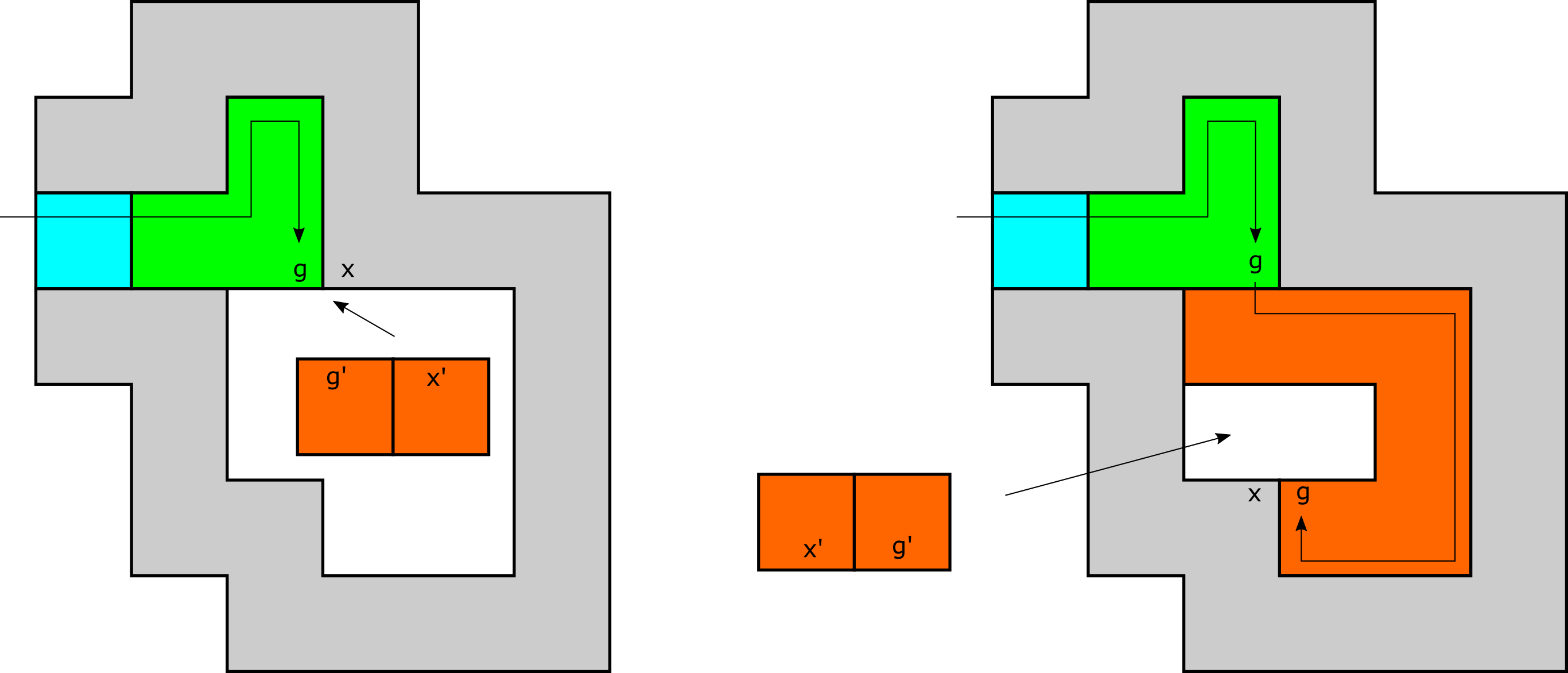}
          \caption{ 
              Situations where $g$ signal cannot be passed without presence of duple from Figure~\ref{fig:post-collision-handling-duple} 
          \vspace{-10pt}}
        \label{fig:post-collision-handling-duple-need}
      \end{figure}


Next we propagate a $br$ message through layer 1 of $\alpha$ that deactivates all of the $x'$ glues of layer 1 except for the $x'$ glue on the north edges of the leader tile.
The beginning of this signal is shown in Figure~\ref{fig:L1-cgSE-g}
This will allow $\alpha$ to disassociate from the frame and allow the frame to be used to replicate itself via the same frame building process.
$br$ propagates as follows. Once the $g$ message is received by the leader tile, it initiates the $br$ message. 
In Section~\ref{sec:signal-following}, we describe a technique for designing a message that ``follows'' another message.  
Here we use this technique, enhancing each frame building tile discussed so far with a $br$ message that follows the $g$ message in order. 
Moreover, as $br$ is passed to each consecutive tile of layer 1, it fires deactivation signals that turn $\off$ any $x'$ glues on each tile. 
In other words, the $br$ message is passed through each tile of layer 1 in the same order that the $g$ message was passed through each tile of layer 1. 
In this way, the $x'$ glues of layer 1 bound to the $x$ glues of $\alpha$ turn $\off$ with the exception of the $x'$ glue on the north edge of the leader tile. 
This $x'$ glue is left $\on$.

After the $br$ message propagates through the tiles of layer 1 and $\alpha$ disassociates, note that the edges that had previously exposed an $x'$ glue bound to an $x$ glue of $\alpha$ now expose (or will eventually expose after pending signals fire) $z$ and $y$ glues on the north edge of the leader tile, and $w$ and $y$ glues otherwise.  After all pending signals of the frame fire, we say that the frame is a \emph{complete frame assembly for shape} $\alpha$. In Section~\ref{sec:rep}, we show how to use a completed frame for a shape to replicate arbitrary shapes.

\subsection{Followable Messages}\label{sec:signal-following}

In this section we show how to pass a massage through a sequence of tiles such that after the message has been passed, a second message can be passed through the exact same sequence of tiles in either the same order. 
For example, in our frame construction, signals propagate a $g$ message through a sequence of layer 1 tiles $\{\T_i\}_{i=0}^n$ (not necessarily distinct). 
In our construction, we then propagate a $br$ message through layer 1 through a series of glue activations such that this message follows the sequence of tiles $\{\T_i\}_{i=0}^n$ in that order. In this case, we say that the $br$ message \emph{follows} the $g$ message.

\update{Figure~\ref{fig:signal-following-none} shows an $g$ message being passed through a tile. Let $T_G$ denote this tile. This message enters from the south and then may potentially be output through the north, east, or south depending on if collisions occur.
 The goal is to ensure that a second message can be output through exactly that same side (and no others). Other cases where the $g$ message enters through the north, east, or west are equivalent up to rotation. 
 For each possible output signal of the $g$ glue in $T_G$, we define glues on the signal input side of the $T_G$ which are activated by the output $g$ glue being bound. 
 As shown in Figure~\ref{fig:signal-following-none}, the north $g$ glue activates $brn^\prime$, the east $g$ glue activates $bre^\prime$, and the south $g$ glue activates $brs^\prime$. Informally, the activated $brn^\prime$, $bre^\prime$, or $brs^\prime$ glue ``records'' the output side of the $g$ message. In the case shown in Figure~\ref{fig:signal-following-none} where the $g$ message enters from the south, the $brn^\prime$, $bre^\prime$, and $brs^\prime$ glues are sufficient for recording the output side of the $g$ message. In cases where the $g$ message enters through the north, east, or west, a $brw^\prime$ glues is required to record the case where the $g$ message exits through the west side of a tile. The $br$ signal is then propagated using $brn^\prime$, $brs^\prime$, $bre^\prime$, and $brw^\prime$ glues. Figure~\ref{fig:signal-following-forward} depicts the signals and glues for propagating the $br$ signal in the case where the $g$ message enters from the south. In this case the $br$ signal will also enter from the south. 
 The $br$ signal is propagated through $T_G$ as exactly one of the $brn^\prime$, $brs^\prime$, and $bre^\prime$ glues binds to one of the $brn$, $bre$, and $brs$ glues on the output side of a tile to the south of $T_G$ that is propagating $br$. All of the $brn$, $bre$, and $brs$ glues must be activated as the tile to the south of $T_G$ has no ability to know which direction the $g$ message of $T_G$ will take. The $br$ signal passed to $T_G$ will have the same output side as the $g$ signal. For example, if the $g$ message enters from the south and exits through the east, then, as shown in Figure~\ref{fig:signal-following-none}, the glue $bre^\prime$ will be activated; $brn^\prime$ and $brs^\prime$ will remain latent. Then, as the $br$ signal propagates through the tile to the south of $T_G$, $brn$, $bre$, and $brs$ are all activated on the north side of the tile. When $bre$ and the $bre^\prime$ glue on the south edge of $T_G$ bind, this binding event activates the glues $bre$, $brs$, and $brw$ on the east edge of $T_G$, effectively propagating the $br$ signal to the tile to the east of $T_G$. This is shown in Figure~\ref{fig:signal-following-forward}. Notice that there are no signals belonging to $T_G$ that fire when $brs^\prime$ binds. This is because no signals are needed to propagate $br$ to the south of $T_G$. The binding of $brs$ and $brs^\prime$ are enough to propagate $br$ to the south of $T_G$. 
 
 }




\begin{figure}[htp]
\centering
  \subfloat[][
        An example of signals used to propagate an $g$ message CCW.]{%
        \label{fig:signal-following-none}%
	        \includegraphics[width=1.1in]{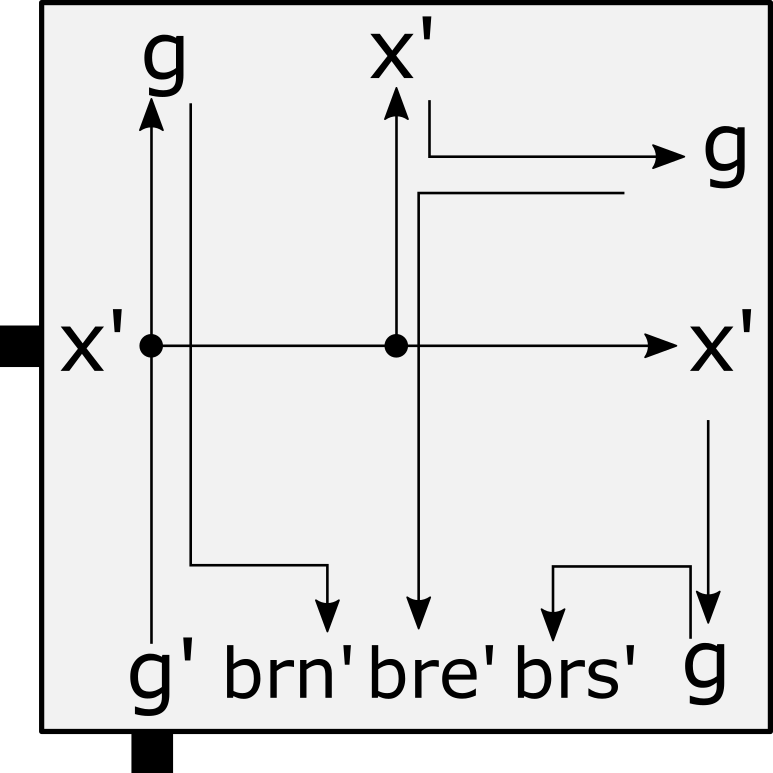}
        }%
        \quad\quad
  \subfloat[][A $br$ message that is following a previously passed $g$ message. 
  ]{%
        \label{fig:signal-following-forward}%
        		\includegraphics[width=1.1in]{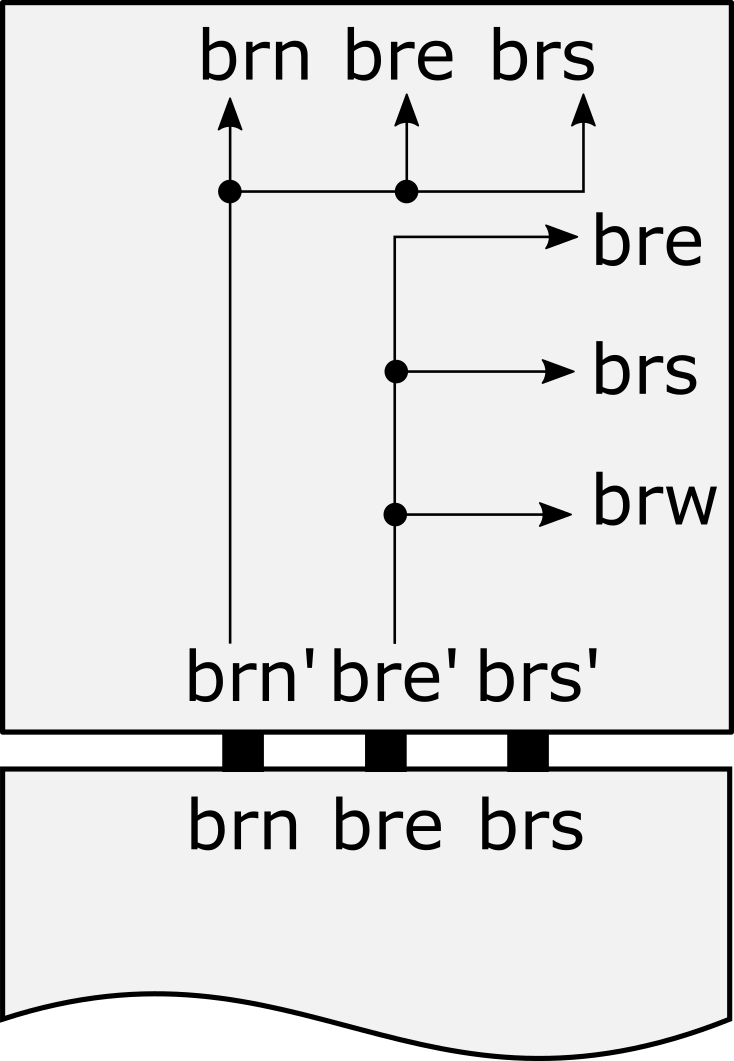}
        }%
  \caption{Tiles which demonstrate signal following.\vspace{-15pt}}
  \label{fig:signal-following-all}
\end{figure}


\subsection{Correctness of the frame construction}\label{sec:frame-proof}

The goal of the frame construction is to result in an assembly which completely encases $\alpha$, making a perfect mold of its shape, and then detaches.  Furthermore, and extremely importantly, that mold must have uniquely identified exactly one tile on its interior (i.e. its ``elected leader'').  We now state the properties which are guaranteed by the frame construction, in the form of lemmas, and prove each of them.

First, for notation we will refer to the collision detection tiles as $\cdtiles \subset T$, the tile type(s) which make up $\alpha$ as $\alphatiles$, and the remaining tiles as $\frametiles = (T - \cdtiles) - \alphatiles$.

\begin{lemma}\label{lem:rectangle-layer}
For each STAM system $\calT = (T,\alpha,2)$ with tile set $T$ (i.e. the frame building tile set) and input assembly $\alpha$, there exists some constant $c_{\alpha}$ such that regardless of the assembly sequence, after $c_{\alpha}$ tile attachments a rectangular frame layer of tiles in $\frametiles$ will have grown around $\alpha$. 
\end{lemma}

\begin{proof}

To prove Lemma~\ref{lem:rectangle-layer}, the main point we must prove is that, regardless of the shape of $\alpha$, every possible assembly sequence of $\calT$ will, in a finite number of steps, build a rectangular layer.  To prove this, we begin by noting that a rectangle is the only shape with no concavities.  The frame, which is the portion of the assembly consisting of tiles of $\frametiles$, grows as a series of (sometimes partial) layers around (portions of) $\alpha$.  Every tile of $\frametiles$ attaches (either as a singleton, preformed duple, or 3-tile $\cgSE$) to $\alpha$, a tile of a layer which is between the attaching tile(s) and $\alpha$ (i.e. part of a layer that grew before and between), or a doubling row.  This follows directly from the definitions of the tile types of $\frametiles$ and the fact that the only glues that are $\on$ before they attach to the assembly require at least one attachment to $\alpha$, a previous layer of the frame, or a doubling row to allow a stable attachment.  Every portion of a layer which cannot grow from the north output of a $\cgSE$ to the west input of a $\cgSE$ (the same or a different $\cgSE$) must be prevented from doing so by (1) a collision with either $\alpha$ or a tile of a previous layer or doubling row, or (2) a situation where it grew over both tiles of the south side of a $\cgSE$.

Case 1:  In any such case of a collision, there must be a location adjacent to that path where a collision detection tile or duple can attach due to the fact that such a path will expose $c$ glues along its entire right side, and any collision ensures that there is a tile (of $\alpha$, some portion of layer 1, some portion of a layer 2 path, or a doubling row) next to that path which exposes a glue matching one on a collision detection tile/duple.  The design of the collision detection tiles, as seen in Figures~\ref{fig:collision-detection-tiles} and \ref{fig:collision-detection-corner-cases}, ensures that regardless of the relative locations of those two glues, there is a collision detection tile/duple which can attach.  This attachment will begin the propagation of a $q$ message which will propagate back to the $\cgSE$ which initiated the growth of that layer.  (Various collision detection scenarios can be seen in Figures~\ref{fig:new-frame2-concave-collision1}, ~\ref{fig:new-frame2-duple-detect-collision}, ~\ref{fig:new-frame2-single-detect-collision}, ~\ref{fig:standard-row-nonquitting-cgse}, ~\ref{fig:special-collision-detection}.)

This $q$ message must be able to follow the path exactly backwards by the definitions of the layer 2 tiles, eventually arriving at the $\cgSE$ from which it grew.  This will cause that $\cgSE$ to activate glues which allow another $\cgSE$ to attach to its outside
, initiating the growth of another layer, and it will also continue to pass the $q2$ message CW which will initiate growth of a doubling row on the south side of a layer 2 path which either may have already collided, or later will, with the west side of that $\cgSE$.  Note that the $\cgSE$ receiving the $q$ message could be prevented from allowing the attachment of a new $\cgSE$ on its outside by a row (or doubling row) which has grown south of it.  This is discussed in case 2 next, but note that in this case a new layer is also initiated. 

\begin{figure}
    \centering
    \includegraphics[width=2.0in]{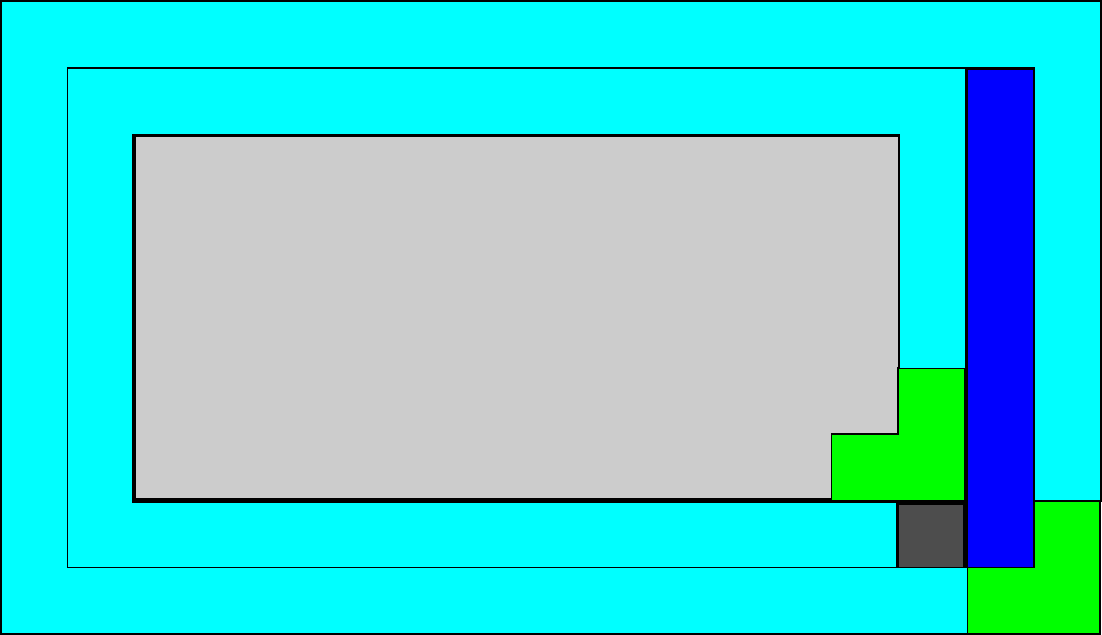}
    \caption{Example frame growth around an input assembly $\alpha$ (grey) which is nearly rectangular (with the exception of a single-tile concavity on the southeast corner). $\cgSE$'s are green, standard paths in light blue, and doubling rows in dark blue. This example shows how, when a path grows across the bottom of a $\cgSE$, it initiates the growth of a vertical doubling row, and the attachment of another $\cgSE$ at an offset of $(2,-2)$ from the first. \vspace{-10pt}}
    \label{fig:collision-detection-corner-cases2}
\end{figure}

Case 2:  When a row grows across the south side of a $\cgSE$, it detects that situation and initiates a vertical doubling row which grows up the east side of the row initiated by that $\cgSE$.  Additionally, the bottom tile of that vertical doubling row exposes the glues necessary to allow a $\cgSE$ to bind to it, thus initiating the growth of a new layer. (An example can be seen in Figure~\ref{fig:collision-detection-corner-cases2}.)

At this point, we distinguish the cases of $\cgSE$'s which attach to corners on $EXT(\alpha)$ (i.e. on the exterior of $\alpha$, not in a concavity) from those which attach to corners on $INT(\alpha)$. (See Section \ref{sec:frame-defns} for the definitions of $EXT$ and $INT$.)  We call the former \emph{exterior} $\cgSE$'s and the latter \emph{interior} $\cgSE$'s.  It is clearly possible for an interior $\cgSE$ to attach in a concavity which is too narrow for a new $\cgSE$ to attach to its outside, but our argument does not depend on the ability of such $\cgSE$'s to initiate growth of new layers so we will ignore them.  (This is because growth inside, and into, concavities must eventually completely fill at least the entrances to those concavities - even if some portions of concavities remain unfilled (at this stage of the construction).  As long as exterior $\cgSE$'s can continue to initiate the growth of new layers until a rectangular layer is formed, the layers will eventually grow over the filled in entrances to all concavities.) Instead we will focus on external $\cgSE$'s (implicitly) for the rest of the proof, and we will note that because of the handling of Case 1 and Case 2 above, a \emph{stack} of exterior $\cgSE$'s (i.e. a sequence of $\cgSE$'s which have bound either directly to each other or the doubling rows separating them, growing outward from the same corner - see Figure~\ref{fig:stack-catch-up} for an example) can always continue outward growth of additional layers except in the case where one stack, $s_1$, which is to the west of a second stack, $s_2$, grows to a point where the $x$-coordinate of the eastern tiles of a $\cgSE$ in $s_1$ is greater than the $x$-coordinates of all portions of $\cgSE$'s in $s_2$.  When this occurs, we say that $s_2$ is \emph{covered}. At this point new layers are no longer initiated by $s_2$, although the final $\cgSE$ of $s_2$ must receive a $q$ message since its layer is not a rectangle. Moreover, the tiles of the path which grew from the final $\cgSE$ of $s_2$ will expose $x2$ glues allowing a path initiated by a $\cgSE$ of $s_1$ to grow around its outside edge. In fact, one path from that stack will collide with the southeast tile of $s_2$'s final $\cgSE$, allowing a collision detection duple to attach and pass a $q$ message backward, finally allowing a path from $s_1$ to envelope $s_2$. Growth of layers will continue to be initiated by $s_1$ (assuming none is yet rectangular).  We will now discuss why this is not only desired, but guaranteed to happen so that eventually there is exactly one stack of $\cgSE$'s, which must be the case for a rectangular layer to exist. 

The above arguments ensure that an exterior stack of $\cgSE$'s can continue to initiate the growth of additional layers as long as that stack is not covered by another.  Since there can be only one southeast convex corner in a rectangular layer, and therefore only one $\cgSE$, from a single stack, contained within it, we now must show that ultimately there will be only one remaining stack.  Therefore, assume that there are $n > 1$ stacks of exterior $\cgSE$'s.  We will discuss the growth of the stack furthest to the west, $s_w$ and that of the stack furthest to the east, $s_e$.

First, we note that after some finite amount of outward growth, as $s_e$ grows outward each $\cgSE$ must be attached directly to the outside of the $\cgSE$ immediately before it, i.e. there are no vertical doubling rows.  This is because neither of the 2 scenarios which result in a vertical doubling row and thus a space between the consecutive $\cgSE$'s of a stack is possible:  (1) given that there is at least one stack of exterior $\cgSE$'s to the west of it, and that it is an exterior stack, then it is geometrically impossible for a path which does not grow into a concavity of $\alpha$ (which must be the case for all paths after some finite number of layers since, as previously discussed, all concavities become covered up) to grow eastward to, and/or along the outside of, a $\cgSE$ in the easternmost stack since all paths only grow CCW (and note that this is the case for any exterior stack which has another to its west), and (2) as the easternmost stack, paths growing upward from $\cgSE$'s of $s_e$ cannot collide with other $\cgSE$'s from a stack to the north or east.  Therefore, both situations which result in vertical doubling rows are impossible for the $\cgSE$'s of $s_e$.  This means that for each path which is output from $s_e$, the southeast corner of the lowest $\cgSE$ in $s_e$ is offset from the previous by the vector $(1,-1)$.

Now, we note that after some finite amount of outward growth of $s_w$, since each additional $\cgSE$ must be placed at an $x$-coordinate which is at least 1 greater than the $\cgSE$ which preceded it, the $x$-coordinate of the eastern tiles of a $\cgSE$ in $s_w$ must be greater than or equal to the $x$-coordinates of some $\cgSE$'s in the stack immediately to its east. An important property of the stacks of $\cgSE$'s is that they all grow along the slope $(1,-1)$ (although sometimes at the vector $(2,-2)$), so no two stacks can have $\cgSE$'s which collide.  However, as soon as the $\cgSE$'s of one stack fall to the immediate south of $\cgSE$'s of a stack to its east (and again, those $\cgSE$'s cannot have paths growing into their west sides after some finite amount of outward growth), then the western stack will start having vertical doubling rows between each pair of consecutive $\cgSE$'s.  This will cause the offsets between them to be $(2,-2)$.  Assuming that there are $n$ stacks, it is possible that each of the western $n-1$ has grown to a position where it is growing over $\cgSE$'s of the stack to its east.  However, as previously mentioned, the easternmost (i.e. $s_e$) can not be growing vertical doubling rows and grows downward by the vector $(1,-1)$.  Since the paths from $s_e$ must eventually grow to $\cgSE$'s in $s_w$, there can only be one $\cgSE$ added to $s_e$ for each added to $s_w$, since the only ways that a $q$ message can be passed back to $s_e$ and initiate growth of another layer require that a $\cgSE$ be added to $s_w$ (i.e. there must be a collision with a $\cgSE$ of $s_w$ or a path initiated by it, which will cause a new layer to be initiated in $c_w$ as well). Now, as noted, each additional $\cgSE$ of $s_w$ eventually grows at a vector of $(2,-2)$ relative to the previous while those of $s_e$ can only grow at the vector $(1,-1)$.  (See Figures~\ref{fig:new-frame2-full-example2} and \ref{fig:new-frame2-full-example2-disconnect} for an example.)  As we are interested in proving that no race condition can occur which allows multiple stacks to grow indefinitely, we will show how
$s_w$ covers those to its east until only $s_w$ remains.  Each stack between $s_w$ and $s_e$ must eventually grow so that they are each placing $\cgSE$'s directly south of stacks which are east of them.  For them not to eventually cover $s_e$, they must grow at the same vector as $s_e$, namely $(1,-1)$ and grow layers at no greater a rate than $s_e$.  However, $s_w$ must grow layers at the same rate as $s_e$, since $s_e$ is not able to grow new layers without continued growth of $s_w$ (as previously shown), and $s_w$ must eventually grow strictly at the vector $(2,-2)$.  This ensures that $s_w$ eventually covers all stacks and becomes the sole stack of $\cgSE$'s.  Otherwise, the stacks between $s_w$ and $s_e$ can grow faster than $s_e$ and eventually cover it, but then as each becomes the new easternmost stack, its growth must slow relative to $s_w$ (as it becomes the new $s_e$ and behaves exactly as the previous), eventually ensuring that all stacks are covered by $s_w$.

The final thing to show is that the paths growing from $s_w$, the only remaining stack of $\cgSE$'s, will eventually become rectangular.    Assuming that at the point at which $s_w$ becomes the single stack, the next layer it grows is not rectangular, then there are two scenarios to consider.  The most recently placed $\cgSE$, $c_x$, grows a path CCW which eventually (1) collides with a $\cgSE$, say $c_y$, in $s_w$ which is north of $c_x$ in the stack (and thus was placed earlier), or (2) grows to the south of $c_x$.  In (1), each path which collides with such a $\cgSE$ will receive a $q$ signal since the path which grows outward from $c_y$ is not rectangular and therefore as previously shown must encounter a collision which ensures the propagation of a $q$ message (even if this message must pass around $\alpha$ multiple times and through multiple $\cgSE$'s).  This will cause a (horizontal) doubling row to grow along the south side of the path to the west of $c_y$.  As this doubling row must grow before another layer can grow across and collide with another $\cgSE$ in $s_w$, and that doubling row covers up one of the $\cgSE$'s in $s_w$, the vertical distance between $c_x$ and $c_y$ will be one greater than the distance between the path which grows from the $\cgSE$ that attaches to $c_x$ and the $\cgSE$ which it collides with.  Eventually, with the distance decreasing by one each time, it will become $0$, meaning that a path grows out of the north of a $\cgSE$ and back into the west of that same $\cgSE$, thus it is a rectangle.  In (2), since $c_x$ cannot initiate the growth of another layer (via the attachment of a $\cgSE$ to its outside) without its path colliding with something, there must be no $\cgSE$ to the south of $c_x$ in $s_w$ (i.e. it would have needed to be there before the collision which enabled it be placed there).  Therefore, the path must grow across the outside (i.e. on the south side) of $c_x$.  In this case, as shown in Figure~\ref{fig:collision-detection-corner-cases2} and previously discussed, a vertical doubling row will be initiated and a new $\cgSE$ will be allowed to attach to the east and south of the bottom of it.  In this way, the path which grew around and back below the $c_x$ terminates and a new $\cgSE$ initiates a row which now must be exactly rectangular.

Thus, we have shown that all possible assembly sequences must eventually result in a rectangular row after a finite number of tile additions. \qed \end{proof}

\begin{figure}
\centering
\subfloat[
\label{fig:new-frame2-full-example2}]{%
\includegraphics[width=1.35in]{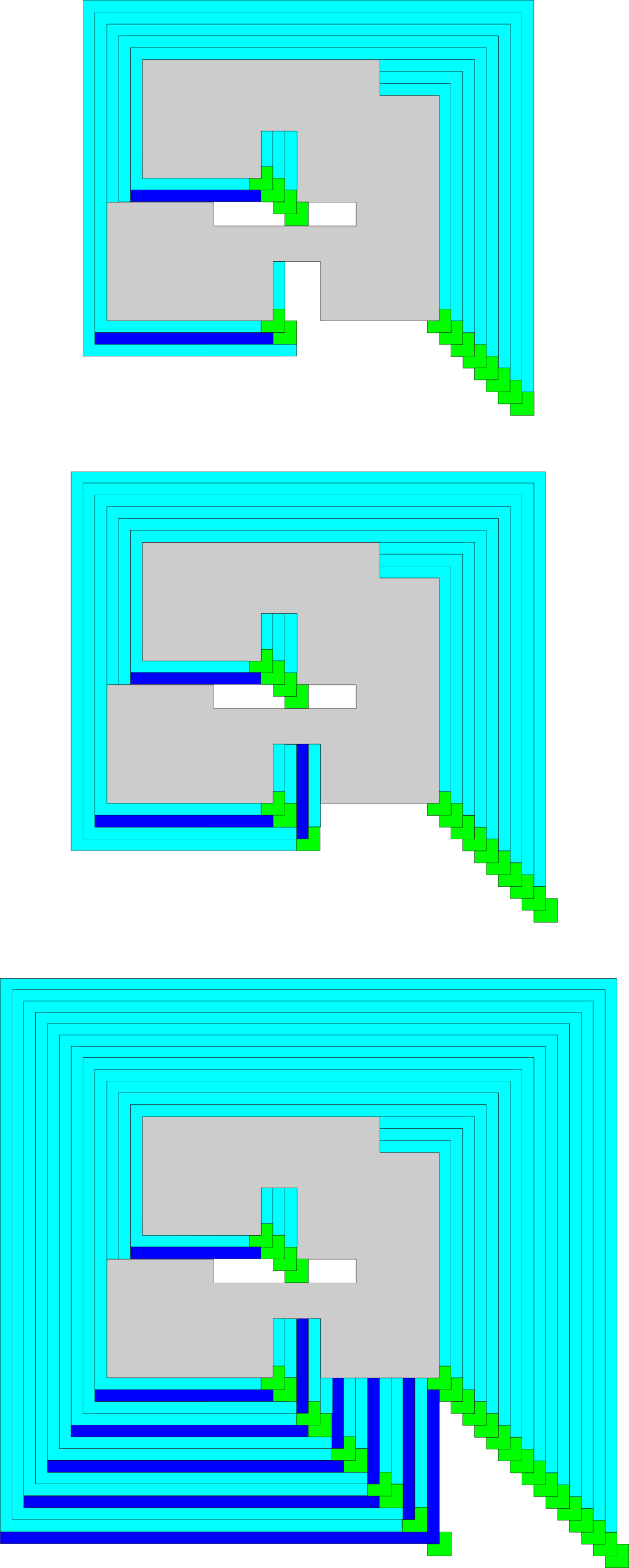}%
}%
\qquad%
\subfloat[
\label{fig:new-frame2-full-example2-disconnect}]{%
\includegraphics[width=3in]{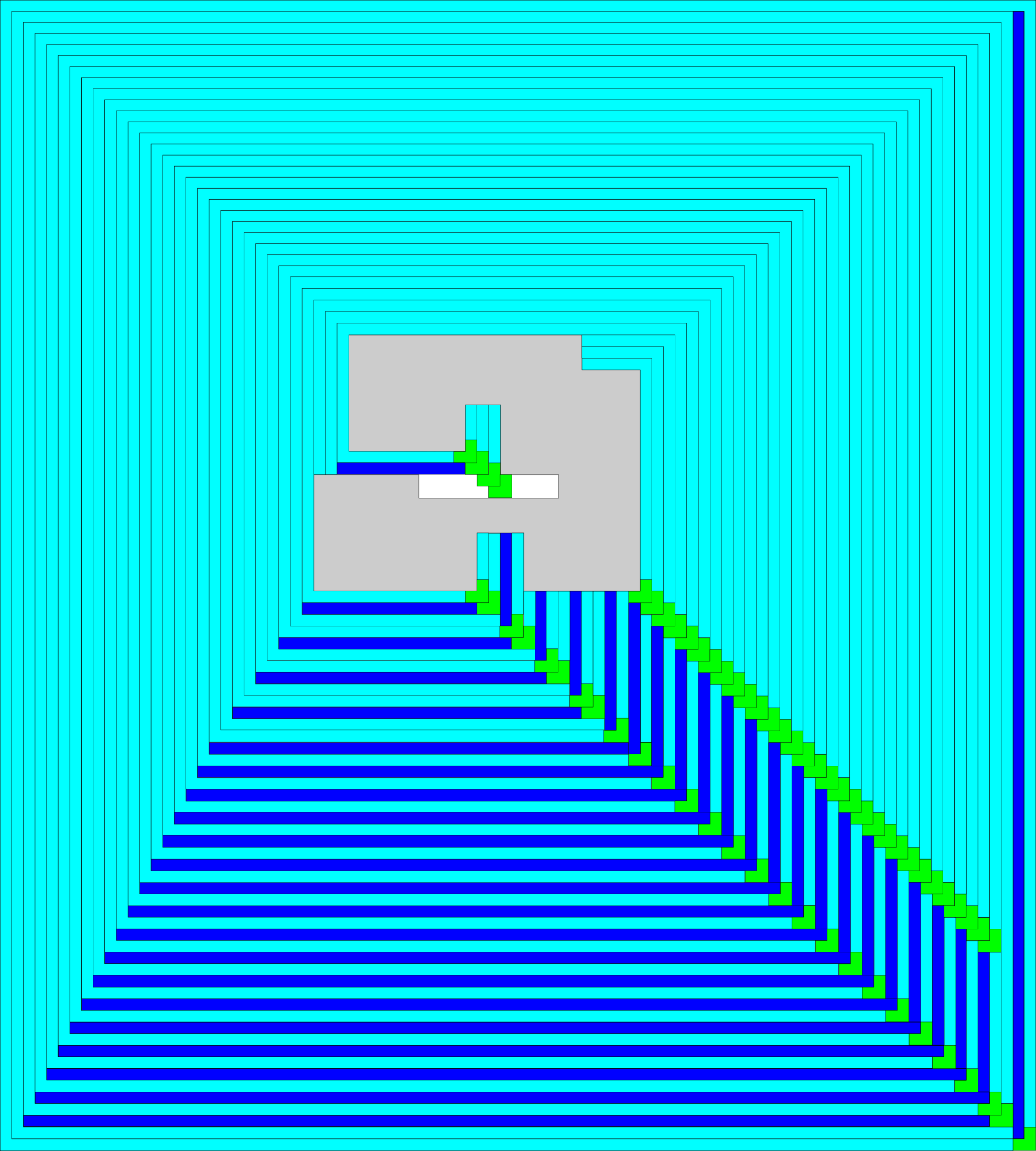}%
}
    \label{fig:frame-full-ex}
    \caption{ (a) The full frame for the example of Figure~\ref{fig:new-frame2-full-example-begin} showing an alternate assembly sequence 
    (from that of Figure~\ref{fig:new-frame2-full-example}) for the growth of the frame around $\alpha$ until it terminates in the first rectangular layer.  
    In this assembly sequence, the eastern stack of $\cgSE$'s continues growth much further than in the other.  
    In particular, the $\cgSE$'s of the eastern stack win the maximum number of race conditions possible, 
    i.e. the rows which grow from them grow very quickly and $q$ messages travel back quickly, so new $\cgSE$'s are able to attach before the $\cgSE$'s of the west 
    stack can grow past them, which could occur if a vertical doubling row from the west stack is able to quickly grow northward and collide with the south 
    side of the $\cgSE$ which caused it to be initiated, preventing the east stack's next $\cgSE$ from attaching. By winning these race conditions, 
    the east stack is able to grow to a maximal depth.  Nonetheless, the depth is still bounded since the west stack can only catch up and surpass it by use of both
    vertical and horizontal doubling rows. In this example, the standard path overtake the western $\cgSE$ before they are able to grow a standard path north, causing
    the second case of doubling row to be activated from Figure~\ref{fig:vdr1}.
     This causes a repeating pattern of vertical and horizontal doubling rows until the western $\cgSE$ begin to overtake the eastern $\cgSE$, at which case the
     vertical doubling row instance from Figure~\ref{fig:vdr2} occurs. 
Light grey: $\alpha$, Green: $\cgSE$, Light blue: standard path, Dark blue: doubling row. 
(b) The example from Figure~\ref{fig:new-frame2-full-example2} completed. 
Note that at the second to last western $\cgSE$, the second vertical doubling row case from Figure~\ref{fig:vdr1} required to ensure the outermost layer is rectangular.
\vspace{-10pt}}
\end{figure}

\begin{lemma}\label{lem:rectangle-detach}
For each STAM system $\calT = (T,\alpha,2)$ with tile set $T$ (i.e. the frame building tile set) and input assembly $\alpha$, as a frame assembly grows, no subassembly containing tiles of $\frametiles$ can completely detach from the assembly unless it consists of a path which forms a complete and exact rectangle.  Furthermore, the first such rectangular path will detach and then no further layers will grow.
\end{lemma}

\begin{proof}
To prove Lemma~\ref{lem:rectangle-detach}, we first show that no frame path $p$ which is non-rectangular can ever detach.  As mentioned in the proof of Lemma~\ref{lem:rectangle-layer}, every frame tile (i.e. those in $\frametiles$) which attaches to the assembly does so by making one of its initial bonds to $\alpha$, a tile of a layer which is between it and $\alpha$, or a doubling row.  If it is a $\cgSE$, it makes two such bonds, and otherwise it makes one of its initial bonds to another tile of the same path.  Thus, before any glue deactivations, all tiles are stably attached to each other and the assembly to their inside (i.e. left sides).  The only way for a frame tile to detach, therefore, is for one or more glues to be deactivated, and the only glues that can be deactivated on a frame tile are (1) the $x2$ glue binding it to the tile to its inside, and (2) the $c$ glue on its outside (i.e. right side) which can attach to a collision detection tile.  Since the $c$ glue only binds (temporarily, before being deactivated) to a collision detection tile, its deactivation is not sufficient to allow a frame tile - possibly along with an attached subassembly - to detach.  Therefore, we now discuss the deactivation of the $x2$ glue, which occurs only as a $d$ message is passed CW through a path.  A $d$ message is initiated only when a frame path grows around three convex corners (NE, NW, and SW) and then places a tile so that its front (east) side is adjacent to the west side of a $\cgSE$, and then a special collision detection duple attaches to that $\cgSE$ and path tile.  This binding causes the west tile of the $\cgSE$ to activate a $d$ glue on its west side, which attaches to the $d'$ glue on the east of the colliding path tile.  This initiates the passing of the $d$ message through the entire path, back to the $\cgSE$ which initiated it.  Each tile which the $d$ message passes through deactivates the $x2$ glue on its left side and continues passing the message.  It is important to note that the $d$ and $d'$ glues are strength-2, so as a path detaches from the tiles inside of it, the tiles of the path itself remain stably connected to each other, as well as to the $\cgSE$ which initiated the $d$ message. (Note, the frame tiles which bind to $\alpha$, i.e. those of layer 1, do not pass the $d$ message, so the innermost layer does not detach following this procedure.)  In this way, the tiles in a path can detach from the assembly inside of them, but they remain bound to the $\cgSE$ which initiated the $d$ message.  The only way that a $\cgSE$ will deactivate its $x2$ glues is if it receives a $d$ message from its north, and it can only do that once it has already sent a $d$ message out its west side (since the glue that receives the signal from the north is only activated after the $d$ glue on its west has bound).  While a series of paths and $\cgSE$'s to which they are attached may cause a $d$ message to make multiple CW passes around the assembly, causing most of that connected path to detach from the inner assembly, that path is ultimately still connected to the most CCW $\cgSE$ which initiated a $d$ message on that path.  However, any path which is not part of a layer which is not perfectly rectangular must, by following this path, be connected to a $\cgSE$ whose north-growing path terminated in a collision.  That means that that final $\cgSE$ must never be able to receive a $d$ message from its north (and will only receive the $q$ message initiated by the collision), and so will never detach, ensuring that the entire path which may be connected to it through its west side also remains attached.  (See Figure~\ref{fig:new-frame2-full-example-disconnect} for an example.)  Therefore, no assembly containing $\frametiles$ can detach from the assembly if they are not part of an exactly rectangular layer.

Second, if a path grows from the north side of a $\cgSE$ $c$, around $\alpha$, and arrives with the front of its final tile at the west side of $c$, then it could only have grown as a perfect rectangle.  As previously mentioned, a special collision detection duple will attach to the final tile of the path and the western tile of $c$, and this attachment will initiate the CW propagation of a $d$ message which will exactly follow the path backward to the north side of $c$.  Each tile that the $d$ message passes through will deactivate the $x2'$ glue on its left side, and once that message arrives at the north side of $c$ it will cause $c$ to activate a $p$ glue which will ``prime'' the $\cgSE$ to the inside of it, whose binding will then cause one $x2'$ glue of $c$ as well as the $x2$ glue of the other $x2$-$x2'$ pair of the $\cgSE$ inside it to be deactivated.  This will result in all glues binding that layer to the layer inside of it, other than the $p$ glue, to turn $\off$ (other than one of those on its $\cgSE$, but the glue it was bound to on the inner $\cgSE$ will turn $\off$), causing the layer to be attached with only a single strength-1 glue and therefore to detach.  The detached layer will only at that point have one active $x2$ glue and one active $p$ glue, and therefore nothing can bind to it since all $p'$ glues which have been activated on other layers are hidden by outer layers.  (Additionally, no $q$ message will ever be initiated or received by that layer, and therefore no additional layer will be able to grow on its outside.)  Finally, the remaining frame assembly, without the outer layer, must only have one of the $x2$ glues on its $\cgSE$ $\on$ and therefore no new $\cgSE$ can attach to begin re-growth of another layer.  Only the next priming tile can attach.
\qed
\end{proof} 

\begin{lemma}\label{lem:marking}
For each STAM system $\calT = (T,\alpha,2)$ with tile set $T$ (i.e. the frame building tile set) and input assembly $\alpha$, as a frame assembly grows, the eventual detachment of the rectangular layer will result in a single tile of the frame, one which is adjacent to a location of $\alpha$, activating a $z$ glue.
\end{lemma}

\begin{proof}
Given Lemmas~\ref{lem:rectangle-layer} and \ref{lem:rectangle-detach}, $\calT$ is guaranteed to form a rectangular frame layer which will detach.  That layer will have exactly one $\cgSE$, and before it can detach that $\cgSE$ will, by the required binding of its $p$ glue, ensure that the $\cgSE$ to which it was attached will activate a $p2'$ glue.  Once the outer layer has detached, the only tile that can attach to the now outer layer of the frame is the ``primer'' tile with $p2$.  Upon binding to the $\cgSE$'s $p2'$ glue, a cascade of messages will begin.  See Section~\ref{sec:leader-details} for specific details of the tiles that pass these messages, but beginning from the $\cgSE$ an $sp$ message will be sent along the entire bottom row of the assembly.  Once it reaches the western end, an `$L$' tile will attach which will cause an $spR$ message to be sent all the way back to the right side.  At this point, an `$R$' tile will attach and cause the tile it binds to to pass the $sp$ message one row to the north and then begin the process again.  This will occur until a westward traveling $sp$ signal eventually arrives at the easternmost of the southernmost $\cgSE$'s which are attached directly to $\alpha$.  Then, that $\cgSE$ is ``marked'' as the leader so that it will expose the single $z$ glue, and the $sp$ and $spR$ messages will no longer be passed.  By the facts that (1) there must be some $\cgSE$ which is the easternmost of the southernmost, (2) the $sp$ signal cannot collide with any tile of $\alpha$ before it stops at the southeasternmost $\cgSE$ (because of the fact that any signal traveling in that pattern must collide with a tile at that position of $\alpha$ and such a convex corner must have a $\cgSE$ attached to it in order for the subsequent layers to have grown, and (3) all frame building tiles can pass these messages, the correct corner gadget must receive the message.  Upon receiving the message, it will activate the necessary $z$ glue.
\qed
\end{proof}

\begin{lemma}\label{lem:full-mold}
The leader tile will initiate a message which will propagate around $\alpha$ and guarantee that for every tile location directly or diagonally adjacent to a tile in $\alpha$, layer 1 contains a tile in this location. Moreover, after all pending signals of the $br$ message have fired, the entire frame assembly detaches from $\alpha$ exposing on its interior a unique glue activated only on the leader.
\end{lemma}

\begin{proof}
All of the perimeter edges of $\alpha$ can be enumerated starting from the north edge of the leader tile and proceeding to each consecutive incident edge in the CCW direction. Then, by the construction given in Section~\ref{sec:outlining}, as the $g$ message propagates around $\alpha$ CCW starting from the leader tile, if the only tiles that are allowed to attach to the existing assembly ($\alpha$ and attached frame building tiles) are the singletons given in Section~\ref{sec:outlining}, then it is clear that each $x$ glue on each consecutive edge enumerated in the CCW direction either binds to a tile passing the $g$ message or binds to a tile via an $x$ glue activated by one of signals of the $g$ message. Also, by the construction presented in Section~\ref{sec:leader-details}, we can ensure that the north edge of the leader tile is the only edge that exposes some unique glue. Then since the $br$ message follows $g$ deactivating each $x'$ glue, after the frame disassociates from $\alpha$, the leader tile will expose a unique glue.

Therefore, to complete the proof of this lemma it suffices to show that as the $g$ message propagates around $\alpha$ CCW starting from the leader tile, the only tiles that are allowed to attach to the existing assembly ($\alpha$ and attached frame building tiles) are the singletons given in Section~\ref{sec:outlining}.
Note that during the propagation of the $g$ message, at most a single $x'$ and $g$ glue may be exposed on an edge of the perimeter of the existing assembly. Therefore, some (possibly much smaller) shape could not bind via this $x'$ glue since it only has strength 1. Then note that if even though two such assemblies exist, they possibly expose $x$ glues and an $x'$ glue. However, these two assemblies cannot bind because the edges that expose these glues on each assembly must belong to tiles that are completely surrounded by tiles of one or more frame layers. Therefore, an assembly that exposes an $g$ and $x'$ glue on its perimeter can only bind to one of the singleton tiles or a 2-tile corner gadget described in Section~\ref{sec:outlining}.
\qed
\end{proof}

\begin{lemma}\label{lem:no-combine}
Let $\beta$ and $\gamma$ be distinct producible assemblies of $\mathcal{T}$ that each have an input assembly as a subassembly. Then $\beta$ and $\gamma$ cannot bind.
\end{lemma}

\begin{proof}

During the assembly of frame layers, many of the glues are activated (mostly for message passing purposes) only after the edge that the glues are on abut an edge of some tile. As long as these glues are never exposed on the perimeter of an assembly, these glues cannot allow for two producible assemblies to bind. Then, by construction, all of glues that are exposed on the exterior of an assembly have specific purpose and only allow for either some singleton tile, 2-tile gadget or 3-tile gadget to bind. For example, we use collision detection tiles/duples to minimize exposed $x'$ glues.
The one exception of this is that when layer 1 is propagating a $g$ message as described in Section~\ref{sec:outlining}, a single $x'$ glue may be exposed on the perimeter of a producible assembly. As noted in the proof of Lemma~\ref{lem:full-mold}, the only tile that may bind to a producible assembly with an input assembly as a subassembly is a singleton tile or a 2-tile corner gadget described in Section~\ref{sec:outlining}.
\qed
\end{proof}

In Theorem~\ref{thm:complete-frame}, let $\mathcal{T} = (T,\sigma,2)$ be a STAM system such that $T$ is the frame building tile set given by the construction in this section and $\sigma$ is a set of input shape assemblies.
Given an input shape assembly $\alpha$, and a TAS $\mathcal{T}$, we say that a terminal assembly $\beta$ is a \emph{completed frame assembly for} $\alpha$ \emph{ that exposes a unique glue}, $z$ say, if (after some translation of the locations in $\dom(\beta)$) there is a single finite connected component of the complement of $\dom(\beta)$ in $\Z^2$ that is equal to $\dom(\alpha)$, and moreover,  the tiles of $\beta$ at locations adjacent to locations of $\dom(\alpha)$ expose glues on the perimeter of $\beta$ such that there is a single tile (a leader tile) which exposes $z$ on a single edge of the perimeter of $\beta$.

\begin{theorem}\label{thm:complete-frame}
For any finite set $\sigma$ of input shape assemblies, if $\beta\in \termasm{T}$, then either $\beta$ is in $\sigma$, $\beta$ is a completed frame assembly for $\alpha$ that exposes a unique glue, which we denote by $z$, or $\beta$ does not expose $z$ on its perimeter (for example, the rectangular frame layer that detaches in our frame construction).
\end{theorem}

\begin{proof}
This follows directly from the Lemmas~\ref{lem:rectangle-layer},~\ref{lem:rectangle-detach},~\ref{lem:marking},~\ref{lem:full-mold}, and \ref{lem:no-combine}.
\end{proof}

\section{Replication} \label{sec:rep}
Having demonstrated the ability to create a frame from any hole-free 2D shape within a constant number of steps, we can utilize this tileset to carry out a variety of tasks related to unique shape recognition. We seek to replicate input shapes by filling in the frames generated, using them as a mold to create shapes of identical form. With careful construction, we can carry out infinite replication of shapes from only a single copy of a shape.

In this section, we provide a construction in the STAM which is capable of replicating the shapes of an arbitrary set of hole-free input assemblies. Furthermore, we show this replication can be done in an exponential manner. A high-level overview of the process by which this construction creates a copy of an arbitrary shape given by some frame is as follows. First, tiles grow in a counter-clockwise manner inside a frame beginning at the southeast corner which is detectable for any hole-free shape as indicated by Lemma \ref{lem:marking}. As this tile growth progresses, it generates a message which follows the edges of the frame. Once the message is detected by the starting tile, a second message passed through the recently placed tiles causes each tile to detach from the frame, activate new exterior glues, and fill in any internal voids.

\subsection{Definition of exponential, unbounded replication}\label{sec:exp-rep}

\begin{wrapfigure}{R}{0.4\textwidth}
    \centering
    \includegraphics[width=0.35\textwidth]{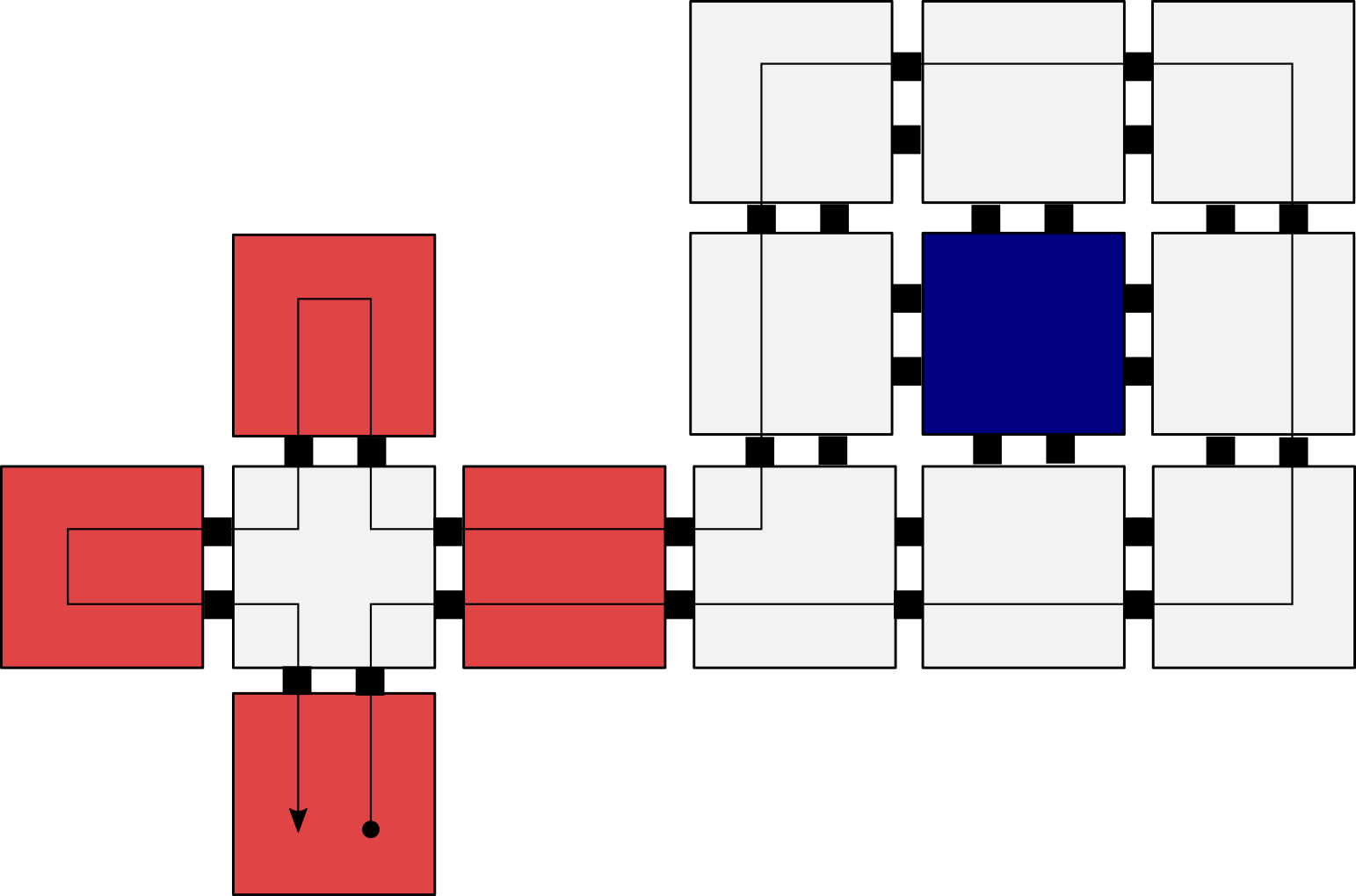}
    \caption{An example width-1 shape, $P$. For an assembly $\sigma$ such that $\dom{\sigma} = P$, red tiles of $\sigma$ fulfill the criterion of having two opposite edges in $PERIM(\sigma)$. The blue tile in the shape will be added by a fill tile. The line starting and finishing at the south eastern-most tile demonstrates the order of tiles visited by the $w$ message in the process of generating the ring.\vspace{-20pt}}
    \label{fig:w1-example}
    
\end{wrapfigure}

We now define what it means for an STAM system to replicate a shape $P$.  
In the following definition, we assume that the input assemblies $\sigma_1, \sigma_2, ..., \sigma_n$ all have the same glue on each edge of their perimeter 
and there is only one of each in the system.  We say that an STAM system $\mathcal{T} = (T, \{\sigma_1, \sigma_2, ..., \sigma_n\}, \tau)$ infinitely 
replicates the shapes $P_1, P_2, ..., P_n$ provided that 1) $\dom{\sigma_i} = P_i$ for all $i$, 2) there are a constant number of different terminal assemblies, 
and 3) there are an infinite number of assemblies of shape $P_i$ for all $i$.  We call assemblies $\sigma_1, \sigma_2, ..., \sigma_n$ replicable supertiles.  
In general, a replicable supertile in a system $\mathcal{T}$ is a supertile which can be ``replicated''.

In the replication process for a shape $P$, we consider a time step to be the worst-case number of tile attachments to produce a new assembly $\alpha$ with $\dom{\alpha} = P_i$ starting from an input assembly $\sigma_i$.  We say that an STAM system $\mathcal{T}$ \emph{exponentially replicates without bound} the shapes $P_1, P_2, ..., P_n$ given four properties. First, the system $\mathcal{T}$ infinitely replicates the shapes $P_1, P_2, ..., P_n$. Second, the number of (super)tile attachments and/or detachments required in the creation of a copy of shape $P_i$ (i.e. one time step) is bounded by $poly(|P_i|)$. Third, every (super)tile attachment involves a singleton tile or a two tile assembly and an existing assembly that is replicating $P_i$ for some $i$ (i.e. tiles either attach to a frame layer or attach as filler tiles); and moreover, all supertile detachment involves either 
a completed frame detaching from a seed  assembly $\sigma_i$ corresponding to $P_i$ for some $i$, or a completed frame detaching from a replicated shape assembly  corresponding to $P_i$ for some $i$. Finally, if there are $n$ copies of assemblies with the shape $P_i$ at time step $t$, then there are $2n$ copies of assemblies with the shape $P_i$ at time step $t+1$. Though this definition is specific to the replication technique given by our construction, it is motivated by the observation that for a single shape $P$, if a system $\mathcal{T}$ exponentially replicates without bound the shape $P$, then $\mathcal{T}$ exponentially replicates the shape $P$ according to the more general definition of exponential replication given in~\cite{SignalsReplication} (Definition 3).

We now provide a high-level overview of the construction.  Let $P$ be a connected, hole-free shape. Our construction begins with utilizing the tile set described in Section \ref{sec:frame-building} that builds a frame around $\sigma$ where $\dom(\sigma) = P$. Next, we design tiles so that a \emph{ring}, all tiles of $\sigma$ which contain at least one edge of $PERIM(\sigma)$, forms within the completed frame.  Once this ring is complete, it detaches and fills in its interior. The frame is then free to make another copy while the newly formed ring can host the growth of another frame. For descriptive purposes, we define \emph{width-1} shapes as shapes such that there exists at least one tile $t \in \sigma$ such that two opposite edges ($N$ and $S$ or $E$ and $W$) are both in $PERIM(\sigma)$. An example width-1 shape is shown in Figure \ref{fig:w1-example}. 

\subsubsection{Frame Tileset}
Our construction begins by adding tiles to our tile set which build a frame around $\sigma$ using the machinery from Section \ref{sec:frame-building}. We design our tile set for this frame system so that once the frame completes, it detaches from $\sigma$ with 1) exactly one south east corner of the frame which exposes a $z$ glue on the west face of the northernmost tile of the $\cgSE$, and 2) on every other exposed interior tile of the frame, a $y$ and $w^\prime$ glue are active. We note that due to the asynchronous nature of STAM signal activation, all $y$, $w^\prime$, and $z$ glue(s) may not be active on the frame at the same time. The ring creation tileset design considers this feature of the model by preventing ring growth until both $y$ and $w^\prime$ glues on the frame have been activated.

\subsubsection{Ring Creation}\label{sec:ring-form}

Since we have shown it is possible to determine the southeast most corner of a shape in Lemma \ref{lem:marking}, we utilize that as our starting point to build the ring of our shape inside the exposed frame. The overall process by which this occurs can be thought of as the following of the perimeter in a CCW direction - since we have a hole-free connected shape as an input, this will ensure that once the start tile is reached again, every tile in the ring has been reached. This guarantee is critical to ensure that the filling and detachment process does not begin spuriously. 


The initial set of tiles which creates a ring is demonstrated in Figure \ref{fig:rep_tiles}. In order for the ring to sense that it is complete and can begin the jettisoning process, we must pass a message CCW which starts at, and returns to, the start tile shown in Figure \ref{fig:rep1d_start}. To complete the tile set for the ring, we add copies of the crawler and duple tiles rotated 90, 180, and 270 degrees. 

\begin{figure}
\centering
  \subfloat[][The tile which begins the growth of the ring.]{%
        \label{fig:rep1d_start}%
	        \includegraphics[width=0.9in]{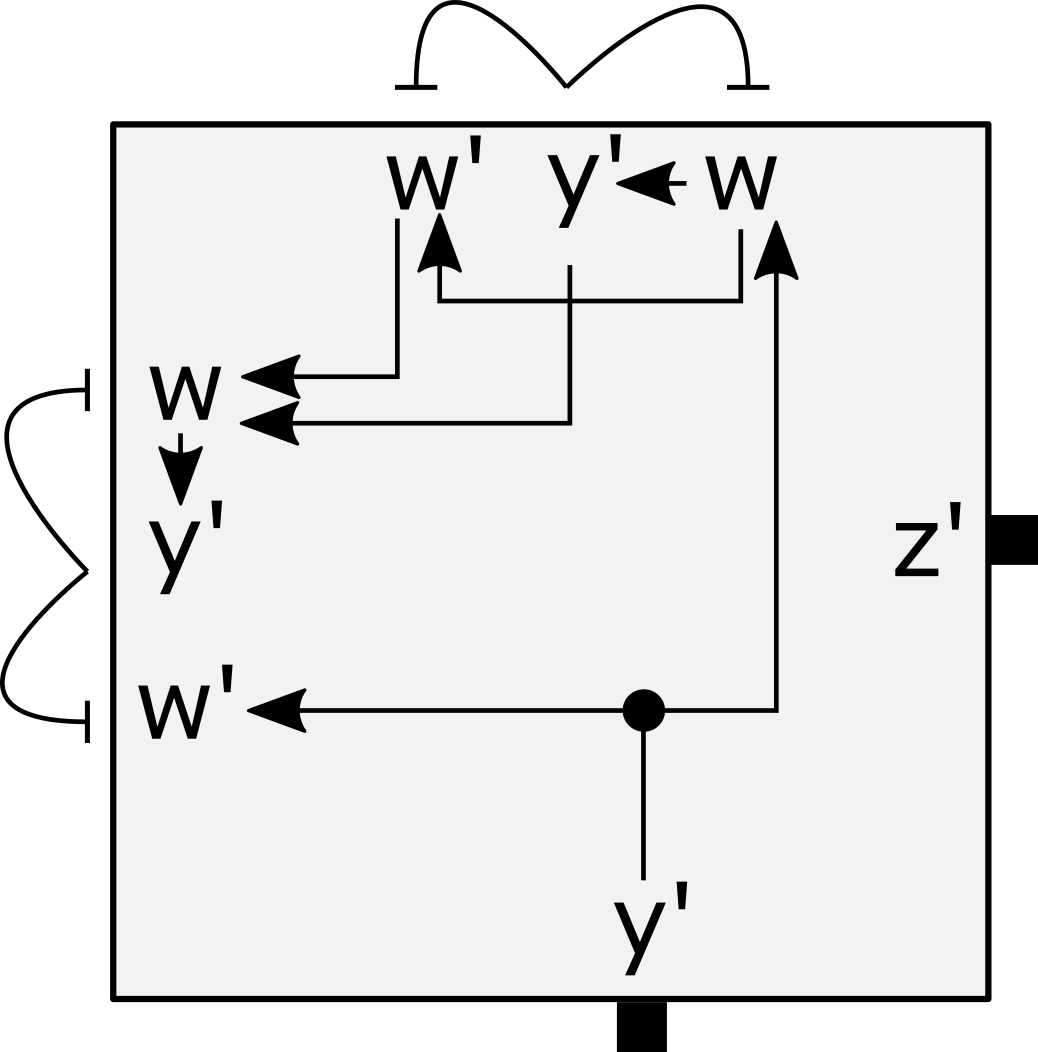}
        }%
        \quad\quad
  \subfloat[][A north crawler tile which propagates the $w$ message north.]{%
        \label{fig:rep1d_crawl}%
        		\includegraphics[width=0.9in]{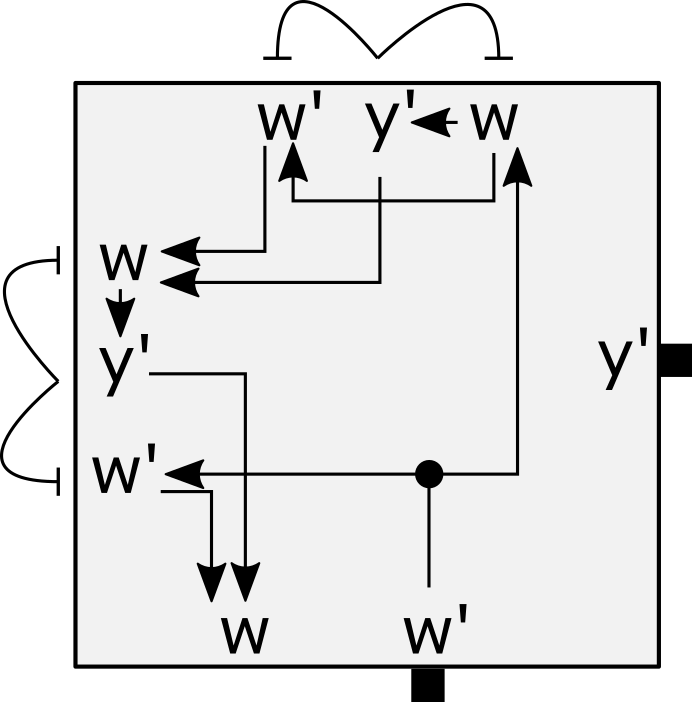}
        }%
        \quad\quad
  \subfloat[][The ``duple'' which allows the $w$ message to propagate to the east when there is not a $y$ glue for a north crawler tile to cooperatively bind to.]{%
        \label{fig:rep1d_duple}%
        		\includegraphics[width=2in]{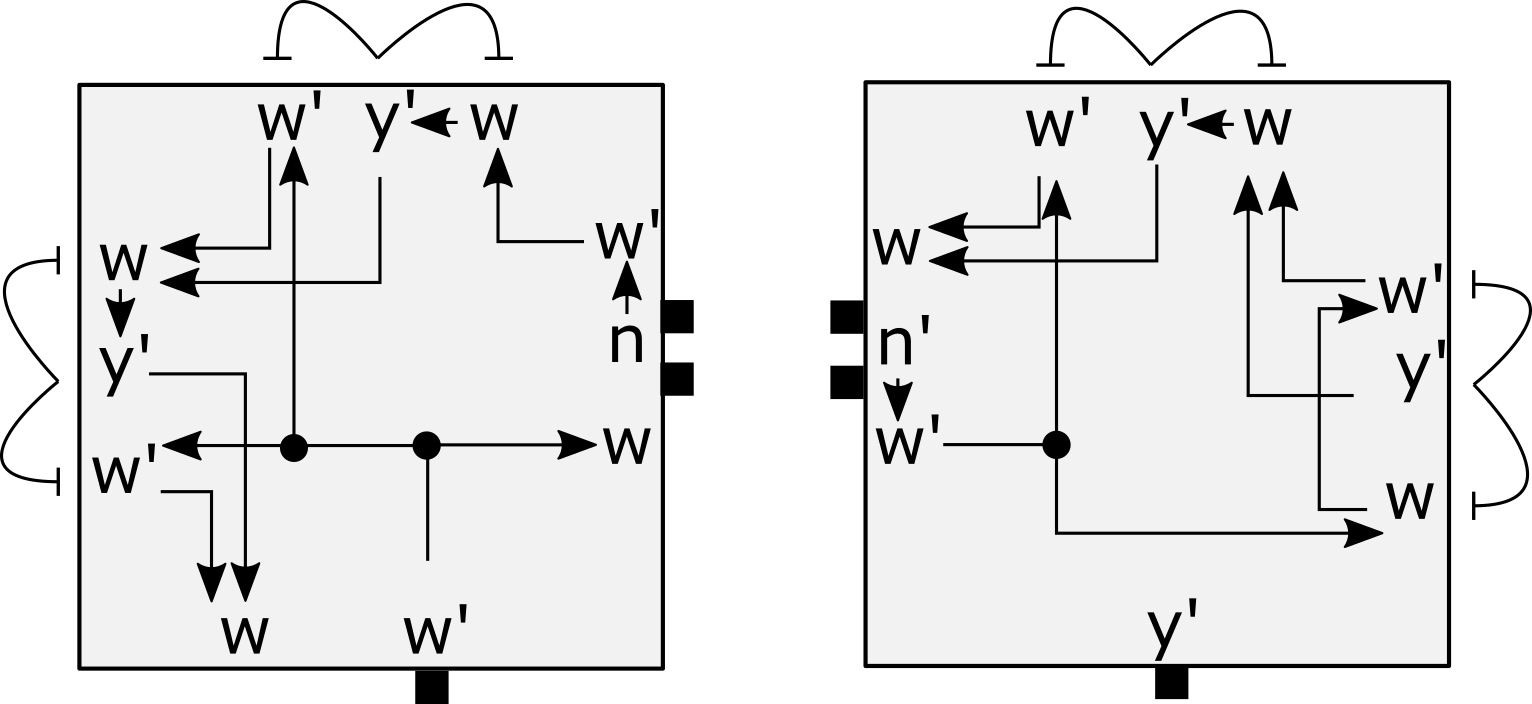}
        }%
        \quad\quad
  \subfloat[][The ``filler'' tile.]{%
        \label{fig:rep1d_fill}%
                \includegraphics[width=0.83in]{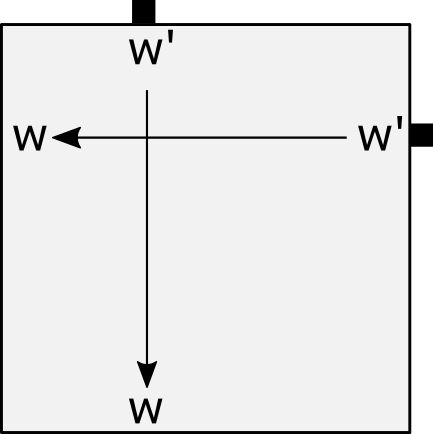}
        }%
  \caption{The partial tile set which builds a ring and fills it.\vspace{-15pt}}
  \label{fig:rep_tiles}
\end{figure}

After the start tile is bound to the southeastern-most corner of the frame, a $w$ glue is exposed on its north side. The exposed $w$ glue provides a location for a north crawler tile (Figure \ref{fig:rep1d_crawl}) to cooperatively bind with a frame tile or an overhang duple (Figure \ref{fig:rep1d_duple}) to follow the perimeter along a 90 degree turn to the right. Alternatively, if the start tile encounters a frame tile to the north, it binds with that frame tile's $w^\prime$ glue using its north $w$ glue, and then activates the north $y^\prime$ glue. Note that the $y^\prime$ glue is activated in both instances, but will only bind in the case of being exposed to a frame tile due to geometric hindrances (i.e. being hidden between ring-forming tiles). At the same time the $w$ glue is sent an activation signal, a $w^\prime$ glue is exposed to the west side to sense for an incoming crawler tile. After either a crawler or duple tile binds to the start tile, that most recently added tile then exposes a $w$ glue to continue the building of the ring. This process of sequentially exposing a $w$ glue is propagating the $w$ message, and allows for the ring to be built such that it follows the rightmost edge of $PERIM(P)$ using all rotations of the crawler and duple tiles. An example of the ring building process is provided in Figure~\ref{fig:rep-build-ex}.

\begin{figure}
    \centering
    \subfloat[Start tile (Figure~\ref{fig:rep1d_start}) added via cooperative binding of $z$ glue (middle right green) and $y$ glue (bottom green) on frame tiles. $w$ glue (fuchsia) on start tile exposed for addition of next tile]{%
        \includegraphics[width=1.75in]{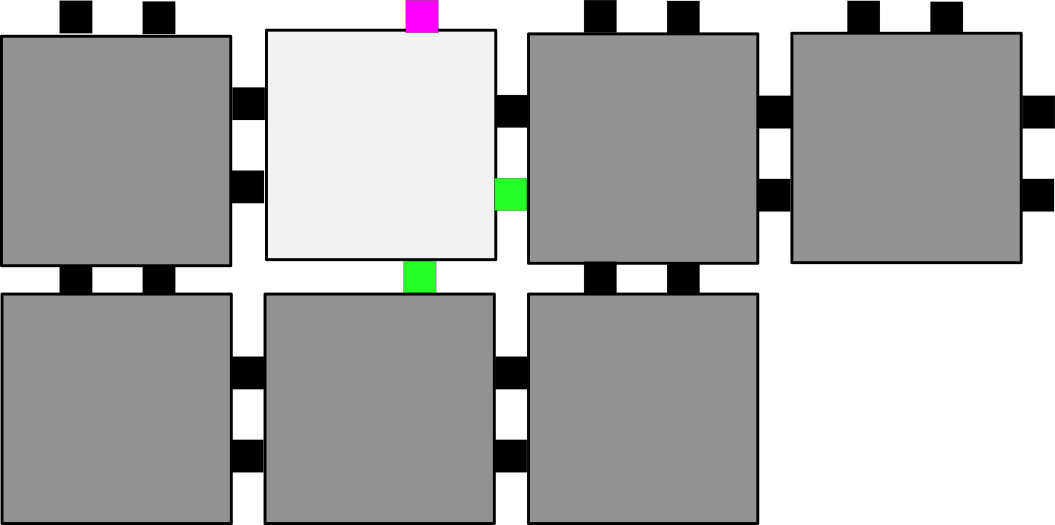}
        \label{fig:rep-build-ex1}
    }%
    \quad%
    \subfloat[Duple tile assembly (Figure~\ref{fig:rep1d_duple}) added via cooperative binding to prior exposed $w$ glue (bottom left green) on start tile and $y$ glue (bottom right green) on frame tile. $w$ glue (fuchsia) on right duple tile follows frame, and $w^\prime$ glues (yellow) exposed to detect incoming $w$ messages.]{%
        \includegraphics[width=1.75in]{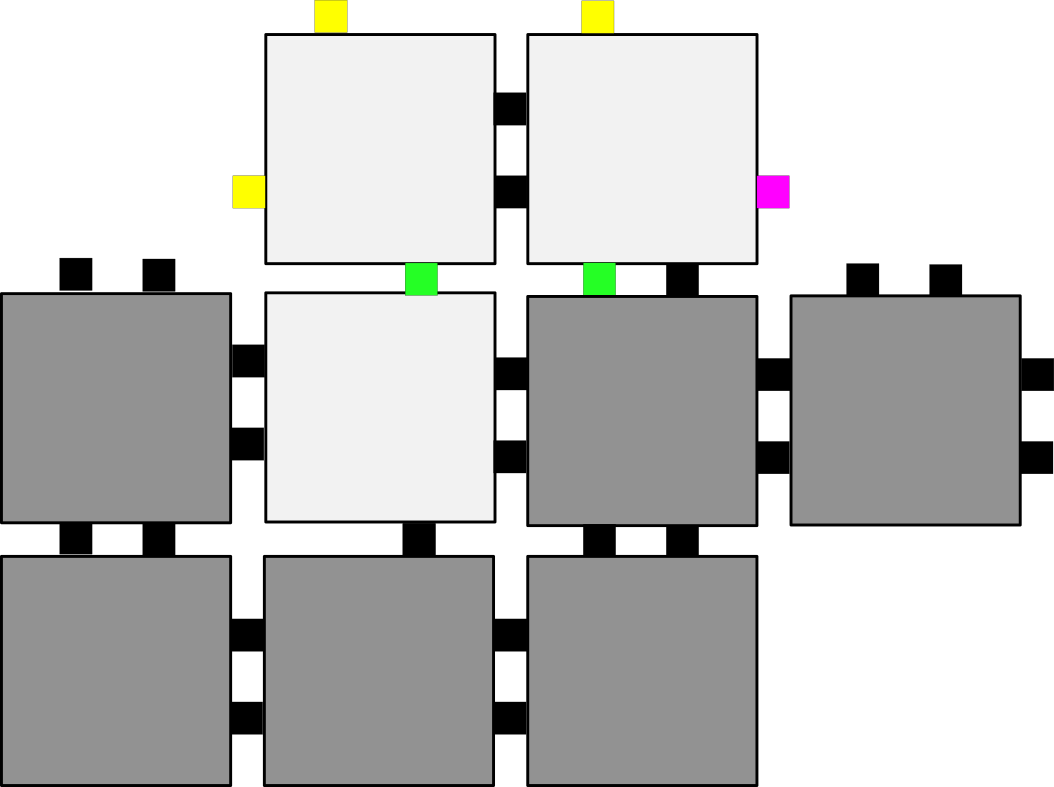}
        \label{fig:rep-build-ex2}
    }%
    \quad%
    \subfloat[East crawler tile (Figure~\ref{fig:rep1d_crawl} rotated $90^\circ$ CW) added via cooperative binding to prior exposed $w$ glue (middle left green) from duple tile and frame $y$ glue (bottom left green). $w$ glue (fuchsia) on tile follows frame, and $w^\prime$ glue (yellow) exposed to detect incoming $w$ messages.]{%
        \includegraphics[width=1.75in]{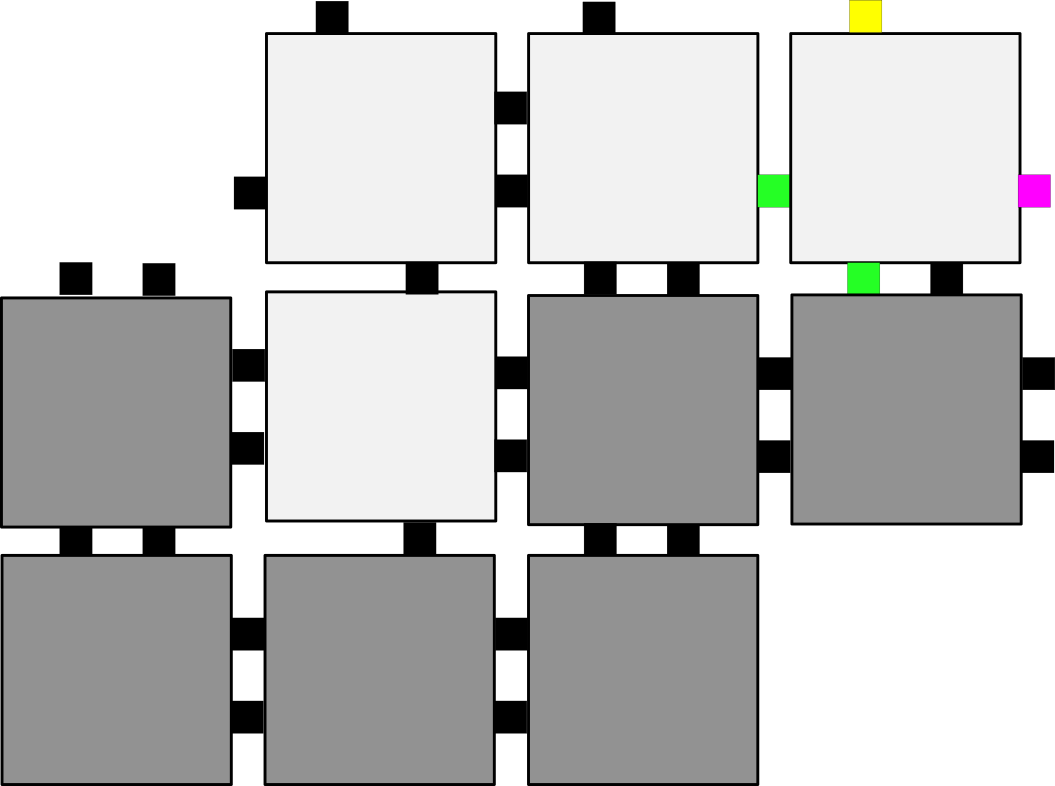}
        \label{fig:rep-build-ex3}
    }%
    \caption{Demonstration of first three tiles added to frame in replication of shape shown in Figure~\ref{fig:w1-example}\vspace{-20pt}}%
    \label{fig:rep-build-ex}%
\end{figure}

We now discuss the process by which a tile can sense when it is adjacent to a corner of $P$ and must redirect the $w$ message. Due to the cooperative nature of crawler or overhang duple tiles binding, each crawler or overhang tile can only be added if it is adjacent to both the exposed $w$ glue of the expanding ring and an adjacent $y$ glue of the frame. \update{We say that the $w$ message \emph{bounces} off the frame when the first $w$ glue activated on a ring tile is adjacent to $P$ and is redirected CCW by a set of activation and deactivation signals. For purposes of outlining the process by which a signal $w$ bounces, we define the E, N, and W $y^\prime$ glues on north crawler tiles as $y^\prime_0, \: y^\prime_1, \text{ and } y^\prime_2$. After $y^\prime_1$ is activated and binds to the frame, it sends a signal to activate the $w$ glue on the W side of the crawler tile. We say that a signal \emph{double bounces} if the same process occurs where both $w$ and $y^\prime_1$ bind to $P$. In a double bounce, instead of a $w$ glue being exposed to an edge of the crawler tile adjacent to $P$, the $w$ message is passed backwards to the preceding ring tile.} 
The process of a double bounce is illustrated in Figure \ref{fig:rep_bounce}. A double bounce is the maximum which can occur for a crawler tile. There exists the possibility for up to four bounces to occur for a duple; this is equivalent to two double bounces in succession.

\begin{figure}
    \centering
    \includegraphics[width=.9\textwidth]{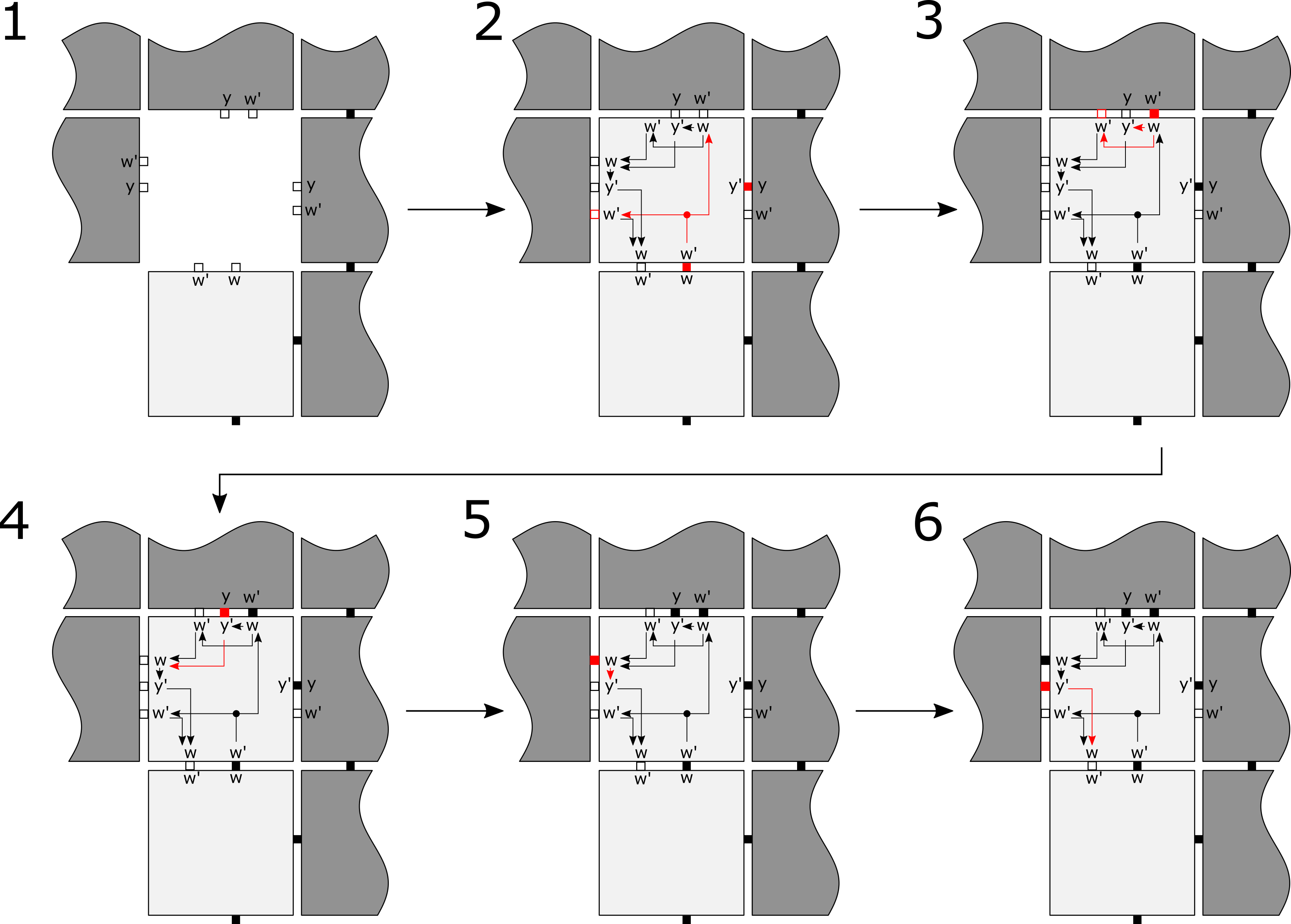}
    \caption{Step by step illustration of the binding of a north crawler tile and the double bouncing of the $w$ message, beginning with exposed glues which allow for cooperative binding. Dark tiles are cutaways of frame tiles. Glues with empty insides are active but not bound. Red signals and glues indicate activation and binding, respectively. Detachment signals are not shown for clarity. \vspace{-10pt}}
    \label{fig:rep_bounce}
\end{figure}

Width-1 shapes cause the need for this tile set to handle the $w$ message being propagated backwards (binding a $w$ glue to a tile which is already a member of the ring) as well as forwards (directing the placement of a ring tile). This requires a $w^\prime$ glue to be active on any edge which has an active $w$ glue. The first case by which messages must be passed backwards is initiated by the process of double bouncing, described in the prior paragraph. There exists a second possible scenario for a message being passed backwards, when a width-1 tile is connected to a set of width-2 tiles. It is caused by either a crawler or overhang duple tile encountering the west edge of another crawler tile or overhang duple, binding a $w$ glue of one to the $w^\prime$ of another. When this occurs, the message must be passed backwards to the prior tile via the activation of a $w$ glue on the same side as the initial $w$ glue which caused the tile to bind to the frame and expanding ring structure. An example of this is demonstrated in Figure~\ref{fig:w1-example} by the bottom left tile in the $3 \times 3$ square, and the first overhang duple.

The ring construction process completes when the $w$ message reaches the west face of the start tile, which can be caused by three unique processes. First, the case where a crawler or duple tile binds to the exposed $w^\prime$ glue on the west side of the start tile. This is how all non width-1 shapes will complete the ring, and some width-1 shapes as well. The second case is when a message is passed backwards to the start tile through the start tile's north face, and the west side of the start tile is adjacent to a frame tile. The $w$ message `bouncing' off the west side via the binding of the $y^\prime$ glue indicates a completed ring. In the final case, the west side of the start tile is available for growth of an additional crawler tile after the $w$ message has been returned to the start tile's north face. When the $w^\prime$ glue on the west side of the start tile is bound, similar to the first case, the ring has been completed. 

\subsubsection{Ring Detachment and Filling}\label{sec:ring-det}

Additional signals that cause the following behavior must be added to allow for replication: 1) the ring separates from the frame upon completion with no active glues which could interfere with other concurrent replication, 2) the ring exposes $x$ glues along its border like the original shape to enable the production of additional frames, and 3) the ring fills any holes internally by adding in fill tiles.

All three of these activities are initialized by the following of the $w$ messages through the completed ring beginning at the start tile. 
We utilize the gadgetry presented in Section \ref{sec:signal-following} to generate sets of glues which follow the $w$ message, deactivate the $y^\prime$ glues, 
and expose $x$ glues along $PERIM(P)$ to allow for construction of new frames. These signals and glues are then layered on top of the tile set shown in Figure \ref{fig:rep_tiles}.
As all the exterior-facing $w$ glues have been sent deactivation signals, the only glues which are guaranteed to be holding the ring to the frame are the active $y^\prime$ glues. 
The $y^\prime$ glues must be deactivated to enable the separation of the ring from the current frame. For purposes of detaching the completed ring from the frame, 
deactivation could be done na\"ively by sending \texttt{off} signals to all $y^\prime$ glues on all ring tiles. 
However, $x$ glues need to be activated only on the perimeter edges of the ring in order to prevent spurious frame formation. 
This is accomplished by adding signals which expose glues along the initial tile binding edge; the signals which cause separation are shown in Figures \ref{fig:rep1d_start-det}, \ref{fig:rep1d_crawl-det}, and \ref{fig:rep1d_duple-det}. The start tile (Figure \ref{fig:rep1d_start-det}) initiates a $w_f$ message which follows the $w$ message, such that $w_f$ glues are activated only on ring tile edges which are bound to successor ring tiles (i.e., ring tiles which the current ring tile's $w$ glues are bound). The $w_f$ message is initiated by the return of the $w$ message to the start tile, indicating completion of the ring.

\begin{figure}
\centering
  \subfloat[][The signals on the start tile which initiate detachment.]{%
        \label{fig:rep1d_start-det}%
	        \includegraphics[width=0.9in]{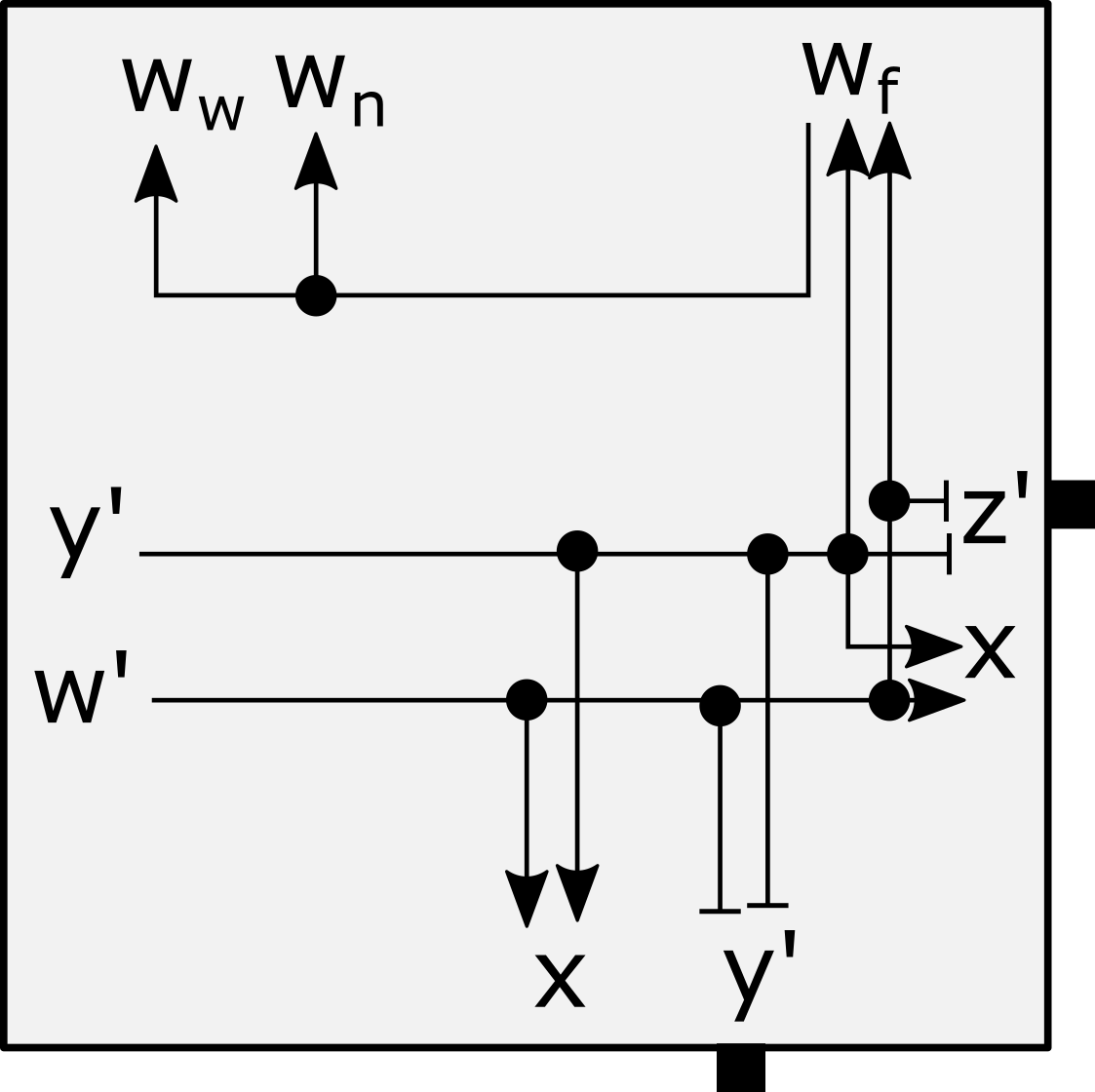}
        }%
        \quad\quad
  \subfloat[][A detachment signals for the north crawler tile.]{%
        \label{fig:rep1d_crawl-det}%
        		\includegraphics[width=0.9in]{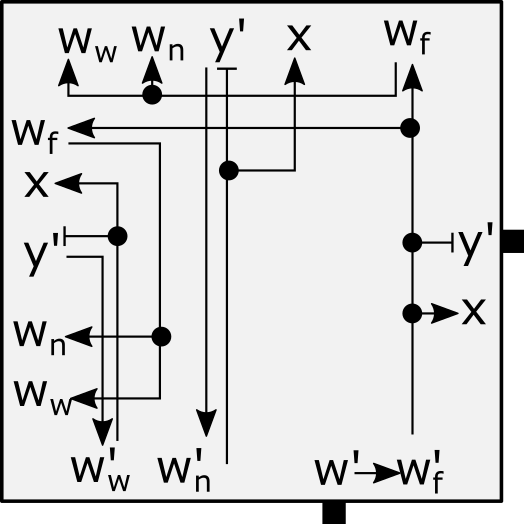}
        }%
        \quad\quad%
  \subfloat[][The detachment signals for the ``duple'' tiles]{%
        \label{fig:rep1d_duple-det}%
        		\includegraphics[width=1.9in]{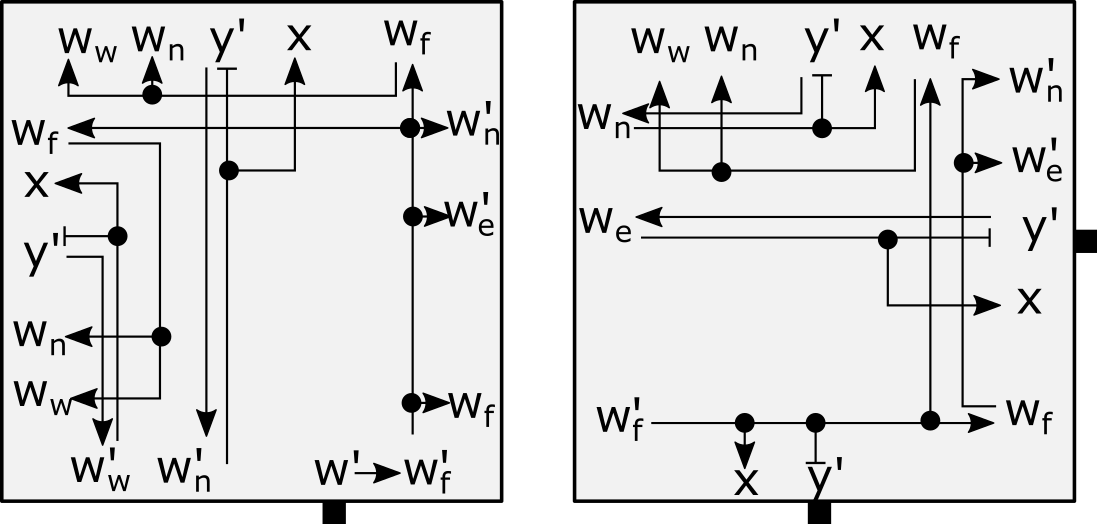}
        }%
  \caption{Addition of signals which enable detachment of ring from frame\vspace{-10pt}}
  \label{fig:rep_tiles-detach}
\end{figure}

The initial step of creating the ring does not allow the placement of any interior tiles. In order for the current assembly to have the same domain as the input shape $\alpha$,
we must place the remaining interior tiles. A north-eastern corner exists which can host cooperative binding of a filler tile demonstrated in Figure~\ref{fig:rep_fill}, 
as the minimal shape which requires a fill tile is a $3 \times 3$ square. With the binding of this initial tile, cooperative binding locations are added. 
This process continues until fill tiles are geometrically inhibited by ring tiles. 
These cooperative sites can be accessed using $w$ glues on crawler tiles which comprise the ring and were not activated in the initial ring formation. 
We activate these glues by sending messages which follow the $w_f$ glue activation. 

\begin{figure}
    \centering
	\includegraphics[width=0.4\textwidth]{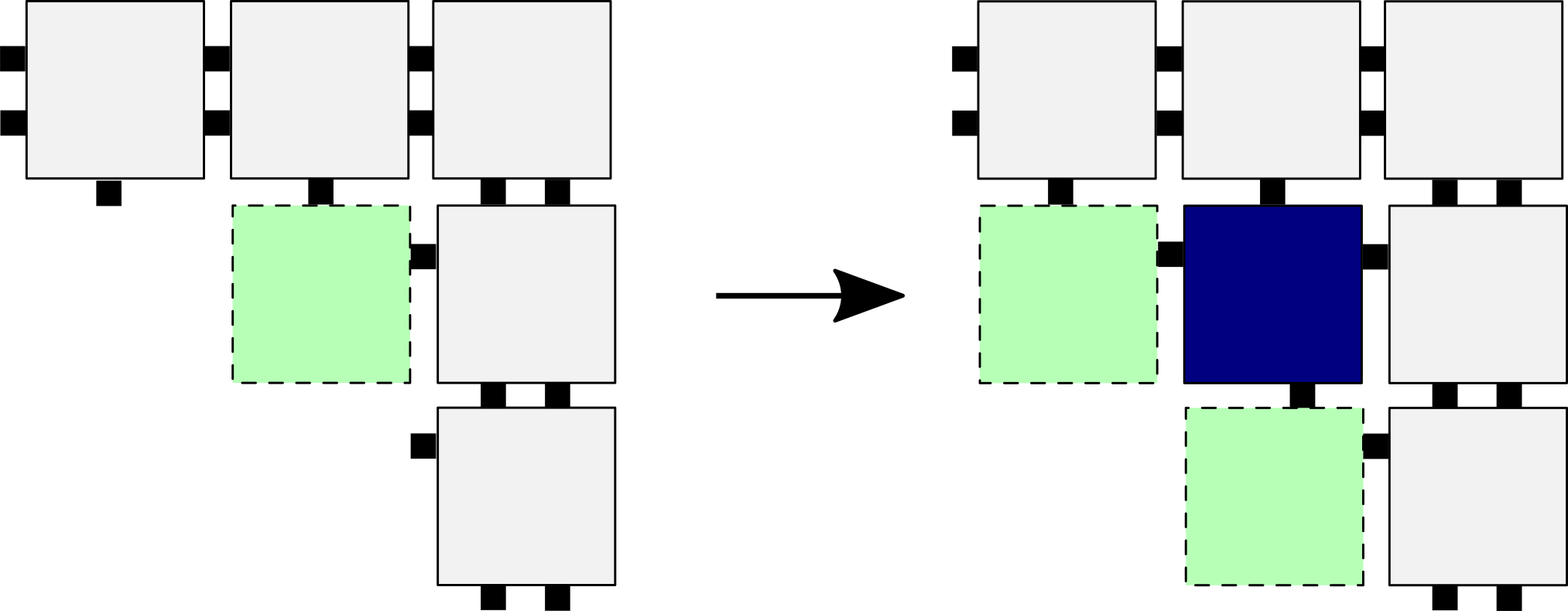}%
       
    \caption{Example dynamics of filling in ring structure.
    Green boxes indicated possible locations for cooperative binding of filler-type tiles. After initial filler tile binds, it opens up additional locations for binding of filler tiles.\vspace{-10pt}}
    \label{fig:rep_fill}
    
\end{figure}

\subsection{Correctness of exponential replication tileset}
We now demonstrate the correctness of the construction by validating each of the requirements of exponential replication.
\begin{theorem}\label{thm:exp_rep}
For any finite set of hole-free, shapes $P$ $=$ $\{P_1, P_2, ...,$ $ P_n\}$, there exists an STAM system $\calT = (T, \sigma_P, 2)$, with $\sigma_P$ a set of $n$ assemblies, one of each shape in $P$, and $\calT$ exponentially replicates without bound all shapes in $P$.
\end{theorem}

\begin{proof}

We first demonstrate that $\calT$ infinitely replicates. 1) Theorem~\ref{thm:complete-frame} demonstrates the ability of a constant sized tileset to generate a frame which defines the perimeter of each shape.
Lemma~\ref{lem:no-combine} demonstrates that each frame created is unique, and is unable to combine with any other. 
2) The only terminal assemblies are the rectangular layers of frames
3) At the end of the generation of an assembly of some shape $P_i$, the only remaining glues exposed to the exterior are strength 1 $x$ glues.
 These are never deactivated and will not be consumed. Given an infinite supply of $T$ as defined, for all shapes $P_i$ an infinite number can be replicated.
 
Second, the number of attachments required to generate an assembly of shape $P_i$ is bound by $poly(|P_i|)$. 
As demonstrated in Section~\ref{sec:ring-form} each singleton is bound to a single predecessor, and at most two successors which are placed by its $w$ glue by the passing of the $w$ message in ring formation. A `duple' assembly is bound to a single predecessor, and has at most three successors. Each fill tile is required to be bound to two tiles, and may be bound to two more tiles.

Third, every (super)tile attachment involves a singleton tile or `duple' assembly. The passing of the $w$ message in ring formation allows for the placement of one tile or duple at a time (beginning with the placement of the start tile), and a singleton filler tile binds to an available location. Deactivation of all exterior glues not of type $x$ on an arbitrary replicated assembly $P_i$ as demonstrated in Section~\ref{sec:ring-det} prevents any attachment to either an existing frame or any other shape at any stage of replication.

Finally, If $n$ copies of any $P_i$ exist at time $t$, then at time $t+1$ there are $2n$. Lemma~\ref{lem:no-combine} indicates each frame created from a shape $P_i$ can only create a copy of $P_i$, and since ring generation tiles can only create a shape outlined by a frame we will end up with $2n$ shapes given $n$ shapes of type $P_i$, as there exist $n$ frames to carry out replication of $P_i$.
 \qed
\end{proof}
\vspace{-13pt}
\bibliographystyle{amsplain}
\bibliography{tam,experimental_refs}

\providecommand{\bysame}{\leavevmode\hbox to3em{\hrulefill}\thinspace}
\providecommand{\MR}{\relax\ifhmode\unskip\space\fi MR }
\providecommand{\MRhref}[2]{%
  \href{http://www.ams.org/mathscinet-getitem?mr=#1}{#2}
}
\providecommand{\href}[2]{#2}
\begin{thebibliography}{10}

\bibitem{RNaseSODA2010}
Zachary Abel, Nadia Benbernou, Mirela Damian, Erik Demaine, Martin Demaine,
  Robin Flatland, Scott Kominers, and Robert Schweller, \emph{Shape replication
  through self-assembly and {RN}ase enzymes}, SODA 2010: Proceedings of the
  Twenty-first Annual ACM-SIAM Symposium on Discrete Algorithms (Austin,
  Texas), Society for Industrial and Applied Mathematics, 2010.

\bibitem{SelfReplicationDNA}
Andrew Alseth, Daniel Hader, and Matthew~J. Patitz, \emph{{Self-Replication via
  Tile Self-Assembly (Extended Abstract)}}, 27th International Conference on
  DNA Computing and Molecular Programming (DNA 27) (Dagstuhl, Germany)
  (Matthew~R. Lakin and Petr \v{S}ulc, eds.), Leibniz International Proceedings
  in Informatics (LIPIcs), vol. 205, Schloss Dagstuhl -- Leibniz-Zentrum
  f{\"u}r Informatik, 2021, pp.~3:1--3:22.

\bibitem{BarSchRotWin09}
Robert~D. Barish, Rebecca Schulman, Paul~W. Rothemund, and Erik Winfree,
  \emph{An information-bearing seed for nucleating algorithmic self-assembly},
  Proceedings of the National Academy of Sciences \textbf{106} (2009), no.~15,
  6054--6059.

\bibitem{chalkUniversalShapeReplicators2017}
Cameron Chalk, Erik~D. Demaine, Martin~L. Demaine, Eric Martinez, Robert
  Schweller, Luis Vega, and Tim Wylie, \emph{Universal shape replicators via
  {{Self-Assembly}} with {{Attractive}} and {{Repulsive Forces}}}, Proceedings
  of the {{Twenty-Eighth Annual ACM-SIAM Symposium}} on {{Discrete
  Algorithms}}, {Society for Industrial and Applied Mathematics}, January 2017,
  pp.~225--238.

\bibitem{AGKS05g}
Qi~Cheng, Gagan Aggarwal, Michael~H. Goldwasser, Ming-Yang Kao, Robert~T.
  Schweller, and Pablo~Moisset de~Espan\'{e}s, \emph{Complexities for
  generalized models of self-assembly}, SIAM Journal on Computing \textbf{34}
  (2005), 1493--1515.

\bibitem{DDFIRSS07}
Erik~D. Demaine, Martin~L. Demaine, S{\'a}ndor~P. Fekete, Mashhood Ishaque,
  Eynat Rafalin, Robert~T. Schweller, and Diane~L. Souvaine, \emph{Staged
  self-assembly: nanomanufacture of arbitrary shapes with ${O}(1)$ glues},
  Natural Computing \textbf{7} (2008), no.~3, 347--370.

\bibitem{IUSA}
David Doty, Jack~H. Lutz, Matthew~J. Patitz, Robert~T. Schweller, Scott~M.
  Summers, and Damien Woods, \emph{The tile assembly model is intrinsically
  universal}, Proceedings of the 53rd Annual IEEE Symposium on Foundations of
  Computer Science, FOCS 2012, 2012, pp.~302--310.

\bibitem{evans2014crystals}
Constantine~Glen Evans, \emph{Crystals that count! {P}hysical principles and
  experimental investigations of {D}{N}{A} tile self-assembly}, Ph.D. thesis,
  California Institute of Technology, 2014.

\bibitem{jSignals3D}
Tyler Fochtman, Jacob Hendricks, Jennifer~E. Padilla, Matthew~J. Patitz, and
  Trent~A. Rogers, \emph{Signal transmission across tile assemblies: 3d static
  tiles simulate active self-assembly by 2d signal-passing tiles}, Natural
  Computing \textbf{14} (2015), no.~2, 251--264.

\bibitem{STAMshapes}
Jacob Hendricks, Matthew~J. Patitz, and Trent~A. Rogers, \emph{Replication of
  arbitrary hole-free shapes via self-assembly with signal-passing tiles},
  Unconventional Computation and Natural Computation - 14th International
  Conference, {UCNC} 2015, Auckland, New Zealand, August 30 - September 3,
  2015, Proceedings (Cristian~S. Calude and Michael~J. Dinneen, eds.), Lecture
  Notes in Computer Science, vol. 9252, Springer, 2015, pp.~202--214.

\bibitem{STAMPatternRep}
Alexandra Keenan, Robert Schweller, and Xingsi Zhong, \emph{Exponential
  replication of patterns in the signal tile assembly model}, Natural Computing
  \textbf{14} (2014), no.~2, 265--278.

\bibitem{SignalsReplication}
Alexandra Keenan, Robert~T. Schweller, and Xingsi Zhong, \emph{Exponential
  replication of patterns in the signal tile assembly model}, DNA (David
  Soloveichik and Bernard Yurke, eds.), Lecture Notes in Computer Science, vol.
  8141, Springer, 2013, pp.~118--132.

\bibitem{jCCSA}
James~I. Lathrop, Jack~H. Lutz, Matthew~J. Patitz, and Scott~M. Summers,
  \emph{Computability and complexity in self-assembly}, Theory Comput. Syst.
  \textbf{48} (2011), no.~3, 617--647.

\bibitem{jSignals}
Jennifer~E. Padilla, Matthew~J. Patitz, Robert~T. Schweller, Nadrian~C. Seeman,
  Scott~M. Summers, and Xingsi Zhong, \emph{Asynchronous signal passing for
  tile self-assembly: Fuel efficient computation and efficient assembly of
  shapes}, International Journal of Foundations of Computer Science \textbf{25}
  (2014), no.~4, 459--488.

\bibitem{SignalTilesExperimental}
Jennifer~E. Padilla, Ruojie Sha, Martin Kristiansen, Junghuei Chen, Natasha
  Jonoska, and Nadrian~C. Seeman, \emph{A signal-passing {DNA}-strand-exchange
  mechanism for active self-assembly of {DNA} nanostructures}, Angewandte
  Chemie International Edition \textbf{54} (2015), no.~20, 5939--5942.

\bibitem{jSADS}
Matthew~J. Patitz and Scott~M. Summers, \emph{Self-assembly of decidable sets},
  Natural Computing \textbf{10} (2011), no.~2, 853--877.

\bibitem{ShapeIdentAlgo}
\bysame, \emph{Identifying shapes using self-assembly.}, Algorithmica
  \textbf{64} (2012), no.~3, 481--510.

\bibitem{RotWin00}
Paul W.~K. Rothemund and Erik Winfree, \emph{The program-size complexity of
  self-assembled squares (extended abstract)}, STOC '00: Proceedings of the
  thirty-second annual ACM Symposium on Theory of Computing (Portland, Oregon,
  United States), ACM, 2000, pp.~459--468.

\bibitem{SchulYurWinfEvolution}
Rebecca Schulman, Bernard Yurke, and Erik Winfree, \emph{Robust
  self-replication of combinatorial information via crystal growth and
  scission}, Proc Natl Acad Sci U S A \textbf{109} (2012), no.~17, 6405--10.

\bibitem{SolWin07}
David Soloveichik and Erik Winfree, \emph{Complexity of self-assembled shapes},
  SIAM Journal on Computing \textbf{36} (2007), no.~6, 1544--1569.

\bibitem{Winf98}
Erik Winfree, \emph{Algorithmic self-assembly of {D}{N}{A}}, Ph.D. thesis,
  California Institute of Technology, June 1998.

\bibitem{drmaurdsa}
Damien Woods, David Doty, Cameron Myhrvold, Joy Hui, Felix Zhou, Peng Yin, and
  Erik Winfree, \emph{Diverse and robust molecular algorithms using
  reprogrammable {D}{N}{A} self-assembly}, Nature \textbf{567} (2019),
  366--372.

\end{thebibliography}

\end{document}